\documentclass[11pt]{article}

\usepackage[a4paper,margin=1in]{geometry}
\usepackage{amsmath,amssymb,amsthm,mathtools}
\usepackage{braket}
\usepackage{array,booktabs}
\usepackage{enumitem}
\usepackage{float}
\usepackage{graphicx}
\usepackage{microtype}
\usepackage[table]{xcolor}
\usepackage[affil-it]{authblk}
\usepackage{tikz}
\usetikzlibrary{arrows.meta,decorations.pathreplacing,decorations.markings,calc}
\usepackage[hidelinks]{hyperref}

\newtheorem{theorem}{Theorem}[section]
\newtheorem{proposition}[theorem]{Proposition}
\newtheorem{corollary}[theorem]{Corollary}
\newtheorem{lemma}[theorem]{Lemma}
\theoremstyle{definition}
\newtheorem{definition}[theorem]{Definition}
\newtheorem{assumption}[theorem]{Assumption}
\theoremstyle{remark}
\newtheorem{remark}[theorem]{Remark}
\newtheorem{example}[theorem]{Example}

\usepackage[nameinlink,capitalize,noabbrev]{cleveref}
\usepackage{etoolbox}
\crefname{assumption}{Assumption}{Assumptions}
\Crefname{assumption}{Assumption}{Assumptions}
\crefname{problem}{Problem}{Problems}
\Crefname{problem}{Problem}{Problems}
\AtBeginEnvironment{proposition}{\crefalias{theorem}{proposition}}
\AtBeginEnvironment{corollary}{\crefalias{theorem}{corollary}}
\AtBeginEnvironment{lemma}{\crefalias{theorem}{lemma}}
\AtBeginEnvironment{definition}{\crefalias{theorem}{definition}}
\AtBeginEnvironment{assumption}{\crefalias{theorem}{assumption}}
\AtBeginEnvironment{remark}{\crefalias{theorem}{remark}}
\AtBeginEnvironment{example}{\crefalias{theorem}{example}}
\AtBeginEnvironment{problem}{\crefalias{theorem}{problem}}

\definecolor{addition}{HTML}{000000}  
\definecolor{notation}{HTML}{000000}  
\newcommand{\chg}[1]{\textcolor{notation}{#1}}
\newcommand{\add}[1]{\textcolor{addition}{#1}}
\DeclareMathOperator{\Cl}{Cl}
\DeclareMathOperator{\im}{im}
\DeclareMathOperator{\rank}{rank}
\DeclareMathOperator{\Spec}{Spec}
\DeclareMathOperator{\pdet}{pdet}
\DeclareMathOperator{\Tor}{Tor}
\DeclareMathOperator{\Tr}{Tr}
\newcommand{\Z}{\mathbb{Z}}
\newcommand{\Q}{\mathbb{Q}}
\newcommand{\R}{\mathbb{R}}
\newcommand{\F}{\mathbb{F}}
\newcommand{\wtH}{\widetilde H}
\newcommand{\wtbeta}{\widetilde\beta}
\newcommand{\up}{\mathrm{up}}
\newcommand{\down}{\mathrm{down}}

\newcolumntype{L}[1]{>{\raggedright\arraybackslash}p{#1}}

\title{Quantum Log-Determinant Methods for Torsion-Sensitive\\
Topological Data Analysis}
\author[1]{Dimitrios Thanos\thanks{Corresponding author: \texttt{dimitrios.thanos.0@gmail.com}}}
\author[2,3]{Caesnan Leditto}

\author[4,5]{Adam Wesolowski}
\author[6,7]{Andreea-Iulia Lefterovici}
\author[1]{Mahtab Yaghubi Rad}
\affil[1]{Institute of Advanced Computer Science, Leiden University, The Netherlands}

\affil[2]{School of Physics and Astronomy, Monash University, Melbourne, Australia}
\affil[3]{Faculty of Science and Computer, Universitas Kristen Immanuel, Yogyakarta, Indonesia}

\affil[4]{Department of Computer Science, University of Oxford, UK}
\affil[5]{Department of Computer Science, Royal Holloway University of London, UK}
\affil[6]{Institut f\"ur Theoretische Physik, Leibniz Universit\"at Hannover, Hannover, Germany}
\affil[7]{Matematikos ir informatikos fakultetas, Vilniaus universiteto, Vilnius, Lithuania}
\date{}

\begin{document}

\maketitle

\begin{abstract}
Estimating Betti numbers is a primary objective in topological data analysis,
where they are used to describe and distinguish topological spaces. However,
Betti numbers fail to capture torsion in integral homology. Torsion can distinguish topological structures with identical Betti
numbers, and examples from classical topological data analysis show that it can
affect the topology inferred from data. Related torsion-sensitive quantities are
also known to have practical applications in graph-based machine learning,
including biomedical applications. We investigate how quantum spectral
methods can access this additional information.

Log pseudodeterminants, obtained by summing the logarithms of the nonzero
boundary-Laplacian eigenvalues, summarize spectral information about a simplicial
complex.  Under explicit topological assumptions, an alternating combination of
these quantities gives the size of a higher critical group, the higher-dimensional
analogue of a graph sandpile group.  With stronger certification, a related
construction recovers the size of the torsion subgroup in homology itself. This
is, to our knowledge, the first provably correct quantum estimator of
integral-homology torsion order.

We give a new quantum algorithm for estimating these quantities, with explicit
dependence on how the input is queried, the required accuracy, and the spectral
gaps. Exact spectra for balanced multipartite clique complexes provide favorable
gaps, and a related construction gives an explicit torsion-bearing benchmark. With
compact query access, the algorithm need not construct the exponentially large
boundary matrices explicitly.  For a restricted torsion-bearing family, a tailored
amplitude-estimation routine uses quadratically fewer queries than any randomized
classical algorithm given comparable access when estimating the log pseudodeterminant.  Estimating
the full, unnormalized quantity to fixed additive error still costs in
proportion to the candidate-space dimension $D$.
\end{abstract}

\section{Introduction}
\label{sec:introduction}
Topological data analysis (TDA) represents a dataset by a simplicial complex, whose
vertices, edges, triangles, and higher-dimensional faces encode relations in the
data.  In most applications, TDA uses a filtration: a nested family
$X_s\subseteq X_t$ for $s\leq t$ that records how the complex changes with scale.
Homology distinguishes cycles that do not bound higher-dimensional chains, and
persistent homology tracks when such classes appear and disappear along the
filtration.  Quantum TDA (QTDA) follows the same principle but accesses homology
spectrally: the $j$th Betti number is the dimension of the null space of a
combinatorial Laplacian~\cite{lloyd2015quantumalgorithmstopologicalgeometric,Hayakawa_2022,Scali_2024}.
The $j$th Betti number counts independent $j$-dimensional homology classes;
familiar examples are connected components, loops, and enclosed voids in
dimensions zero, one, and two, respectively.

The null space of a Hodge Laplacian, however, sees only homology over a field of
characteristic zero.  For a finite simplicial complex $X$,
\begin{equation}
    H_j(X;\Z)
    \cong
    \Z^{\beta_j}\oplus \Tor H_j(X;\Z),
    \label{eq:integral-homology-decomposition}
\end{equation}
Here $\beta_j$ is the rank, or Betti number, and $\Tor H_j(X;\Z)$ is the finite
subgroup of classes annihilated by some nonzero integer; its cardinality is the
torsion order.  Tensoring with $\R$ annihilates the finite summand.  The real Hodge
kernel therefore recovers $\beta_j$, but two complexes with the same real homology
can still differ because one has nontrivial integral torsion and the other does not.

This distinction already appears in classical TDA. Carlsson et al.~\cite{CarlssonEtAl2008} used persistent homology with \(\mathbb Z_2\) coefficients, together with an explicit parametrization of the patches by linear and quadratic image gradients, to identify a Klein-bottle structure in high-contrast \(3\times3\) natural-image patches. More generally, coefficient-field dependence along a filtration is controlled by torsion in relative homology~\cite{ObayashiYoshiwaki2023}. Torsion-based quantities have also been used in practical settings, including graph neural networks and circRNA--drug association prediction~\cite{ShenEtAl2025TorGNN,torsion_in_med}. Together, these results motivate extending quantum Betti-number estimation to torsion-sensitive invariants.

A recent preprint proposed detecting torsion quantumly by equating the dimension of
homology modulo a prime with the dimension of the kernel of a corresponding
Laplacian~\cite{nghiem2025torsion}. This relation does not hold in general, so its
stated correctness guarantee does not follow. Our method avoids this issue and,
under explicit assumptions, provides, to our knowledge, the first rigorous quantum
estimator of torsion order in integral homology.

The nonzero Laplacian spectrum contains arithmetic information that the kernel
discards.  For a positive semidefinite matrix $A$, $\pdet(A)$ denotes the product
of its positive eigenvalues, counted with multiplicity; unlike the determinant, it
ignores zero modes.  For the signed boundary matrix
$B_r:C_r\to C_{r-1}$, define
$$\ell_r(X)=\log\pdet(B_rB_r^{\mathsf T}).$$  Matrix-tree theorems connect these products
to higher critical groups, the higher-dimensional analogues of graph sandpile
groups, and to torsion-weighted spanning-tree sums.  The critical-group order is
torsion-sensitive, but it is a global tree aggregate rather than, in general, the
homology-torsion order of $X$ itself.

The paper is organized around one thesis: \emph{a bottom-left-QSVT amplitude-estimation algorithm
estimates boundary log pseudodeterminants, and matrix-tree identities determine
when alternating combinations of those estimates have a torsion-sensitive
topological interpretation}.  The logical dependence is summarized in
\cref{tab:invariant-map}.  The lower-homology correction terms in the table are
algebraic inputs; the quantum routine estimates the spectral terms $\ell_r$.
Here a pure $d$-dimensional complex is \emph{acyclic in positive codimension}
(APC) if its rational reduced homology vanishes below dimension $d$; see
\cref{def:apc}.
For such a complex, write $t_q(X)=|\wtH_q(X;\Z)|$.  For $0\leq k\leq d$, the
corrected alternating expression used below is
\begin{equation}
  \underbrace{\sum_{r=0}^k(-1)^{k-r}\ell_r(X)}_
    {\text{estimated by the quantum routine}}
  +
  \underbrace{2\sum_{q=0}^{k-2}(-1)^{k-q-2}\log t_q(X)}_
    {\text{separately supplied homology correction}}.
  \label{eq:intro-alternating-combination}
\end{equation}

\begin{table}[!b]
\centering
\small
\begin{tabular}{@{}L{0.27\linewidth}L{0.31\linewidth}L{0.34\linewidth}@{}}
\toprule
Available hypotheses & What to do & Guaranteed interpretation \\
\midrule
none beyond finiteness
  & use an individual $\ell_r(X)$
  & boundary spectral statistic; no group order is implied \\
\arrayrulecolor{black!15}\specialrule{0.3pt}{2pt}{2pt}\arrayrulecolor{black}
$X$ is APC
  & form the full corrected expression \eqref{eq:intro-alternating-combination}
  & $\log\tau_k(X)$, the logarithm of a torsion-weighted tree enumerator \\
\arrayrulecolor{black!15}\specialrule{0.3pt}{2pt}{2pt}\arrayrulecolor{black}
\emph{Critical-group branch:}
$X$ satisfies \cref{cor:root-free-critical}
  & set $k=i+1$ in the corrected expression
  & $\log|K_i(X)|$, the logarithm of a higher critical-group order \\
\arrayrulecolor{black!15}\specialrule{0.3pt}{2pt}{2pt}\arrayrulecolor{black}
\emph{Certified-tree branch:}
$T$ satisfies \cref{thm:certified-tree-torsion}
  & evaluate the corrected expression on $T$ and divide by $2$
  & $\log|\wtH_{k-1}(T;\Z)|$, the logarithm of the homology order \\
\bottomrule
\end{tabular}
\caption{How to interpret the spectral output.  Each row states the hypotheses
needed for its conclusion; the quantum routine supplies the $\ell_r$ terms, not
the lower-homology correction in \eqref{eq:intro-alternating-combination}.
The last two rows are alternative specializations of the APC row, not successive
steps. The tree enumerator $\tau_k$ is defined in \eqref{eq:tree-enumerator}.}
\label{tab:invariant-map}
\end{table}

Our main algorithm estimates the spectral terms underlying every row of
\cref{tab:invariant-map}.
Given coherent boundary access and a certified gap,
the algorithm uses the complementary, or bottom-left, block of a QSVT circuit
to encode logarithmic spectral weights into a success probability. The response
at zero is calibrated so that the desired sum over the nonzero spectrum can be
recovered from a known offset, without separately counting the zero singular
values. Amplitude estimation, followed by classical rescaling and subtraction
of this offset, then gives $\ell(B)/D$, i.e., the sum of the logarithms of the
nonzero squared singular values normalized by the input-space dimension $D$.
Taking corrected alternating sums of the resulting estimates implements the identities in
the remaining rows.
This route requires no oracle for a selected higher-dimensional tree. The
estimator naturally returns this normalized quantity, while estimating the full
$\ell(B)$ to fixed additive error incurs an additional factor $D$.
This pipeline is summarized in \cref{fig:quantum-pipeline}.

\begin{figure}[H]
\centering

\definecolor{accblue}{RGB}{31,78,140}
\definecolor{accred}{RGB}{160,44,44}
\definecolor{celldark}{RGB}{26,86,160}
\definecolor{celllight}{RGB}{176,203,232}
\definecolor{panelgray}{RGB}{240,241,244}
\definecolor{panelblue}{RGB}{233,240,250}
\definecolor{panelred}{RGB}{250,235,233}
\definecolor{subtxt}{RGB}{105,105,110}

\def\PW{2.75}
\def\PH{2.20}
\def\DX{3.18}
\def\Ytitle{-1.43}
\def\Ysub{-2.02}
\def\Ybrace{-2.79}
\def\Ycap{-3.08}

\begin{tikzpicture}[
  panel/.style={rounded corners=8pt, draw=none,
                minimum width=\PW cm, minimum height=\PH cm},
  flow/.style={-{Stealth[length=4pt,width=4pt]}, gray!60, line width=1.0pt},
  ttl/.style={font=\scriptsize\bfseries, align=center, text width=2.95cm},
  sub/.style={font=\scriptsize, color=subtxt, align=center,
              text width=2.85cm},
  outcome/.style={draw=accred, rounded corners=2pt, line width=0.7pt,
                  fill=white, text=accred, font=\scriptsize,
                  align=center, text width=1.58cm, minimum height=0.38cm,
                  inner sep=1pt},
]

\node[panel, fill=panelgray] (P1) at (0*\DX,0) {};
\node[panel, fill=panelblue] (P2) at (1*\DX,0) {};
\node[panel, fill=panelblue] (P3) at (2*\DX,0) {};
\node[panel, fill=panelblue] (P4) at (3*\DX,0) {};
\node[panel, fill=panelred ] (P5) at (4*\DX,0) {};
\foreach \i in {0,1,2,3}{%
  \draw[flow] ({\i*\DX+1.44},0) -- ({\i*\DX+1.74},0);}

\begin{scope}[shift={(0*\DX,0)}, scale=0.88]
  \coordinate (v1) at (-0.52, 0.76);
  \coordinate (v2) at (-1.02,-0.06);
  \coordinate (v3) at ( 0.05, 0.16);
  \coordinate (v4) at (-0.34,-0.76);
  \coordinate (v5) at ( 0.82, 0.69);
  \coordinate (v6) at ( 0.76,-0.61);
  \fill[black!22] (v1) -- (v2) -- (v3) -- cycle;
  \fill[black!22] (v3) -- (v5) -- (v6) -- cycle;
  \draw[black!65, line width=1.0pt]
     (v1) -- (v2) (v1) -- (v3) (v2) -- (v3)
     (v2) -- (v4) (v3) -- (v4) (v4) -- (v6)
     (v3) -- (v6) (v3) -- (v5) (v5) -- (v6);
  \foreach \v in {v1,v2,v3,v4,v5,v6}{\fill[black] (\v) circle (0.085);}
\end{scope}

\begin{scope}[shift={(1*\DX,0)}]
  \def\cs{0.160}
  \foreach \r in {1,...,6}{\foreach \c in {1,...,8}{
    \draw[gray!25, line width=0.2pt]
      ({(\c-5)*\cs},{(3-\r)*\cs}) rectangle
      ({(\c-4)*\cs},{(4-\r)*\cs});}}
  \foreach \r/\c in {1/1,1/2,2/2,2/3,3/3,3/4,4/4,4/5,5/5,5/6,6/6,6/7}{
    \fill[celldark] ({(\c-5)*\cs},{(3-\r)*\cs}) rectangle
                    ({(\c-4)*\cs},{(4-\r)*\cs});}
  \foreach \r/\c in {1/4,2/5,3/6,4/7,5/8,2/1,3/2,4/3}{
    \fill[celllight] ({(\c-5)*\cs},{(3-\r)*\cs}) rectangle
                     ({(\c-4)*\cs},{(4-\r)*\cs});}
  \draw[line width=0.8pt] (-0.70, 0.57) -- (-0.80, 0.57)
                          -- (-0.80,-0.57) -- (-0.70,-0.57);
  \draw[line width=0.8pt] ( 0.70, 0.57) -- ( 0.80, 0.57)
                          -- ( 0.80,-0.57) -- ( 0.70,-0.57);
  \node[font=\scriptsize, color=accblue] at (0,-0.84) {$\mathcal B_r$};
\end{scope}

\begin{scope}[shift={(2*\DX,0)}]
  \def\gam{0.39}
  \fill[gray!15] (-\gam,-0.67) rectangle (\gam,0.78);
  \draw[gray!55, line width=0.6pt] (-1.05,-0.48) -- (1.05,-0.48);
  \draw[gray!55, line width=0.6pt] (0,-0.60) -- (0,0.80);
  \draw[gray!55, dashed, line width=0.5pt] (-\gam,-0.60) -- (-\gam,0.70);
  \draw[gray!55, dashed, line width=0.5pt] ( \gam,-0.60) -- ( \gam,0.70);
  \draw[accblue, line width=1.45pt]
    (-1.00,-0.37) .. controls (-0.78,-0.25) and (-0.52,0.17) .. (-\gam,0.63);
  \draw[accblue, line width=1.45pt]
    ( 1.00,-0.37) .. controls ( 0.78,-0.25) and ( 0.52,0.17) .. ( \gam,0.63);
  \fill[accred] (0,0.10) circle (0.075);
  \node[font=\tiny, color=subtxt] at (-\gam,-0.78) {$-\gamma$};
  \node[font=\tiny, color=subtxt] at ( \gam,-0.78) {$\gamma$};
\end{scope}

\begin{scope}[shift={(3*\DX,0)}]
  \node[draw=accblue, rounded corners=2pt, line width=0.9pt,
        fill=accblue!8, text=accblue, font=\scriptsize\bfseries,
        minimum width=0.62cm, minimum height=0.50cm] (AE) at (-0.82,-0.12)
        {$\mathrm{AE}$};
  \node[draw=accblue, rounded corners=2pt, line width=0.75pt,
        fill=white, text=accblue, font=\scriptsize,
        minimum width=1.48cm, minimum height=0.50cm] (Ell) at (0.53,-0.12)
        {$\widehat{\ell_r(X)/D_r}$};
  \draw[-{Stealth[length=3pt]}, accblue, line width=0.75pt]
        (AE.east) -- (Ell.west);
\end{scope}

\begin{scope}[shift={(4*\DX,0)}]
  \coordinate (hub) at (-1.02,0);
  \fill[accred] (hub) circle (0.065);
  \node[outcome] (tree) at (0.22,0.68) {tree sum};
  \node[outcome] (group) at (0.22,0) {critical\\[-2pt]group};
  \node[outcome] (torsion) at (0.22,-0.68) {torsion\\[-2pt]order};
  \draw[-{Stealth[length=3pt]}, accred, line width=0.75pt] (hub) -- (tree.west);
  \draw[-{Stealth[length=3pt]}, accred, line width=0.75pt] (hub) -- (group.west);
  \draw[-{Stealth[length=3pt]}, accred, line width=0.75pt] (hub) -- (torsion.west);
\end{scope}

\node[ttl] at (0*\DX,\Ytitle) {Simplicial complex};
\node[ttl] at (1*\DX,\Ytitle) {Boundary map};
\node[ttl] at (2*\DX,\Ytitle) {Bottom-left QSVT};
\node[ttl] at (3*\DX,\Ytitle) {Amplitude\\estimation};
\node[ttl] at (4*\DX,\Ytitle) {Topological output};

\node[sub] at (0*\DX,\Ysub) {topological input};
\node[sub] at (1*\DX,\Ysub) {block encoding};
\node[sub] at (2*\DX,\Ysub) {logarithmic weights;\\calibrated zero mode};
\node[sub] at (3*\DX,-2.20)
      {recover the\\normalized log\\pseudodeterminant};
\node[sub] at (4*\DX,\Ysub) {matrix--tree identities};

\draw[accblue, line width=0.75pt, decorate,
      decoration={brace, mirror, amplitude=5pt, raise=1pt}]
      ({1*\DX-1.36},\Ybrace) -- ({3*\DX+1.36},\Ybrace);
\draw[accred, line width=0.75pt, decorate,
      decoration={brace, mirror, amplitude=5pt, raise=1pt}]
      ({4*\DX-1.36},\Ybrace) -- ({4*\DX+1.36},\Ybrace);

\node[anchor=north, align=center, color=accblue, font=\scriptsize,
      text width=8.4cm]
      at (2*\DX,\Ycap)
      {quantum output: an estimate of $\ell_r(X)/D_r$\\
       (the normalized log pseudodeterminant)};
\node[anchor=north, align=center, color=accred, font=\scriptsize,
      text width=2.75cm]
      at (4*\DX,\Ycap)
      {classical matrix--tree interpretation};

\end{tikzpicture}
\caption{Quantum-to-topological pipeline. For each degree $r$, the padded
boundary map $\mathcal B_r$ is block encoded, bottom-left QSVT encodes
calibrated logarithmic spectral weights into a success probability, and
amplitude estimation estimates that probability. Known rescaling and offset
subtraction then recover an estimate of $\ell_r(X)/D_r$ without a separate
rank estimate. Estimates from several degrees are combined classically using
the applicable matrix--tree identity and any required lower-homology
corrections. Under the stated hypotheses, this yields a torsion-weighted tree
enumerator, a critical-group order, or a certified homology-torsion order.}
\label{fig:quantum-pipeline}
\end{figure}
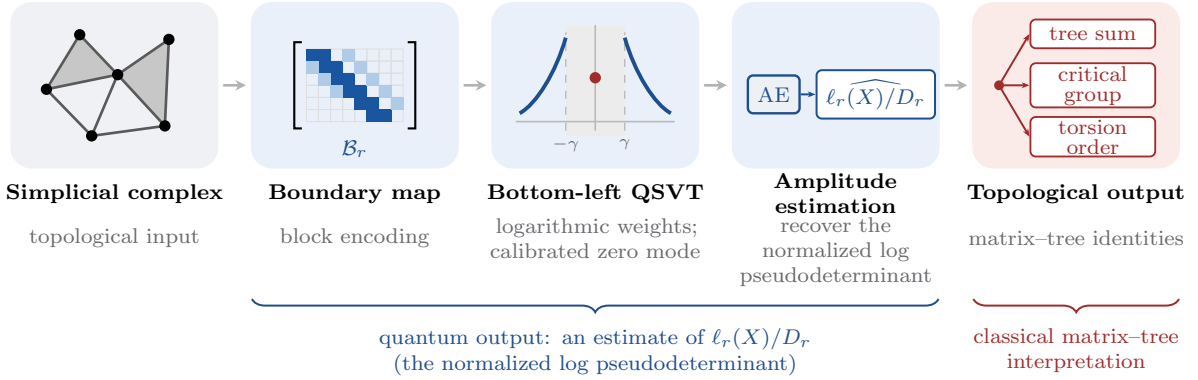

\subsection{Contributions}

The main contributions are:
\begin{enumerate}[leftmargin=2em]
  \item\label{contrib:log-pseudodeterminants}
  \textbf{Improved gap dependence for boundary log pseudodeterminants.}

  We give a bottom-left-QSVT amplitude-estimation algorithm for boundary
  log pseudodeterminants normalized by the size of the candidate simplex
  space, under explicit block-encoding, state-preparation, singular-gap,
  and normalized-sum-bound assumptions. For the clique encodings considered
  here, the QSVT polynomial degree is
  $\widetilde O(\gamma^{-1})$, compared with the
  $\widetilde O(\gamma^{-2})$ obtained by applying the logarithm
  approximation of Luongo and Shao
  \cite{luongo2024quantumalgorithmsspectralsums} directly to the normalized
  upper Hodge Laplacian.

  \item\label{contrib:group-orders}
  \textbf{Quantum estimation of critical-group and homology-torsion orders.}

  Leveraging the higher-dimensional matrix-tree identities of Duval,
  Klivans, and Martin
  \cite{DuvalKlivansMartin2011Cellular,DuvalKlivansMartin2011FPSAC},
  we combine boundary log pseudodeterminants to estimate tree enumerators,
  critical-group orders, and, for certified trees, homology-torsion orders.
  The necessary lower-homology correction orders must be supplied separately.

  \item\label{contrib:resource-comparison}
  \textbf{An explicit torsion-bearing benchmark and its resource scaling.}

  We provide a certified torsion-bearing family by joining balanced multipartite
  clique complexes, a specialization
  of the complete colorful complexes\footnote{Complete colorful complexes are
  clique complexes of complete multipartite graphs: each simplex contains at
  most one vertex from each color class. They are balanced when all classes
  have equal size.} studied by Duval, Klivans, and Martin
  \cite[Section~6.1]{DuvalKlivansMartin2011Cellular}, with a fixed flag
  triangulation of $\mathbb{RP}^2$.
  The resulting complexes are pure and APC, are their own unique
  top-dimensional spanning trees, have vanishing lower-homology corrections,
  and have explicit inverse-polynomial boundary gaps. Building on the
  all-degree Hodge spectra of complete colorful complexes
  \cite{DuvalKlivansMartin2011Cellular}, we derive exact boundary
  singular-value multiplicities, condition numbers, and log
  pseudodeterminants.

  When the family has $k$ parts of $2^k$ vertices each, our method estimates
  the log torsion order to relative error with polynomially many oracle
  queries and, under our implementation assumptions, quantum gates, whereas
  the explicit boundary matrix has $N^{\Theta(\log N)}$ entries. Following
  the resource comparison of Berry et al.\
  \cite[Section~IV~A]{BerryEtAl2024QTDA}, this measures the advantage of
  avoiding explicit matrix construction.
\end{enumerate}

%

\paragraph{Scope.}
Subject to the stated access, state-preparation, and spectral-gap assumptions,
the quantum routine estimates boundary log pseudodeterminants, which are
spectral sums over the nonzero singular values of the boundary maps. This
spectral output does not itself require a topological promise, that is a
separate step. As summarized in \cref{tab:invariant-map}, corrected alternating
combinations of these estimates can give a torsion-weighted tree count, a
critical-group order, or, when the input is a certified higher-dimensional
tree, the order of its finite homology group. For a general complex $X$,
however, the result should not be interpreted as
$\log|\Tor H_i(X;\Z)|$, the logarithm of the number of torsion elements in
its $i$-th homology group. Any required lower-homology correction orders must
be supplied separately. Moreover, an order does not determine the full
structure of a finite group as different groups can have the same order.

The rest of the paper is organized as follows.  \Cref{sec:related-work} places the
problem among QTDA, torsion-sensitive persistence, and spectral-sum algorithms.
\Cref{sec:preliminaries} fixes the augmented-chain, pseudodeterminant, and
critical-group conventions.  \Cref{sec:matrix-tree} derives the determinant
identities and the certified-tree specialization.  \Cref{sec:quantum-algorithm}
gives the bottom-left-QSVT estimator and its clique-complex implementation.
\Cref{sec:multipartite} develops the spectral and torsion-bearing benchmarks.~\Cref{sec:interpretation} states the interpretation and complexity limitations,
and~\cref{sec:conclusion} concludes.  The kernel-suppressing squared-amplitude
variant is given in \cref{app:squared-amplitude-estimator}; weighted-polynomial
and QSVT proofs, implementation-error analysis,
error-allocation and benchmark-resource details, and indexing-convention check
are deferred to the appendices.

\section{Related work}
\label{sec:related-work}

\subsection{Quantum topological data analysis}

{
The foundational QTDA proposal estimates Betti numbers by combining simplex-state
preparation with spectral access to a combinatorial Laplacian~\cite{lloyd2015quantumalgorithmstopologicalgeometric}.  Subsequent work has refined
the persistent setting, density-of-states estimation, qubit counts, input models, and
resource requirements~\cite{Hayakawa_2022,Scali_2024,McArdle2022AQubits}.  Analyses of prospective quantum
advantage emphasize that simplex density, state preparation, precision, and the
smallest nonzero Laplacian eigenvalue can dominate the running time~\cite{BerryEtAl2024QTDA}.}
The present work is subject to the same constraints but estimates a different
spectral statistic: a logarithmic sum over the nonzero spectrum rather than only
its nullity.

\subsection{Complexity of quantum Betti-number estimation}

The complexity landscape for Betti-number estimation has been studied
extensively, although these results do not directly characterize the
log-pseudodeterminant problem considered here.
Deciding whether a Betti number of a clique complex is nonzero is
$\mathsf{QMA}_1$-hard, and therefore so is any multiplicative approximation
that preserves zero~\cite{crichigno2024}; a gapped weighted version is
$\mathsf{QMA}_1$-hard and contained in $\mathsf{QMA}$~\cite{king2024}. Average-case hardness of Betti number estimation and related problems has been recently investigated in~\cite{strelchuk2026average}. 
Closely related low-lying spectral-density and quasi-Betti-number estimation
problems are $\mathsf{DQC}1$-hard at inverse-polynomial precision~\cite{cade2021,gyurik2022}, providing evidence that normalized spectral
estimation can be hard for classical computers. In the other direction,
Apers, Gribling, Sen, and Szab\'o~\cite{apers2023} gave a
classical path integral Monte Carlo algorithm that estimates the normalized
Betti number of a general complex in time
$n^{O(\gamma^{-1/2}\log(1/\varepsilon))}$, and of a clique complex in time
$(n/\lambda_{\max})^{O(\gamma^{-1/2}\log(1/\varepsilon))}\cdot\mathrm{poly}(n)$.
Here $n$ is the number of vertices, $k$ is the target simplex dimension,
$\gamma$ is the normalized gap, $\varepsilon$ is the additive precision, and
$\lambda_{\max}$ is the largest relevant combinatorial-Laplacian eigenvalue.
Their algorithm is efficient for constant $\gamma$ and $\varepsilon$; because
$\lambda_{\max}\geq k$, for clique complexes with $k\in\Omega(n)$ it also
remains efficient at inverse-polynomial precision when $\gamma$ is constant.
Betti-number hardness does not directly classify our spectral target: gapped
clique homology is $\mathsf{QMA}_1$-hard, its ungapped variant is
PSPACE-complete~\cite{crichigno2024,king2024,Rudolph2026}, normalized
low-lying spectral-density and quasi-Betti estimation are
$\mathsf{DQC1}$-hard~\cite{cade2021,gyurik2022}, and truncated random clique
complexes are conditionally hard on average~\cite{strelchuk2026average}.
An earlier, privately circulated version of the present work introduced the
boundary-log-pseudodeterminant approach to estimating integral-homology torsion
and the projective-plane join construction used to obtain torsion-bearing
complexes with controlled spectral gaps. Weso{\l}owski subsequently built on
and extended these ideas, using the join construction to prove NP-hardness and
conditional $\mathsf{QMA}_1$-hardness of torsion
detection~\cite{wesolowski2026torsion}. More recently, Nghiem, Berry, and Phan
used a similar construction to prove NP-hardness of $p$-torsion detection and
proposed a one-sided quantum witness based on mod-$p$ rank
estimation~\cite{nghiemBerryPhan2026torsion}. The proposed witness appears to
require further clarification. As currently presented, its sketching circuits
show no $i,j$ dependence, which would force
$\operatorname{rank}(U\partial_rV)\leq 1$. Further details are also needed to
justify the normalization, error analysis, and resulting complexity bounds. We
therefore do not use these claims here.

\subsection{Torsion and coefficient choice in classical TDA}

Persistent homology is almost always computed over a field, because field
coefficients guarantee that a persistence module decomposes into intervals.  Over the integers this decomposition
fails and torsion enters, which makes the algebra considerably more
subtle~\cite{ZomorodianCarlsson2005}.  TDA therefore usually sees torsion only
indirectly, through the choice of coefficients.

The universal coefficient theorem makes this precise: the $j$th Betti number
over $\F_p$ exceeds the rational one by the number of cyclic summands of order
divisible by $p$ in $H_j(X;\Z)$ and $H_{j-1}(X;\Z)$.  Comparing fields thus
reveals which primes divide the torsion, but not their exponents; for example,
$\Z/2\Z$ and $\Z/4\Z$ look the same over $\F_2$.  Boissonnat and Maria compute
persistence over many fields in a single pass and infer the prime divisors of
torsion coefficients~\cite{BoissonnatMaria2019}. Frosini gives stable
comparisons of persistent homology with arbitrary abelian coefficient
groups~\cite{Frosini2013}.  Obayashi and Yoshiwaki characterize when the
choice of field changes a persistence diagram, in terms of torsion in the
relative homology groups of the filtration~\cite{ObayashiYoshiwaki2023}.

The natural-image example of \cref{sec:introduction} illustrates the point.
Over $\Z_2$ the torus and the Klein bottle both have Betti numbers $(1,2,1)$,
so $\Z_2$ persistence alone cannot separate them.  Carlsson et al.\ supported a
Klein-bottle description by combining $\Z_2$ persistence with an explicit
parametrization by linear and quadratic image gradients; they neither compared
coefficient fields nor computed integral torsion~\cite{CarlssonEtAl2008}.
These works show that torsion matters in TDA, but they differ from ours in
output and computational model.  Field comparisons detect which primes divide
the torsion, whereas the log-pseudodeterminant route estimates the logarithm of
a group order, not its prime decomposition.

\subsection{Matrix-tree theorems, critical groups, and spectral sums}

{
For graphs, reduced Laplacian determinants, spanning trees, and critical-group orders
coincide.  Smith normal form is an integer row-and-column reduction whose nonunit
diagonal entries give the invariant factors of the group; it recovers the full
group decomposition, while the determinant recovers only its order~\cite{https://doi.org/10.1112/S0024609397003305,Lorenzini2008SmithNF}. Matrix-tree theorems generalize these statements to higher dimensions~\cite{duval2008simplicialmatrixtreetheorems, DuvalKlivansMartin2011Cellular}.Higher critical groups are cokernels of
reduced upper Laplacians under suitable hypotheses and have torsion-weighted tree
enumerators as their orders~\cite{dmtcs:2909}.

The corresponding order identities make the spectral connection explicit.
{For a connected
graph $G$ with Laplacian $L_G$,
\begin{equation}
  |K(G)|=\tau(G)=\frac{\pdet(L_G)}{|V(G)|},
  \label{eq:graph-mtt-intro}
\end{equation}
where $K(G)$ is the sandpile or critical group and $\tau(G)$ is the number of
spanning trees~\cite{https://doi.org/10.1112/S0024609397003305,Lorenzini2008SmithNF}.
Consequently, $\log\pdet(L_G)$ is both a spectral statistic and, after a known
normalization, the logarithm of the order of a finite abelian group.  In higher
dimensions, simplicial matrix-tree theorems replace the ordinary tree count by a sum
weighted by squares of the orders of homology-torsion groups~\cite{duval2008simplicialmatrixtreetheorems}.  Higher critical groups extend graph
sandpile groups and retain an exact determinant interpretation~\cite{dmtcs:2909}. Under the hypotheses of~\cref{cor:root-free-critical}, these descriptions agree:
\begin{equation}
  |K_i(X)|
  =
  \tau_{i+1}(X)
  =
  \sum_{T\in\mathcal T_{i+1}(X)}
  \left|\widetilde H_i(T;\Z)\right|^2,
  \label{eq:critical-tree-intro}
\end{equation}
where $\mathcal T_{i+1}(X)$ is the set of $(i+1)$-dimensional spanning trees.}

Quantum algorithms for generic spectral sums already include log-determinants of
nonsingular positive matrices and graph spanning-tree applications~\cite{luongo2024quantumalgorithmsspectralsums}.}
{Here the matrices are singular boundary Laplacians.  We encode the logarithmic
contributions of their nonzero spectrum in a success probability and assign the
kernel a complementary value whose contribution cancels the unknown rank.  We
combine estimates from several degrees through higher-dimensional matrix-tree
identities.  This avoids
constructing a reduced Laplacian, which would require a chosen higher-dimensional
tree.}

Complete colorful (multipartite clique) complexes already have explicit
all-degree total-Laplacian spectra and tree enumerators~\cite[Section~6.1, Equations~(30)--(34) and Theorem~6.2]
{DuvalKlivansMartin2011Cellular}.
Berry et al.~use the balanced top-degree case as a QTDA benchmark~\cite[Appendix~F]{BerryEtAl2024QTDA}.  In~\cref{sec:multipartite} we specialize
the all-degree formulas and extract the boundary-map parameters needed here.

\section{Algebraic and spectral preliminaries}
\label{sec:preliminaries}

\subsection{Boundary operators, Hodge Laplacians, and pseudodeterminants}

We outline a concise review of the mathematical framework used
throughout the paper, fixing notation and conventions along the way.
For a more elaborate introduction, see~\cite{edelsbrunner2010,hatcher2002} for
the topological preliminaries and~\cite{SchaubEtAl2020Hodge,WeiWei2025PersistentTopologicalLaplacians} for
topological Laplacians.  For quantum spectral-sum methods, see~\cite{luongo2024quantumalgorithmsspectralsums}.

A finite abstract simplicial complex $X$ is a finite collection of sets of vertices
closed under taking subsets.  Its elements are called faces; a face with $r+1$
vertices is a simplex of dimension $r$, $r$-simplex, and $\dim X$ is the largest simplex dimension.
The $r$-skeleton $X^{(r)}$ contains all faces of dimension at most $r$, and
$X$ is \emph{pure $d$-dimensional} if every face is contained in a
$d$-simplex.

Assign an orientation to every face, write $X_r$ for the set of $r$-simplices,
and set $f_r=|X_r|$.  The chain group $C_r(X;\Z)$ is the free abelian group
generated by the oriented $r$-simplices, which means it consists of all integer
linear combinations of those simplices.  For an oriented $r$-simplex, the boundary map is
\begin{equation}
  B_r[v_0,\ldots,v_r]
  =
  \sum_{j=0}^{r}(-1)^j
  [v_0,\ldots,\widehat v_j,\ldots,v_r],
  \label{eq:simplicial-boundary}
\end{equation}
where $v_j$ is the $j$-th vertex in the chosen ordering and the hat denotes its
omission.  This formula determines the boundary matrix $B_r$, whose columns are
indexed by oriented $r$-simplices and rows are indexed by oriented
$(r-1)$-simplices.  The same matrix is used when the coefficients are extended
to $\mathbb{Q}$ or $\mathbb{R}$.  The empty face is the unique $(-1)$-simplex, so
$f_{-1}=1$ for nonempty $X$.  With augmented chains,
\begin{equation}
  \cdots
  \xrightarrow{B_{r+1}} C_r(X;\Z)
  \xrightarrow{B_r} C_{r-1}(X;\Z)
  \xrightarrow{}
  \cdots
  \xrightarrow{B_0} C_{-1}(X;\Z)
  \longrightarrow 0,
  \label{eq:augmented-chain-complex}
\end{equation}
where $C_{-1}(X;\Z)\cong\Z$ and $B_0$ sends each vertex to the empty
simplex.  In the oriented-simplex bases, each $B_r$ is an integer matrix.
These matrices satisfy $B_rB_{r+1}=0$, because taking the boundary twice gives
zero.  Thus every $r$-boundary is an $r$-cycle, with
\begin{equation}
  \im B_{r+1}\subseteq\ker B_r.
  \label{eq:inclusion-boundary-cycle}
\end{equation}
The elements of $\ker B_r$ are called $r$-cycles, and those of $\im B_{r+1}$
are called $r$-boundaries.  For notational convenience, we set $C_r=0$ outside
$-1\leq r\leq\dim X$, so the corresponding boundary maps, including $B_{-1}$
and $B_{\dim X+1}$, are zero.

Over $\R$, for $-1\leq r\leq\dim X$, define the upper, lower, and total
Hodge Laplacians by
\begin{equation}
  L_r^{\up}=B_{r+1}B_{r+1}^{\mathsf T},
  \qquad
  L_r^{\down}=B_r^{\mathsf T}B_r,
  \qquad
  {\Delta_r
  =B_r^{\mathsf T}B_r+B_{r+1}B_{r+1}^{\mathsf T}
  =L_r^{\up}+L_r^{\down}.}
  \label{eq:hodge-laplacians}
\end{equation}
Here $^{\mathsf T}$ denotes the matrix transpose.  The lower term compares simplices through their
faces, while the upper term compares them through higher-dimensional cofaces.
Let $\wtH_r(X;\R)$ denote reduced homology, and define the reduced Betti number by
\begin{equation}
  \wtbeta_r(X):=\dim_{\R}\wtH_r(X;\R).
  \label{eq:reduced-betti-def}
\end{equation}
The Hodge decomposition then gives
\begin{equation}
  C_r(X;\R)
  =
  \im B_{r+1}
  \mathbin{\oplus}
  \ker\Delta_r
  \mathbin{\oplus}
  \im B_r^{\mathsf T},
  \qquad
  {\dim\ker\Delta_r=\wtbeta_r(X).}
  \label{eq:hodge-decomposition}
\end{equation}
The three summands are mutually orthogonal.  Each reduced homology class has a
unique representative in $\ker\Delta_r$.  For a nonempty complex, reduced homology agrees
with ordinary homology in positive degrees and removes one ubiquitous
degree-zero component class; thus $\wtbeta_0=\beta_0-1$, and a connected
complex has $\wtH_0=0$.  The augmented convention also gives
$\Delta_{-1}=B_0B_0^{\mathsf T}$ and $\wtH_{-1}(X)=0$.

{
While the kernel of the combinatorial Laplacian captures Betti numbers, the
torsion-sensitive quantities studied in this paper are built from its nonzero
spectrum.  For a positive semidefinite matrix $A$, let
\begin{equation}
  \pdet(A)=\prod_{\lambda\in\Spec(A),\,\lambda>0}\lambda.
  \label{eq:pdet-definition}
\end{equation}
}
Here $\Spec(A)$ is the eigenvalue multiset, so multiplicities are included.
Taking the logarithm converts this eigenvalue product into the spectral sum
\begin{equation}
  \log\pdet(A)=\sum_{\lambda>0}\log\lambda,
  \label{eq:log-pdet-definition}
\end{equation}
which is the form estimated by our quantum algorithm.

For a matrix $B$, the positive eigenvalues of $BB^{\mathsf T}$ are the squares
of its positive singular values.  Thus, if $\zeta_1,\ldots,\zeta_{\chg{\rank B}}$ are the
positive singular values of $B$, then
\begin{equation}
  \pdet(BB^{\mathsf T})=\prod_{j=1}^{\chg{\rank B}}\zeta_j^2.
  \label{eq:pdet-singular-values}
\end{equation}
Therefore,
\begin{equation}
  \log\pdet(BB^{\mathsf T})=\sum_{j=1}^{\chg{\rank B}}\log\zeta_j^2.
  \label{eq:log-pdet-singular-values}
\end{equation}

We now fix notation for the boundary-matrix pseudodeterminants used throughout
the paper.  For $r\geq0$, define
\begin{equation}
  \pi_r(X)
  :=
  \pdet(B_rB_r^{\mathsf T}),
  \qquad
  \ell_r(X):=\log\pi_r(X).
  \label{eq:pi-and-ell}
\end{equation}
Since $B_rB_r^{\mathsf T}=L_{r-1}^{\up}$, we can write
\begin{equation}
  \pi_r(X)=\pdet(L_{r-1}^{\up}).
  \label{eq:pi-upper-laplacian}
\end{equation}
We use the convention $\pdet(0)=1$.

\paragraph{Degree-zero convention.}
\label{par:augmentation}
At degree zero, augmentation makes $B_0$ the row vector $(1,\ldots,1)$, so
$B_0B_0^{\mathsf T}=[f_0]$.  Therefore,
\begin{equation}
  \pi_0(X)=f_0,
  \qquad
  \ell_0(X)=\log f_0.
  \label{eq:pi-zero}
\end{equation}
This supplies the degree-zero factor required later in the alternating
matrix-tree formula.  Without augmentation, $B_0$ is absent and the constant
vector remains as an additional zero mode of the degree-zero Laplacian.

\subsection{Integral reduced homology and higher critical groups}

The reduced integral homology of $X$, for $-1\leq r\leq\dim X$, is
\begin{equation}
  \wtH_r(X;\Z)=\ker B_r/\im B_{r+1}.
  \label{eq:integral-homology}
\end{equation}
Two cycles represent the same homology class when they differ by a boundary.  A
nonzero class represented by a cycle $z$ is a \emph{torsion class} if $nz$ is a
boundary for some positive integer $n$.  Equivalently, it is a nonzero
finite-order element of the integral homology group.  These classes form the
finite subgroup $\Tor\wtH_r(X;\Z)$.
For $0\leq r<\dim X$, the $r$th higher critical group is
\begin{equation}
  K_r(X)
  :=
  \ker B_r\big/
  \im(B_{r+1}B_{r+1}^{\mathsf T}).
  \label{eq:critical-group-definition}
\end{equation}
This is the higher-dimensional analogue of the graph sandpile group.  It
differs from homology because cycles are taken modulo the image of the upper
Laplacian rather than modulo all boundaries.

For a finitely generated abelian group, $\rank$ denotes its free abelian rank.
The following exact sequence makes this distinction precise.

\begin{proposition}[Critical group versus homology]
\label{prop:critical-exact-sequence}
For every finite simplicial complex $X$ and every
$0\leq r<\dim X$, there is a natural short
exact sequence
\begin{equation}
  0
  \longrightarrow
  \frac{\im B_{r+1}}
       {\im(B_{r+1}B_{r+1}^{\mathsf T})}
  \longrightarrow
  K_r(X)
  \longrightarrow
  \wtH_r(X;\Z)
  \longrightarrow
  0.
  \label{eq:critical-exact-sequence}
\end{equation}
The leftmost quotient is always finite.  Consequently,
\begin{equation}
  \rank K_r(X)=\wtbeta_r(X),
  \qquad
  K_r(X)\text{ is finite if and only if }\wtbeta_r(X)=0.
  \label{eq:critical-finiteness}
\end{equation}
In that case all three terms are finite and
\begin{equation}
  |K_r(X)|
  =
  \left|
  \frac{\im B_{r+1}}
       {\im(B_{r+1}B_{r+1}^{\mathsf T})}
  \right|
  |\wtH_r(X;\Z)|.
  \label{eq:critical-order-factorization}
\end{equation}
\end{proposition}

\begin{proof}
The inclusions
\begin{equation}
  \im(B_{r+1}B_{r+1}^{\mathsf T})
  \subseteq
  \im B_{r+1}
  \subseteq
  \ker B_r
  \label{eq:image-chain-inclusions}
\end{equation}
follow from $B_rB_{r+1}=0$.  Taking successive quotients gives
\eqref{eq:critical-exact-sequence}.  Over $\Q$, the matrices $B_{r+1}$ and
$B_{r+1}B_{r+1}^{\mathsf T}$ have the same rank and image, so their integral image
lattices have finite index.  Taking ranks in
\eqref{eq:critical-exact-sequence} gives \eqref{eq:critical-finiteness};
multiplicativity of orders gives \eqref{eq:critical-order-factorization} in the
finite case.
\end{proof}

The exact sequence explains why $|K_r(X)|$ cannot generally be identified with
the torsion order of reduced integral homology.  Even when reduced homology is
finite, the critical group contains an additional lattice-index contribution.
This additional factor measures the difference between all integral boundaries
and those produced by the upper Laplacian.

The six-vertex triangulation $T$ of $\R P^2$, the real projective plane, makes
the distinction concrete:
$\wtH_1(T;\Z)\cong\Z/2\Z$, whereas the corresponding torsion-weighted
two-tree enumerator equals $|K_1(T)|=4$.  \Cref{ex:rp2-tree} verifies these
values explicitly.

\begin{figure}[H]
  \centering
  \begin{tikzpicture}[
    x=1.08cm,
    y=1.25cm,
    >=Stealth,
    interior edge/.style={draw=black!58,line width=0.55pt},
    boundary edge/.style={draw=teal!65!black,line width=1.05pt},
    one arrow/.style={postaction={decorate},decoration={markings,
      mark=at position 0.52 with {\arrow{>}}}},
    two arrows/.style={postaction={decorate},decoration={markings,
      mark=at position 0.44 with {\arrow{>}},
      mark=at position 0.60 with {\arrow{>}}}},
    three arrows/.style={postaction={decorate},decoration={markings,
      mark=at position 0.38 with {\arrow{>}},
      mark=at position 0.52 with {\arrow{>}},
      mark=at position 0.66 with {\arrow{>}}}},
    vertex/.style={circle,fill=blue!65!black,inner sep=1.7pt},
    face label/.style={font=\scriptsize\sffamily,text=black!58}
  ]
    \coordinate (top) at (0,3.15);
    \coordinate (ur) at (2.35,2.05);
    \coordinate (lr) at (2.35,0);
    \coordinate (bottom) at (0,-1.12);
    \coordinate (ll) at (-2.35,0);
    \coordinate (ul) at (-2.35,2.05);
    \coordinate (v3) at (-0.92,1.32);
    \coordinate (v2) at (0.92,1.32);
    \coordinate (v6) at (0,0.22);

    \fill[blue!7] (v3)--(v2)--(v6)--cycle;
    \draw[interior edge]
      (top)--(v3) (top)--(v2)
      (ul)--(v3) (ur)--(v2)
      (ll)--(v3) (ll)--(v6)
      (lr)--(v2) (lr)--(v6)
      (bottom)--(v6)
      (v3)--(v2)--(v6)--cycle;

    \draw[boundary edge,one arrow] (top)--(ur);
    \draw[boundary edge,two arrows] (ur)--(lr);
    \draw[boundary edge,three arrows] (lr)--(bottom);
    \draw[boundary edge,one arrow] (bottom)--(ll);
    \draw[boundary edge,two arrows] (ll)--(ul);
    \draw[boundary edge,three arrows] (ul)--(top);

    \foreach \p in {top,ur,lr,bottom,ll,ul,v3,v2,v6}
      \node[vertex] at (\p) {};
    \node[above=2pt] at (top) {$1$};
    \node[right=3pt] at (ur) {$4$};
    \node[right=3pt] at (lr) {$5$};
    \node[below=2pt] at (bottom) {$1$};
    \node[left=3pt] at (ll) {$4$};
    \node[left=3pt] at (ul) {$5$};
    \node[left=3pt] at (v3) {$3$};
    \node[right=3pt] at (v2) {$2$};
    \node[below=3pt] at (v6) {$6$};

    \node[face label] at (0,2.02) {$123$};
    \node[face label] at (-1.08,2.20) {$135$};
    \node[face label] at (1.08,2.20) {$124$};
    \node[face label] at (-1.84,1.15) {$345$};
    \node[face label] at (1.84,1.15) {$245$};
    \node[face label] at (-1.08,0.58) {$346$};
    \node[face label] at (1.08,0.58) {$256$};
    \node[face label,text=black!72] at (0,0.94) {$236$};
    \node[face label] at (-0.78,-0.30) {$146$};
    \node[face label] at (0.78,-0.30) {$156$};
  \end{tikzpicture}
  \caption{The six-vertex triangulation of $\mathbb{R}P^2$ used in
  \cref{ex:rp2-tree}.  Opposite boundary positions bearing the same vertex
  label are identified, and boundary edges
  with the same number of arrowheads are glued in the indicated orientation.
  The ten face labels are exactly the facets listed in the example; the central
  facet $236$ is shaded.}
  \label{fig:rp2-triangulation}
\end{figure}
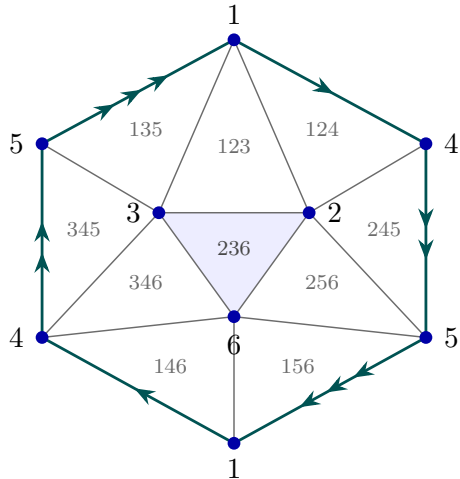

\section{From boundary log pseudodeterminants to torsion-sensitive invariants}
\label{sec:matrix-tree}

\subsection{Simplicial spanning trees and matrix-tree identities}

The definitions and matrix-tree theorem in this subsection are due to
Duval, Klivans, and Martin~\cite{duval2008simplicialmatrixtreetheorems,dmtcs:2909}; we restate them to
fix the conventions needed later.

\begin{definition}[APC complex]
\label{def:apc}
A pure $d$-dimensional complex $X$ is \emph{acyclic in positive codimension}
(APC) if
\begin{equation}
  \wtH_q(X;\Q)=0
  \qquad\text{for all }q<d.
  \label{eq:apc}
\end{equation}
\end{definition}

APC removes free homology below the top dimension, but finite integral torsion may
remain.

\begin{definition}[Simplicial spanning tree]
\label{def:simplicial-spanning-tree}
Let $X$ be a pure $d$-dimensional simplicial complex.  A subcomplex
$T\subseteq X$ is a $d$-dimensional simplicial spanning tree if
$T^{(d-1)}=X^{(d-1)}$ and
\begin{align}
  \wtH_d(T;\Z)&=0,
  \label{eq:tree-top-acyclic}\\
  \wtH_{d-1}(T;\Q)&=0,
  \label{eq:tree-penultimate-rational}\\
  f_d(T)&=f_d(X)-\wtbeta_d(X)+\wtbeta_{d-1}(X).
  \label{eq:tree-face-count}
\end{align}
An $r$-dimensional spanning tree of $X$ is a spanning tree of $X^{(r)}$.
\end{definition}

The three displayed conditions are not independent: any two imply the third.
For $d=1$ this is the usual graph spanning tree.  In higher dimensions, the
tree retains the entire $(d-1)$-skeleton and has the minimal
top-dimensional face count compatible with eliminating rational
$(d-1)$-homology without creating a top-dimensional cycle.
In particular,
\begin{equation}
  \wtH_{d-1}(T;\Q)=0
  \quad\Longleftrightarrow\quad
  \wtH_{d-1}(T;\Z)\ \text{is finite}.
  \label{eq:rational-not-integral}
\end{equation}
It doesn't need to vanish over integral homology.  Indeed, higher-dimensional trees with nontrivial
finite homology are precisely what create torsion weights.

{
The following quantity is the torsion-weighted tree enumerator that appears in the
higher-dimensional matrix-tree theorem.
Let $\mathcal T_r(X)$ be the set of $r$-dimensional spanning trees and define
\begin{equation}
  \tau_r(X)
  :=
  \sum_{T\in\mathcal T_r(X)}
  \left|\wtH_{r-1}(T;\Z)\right|^2.
  \label{eq:tree-enumerator}
\end{equation}
}
$\tau_r(X)$ is a torsion-weighted tree count; it reduces to the number of
$r$-dimensional spanning trees when every displayed homology group has
order one.
For an APC complex, each homology group below the top degree is finite.  Here,
\begin{equation}
  t_q(X):=|\wtH_q(X;\Z)|
  \label{eq:tq-definition}
\end{equation}
is the order of the $q$th reduced integral homology group, with the convention
that the trivial group has order one.
For the augmented base cases we set
\begin{equation}
  \tau_{-1}(X):=1,
  \qquad
  t_{-1}(X):=t_{-2}(X):=1.
  \label{eq:augmented-tree-conventions}
\end{equation}
With these conventions, $\tau_0(X)=\pi_0(X)=f_0$.

Every skeleton $X^{(r)}$ of an APC complex is APC, so the simplicial matrix-tree
theorem in~\cite{duval2008simplicialmatrixtreetheorems} can be applied degree by
degree, providing us with an all-degree statement.

\begin{theorem}[Duval--Klivans--Martin matrix-tree theorem]
\label{thm:simplicial-matrix-tree}
Let $X$ be a finite APC simplicial complex of dimension $d$.  For
$0\leq r\leq d$,
\begin{equation}
  \pi_r(X)
  =
  \frac{\tau_r(X)\tau_{r-1}(X)}
       {t_{r-2}(X)^2},
  \label{eq:matrix-tree-recurrence}
\end{equation}
with the conventions in \eqref{eq:augmented-tree-conventions}.
For $1\leq r\leq d$, if $R$ is an $(r-1)$-dimensional spanning tree and
$\widetilde L_{r-1}^{\up}(X;R)$ is obtained from $L_{r-1}^{\up}(X)$ by deleting
the rows and columns indexed by the $(r-1)$-faces of $R$, then
\begin{equation}
  \tau_r(X)
  =
  \frac{t_{r-2}(X)^2}{t_{r-2}(R)^2}
  \det\widetilde L_{r-1}^{\up}(X;R).
  \label{eq:reduced-matrix-tree}
\end{equation}
\end{theorem}

We call such an $(r-1)$-dimensional spanning tree $R$ a \emph{root}.  Deleting
the rows and columns indexed by its $(r-1)$-faces is the higher-dimensional
analogue of deleting one row and column from a graph Laplacian to obtain a
nonsingular cofactor.
The factors $t_{r-2}(X)$ and $t_{r-2}(R)$ are part of the cited theorem.
They cannot in general be dropped when lower integral homology contains torsion.

\subsection{Alternating log-pseudodeterminant identity}
\label{subsec:root-free}

The following identity expresses the torsion-weighted tree enumerator in terms of
the log pseudodeterminants of the boundary matrices and lower-homology corrections.
Taking logarithms and iterating the Duval--Klivans--Martin matrix-tree
recurrence~\cite{duval2008simplicialmatrixtreetheorems} degree by degree
reveals this alternating formula, which is used by our quantum algorithm.

\begin{proposition}[Alternating log-pseudodeterminant formula]
\label{prop:alternating-pdet}
Let $X$ be APC and let $0\leq k\leq\dim X$.  Then
\begin{equation}
  \log\tau_k(X)
  =
  \sum_{r=0}^{k}(-1)^{k-r}\ell_r(X)
  +
  2\sum_{q=0}^{k-2}
     (-1)^{k-q-2}\log t_q(X).
  \label{eq:alternating-pdet-corrected}
\end{equation}
In particular, if $\wtH_q(X;\Z)=0$ for $0\leq q\leq k-2$, then
\begin{equation}
  \log\tau_k(X)
  =
  \sum_{r=0}^{k}(-1)^{k-r}\ell_r(X).
  \label{eq:alternating-pdet-uncorrected}
\end{equation}
\end{proposition}

\begin{proof}
Taking logarithms in \eqref{eq:matrix-tree-recurrence} gives
\begin{equation}
  \log\tau_k
  =
  \ell_k-\log\tau_{k-1}+2\log t_{k-2}.
  \label{eq:tau-recurrence-log}
\end{equation}
Iterating to $\tau_0=\pi_0$ yields
\eqref{eq:alternating-pdet-corrected}.  When the lower homology groups are
trivial, $t_q(X)=1$ and hence $\log t_q(X)=0$ in every relevant degree. The
correction terms therefore vanish, recovering the correction-free
specialization of the Duval--Klivans--Martin recurrence in
\eqref{eq:alternating-pdet-uncorrected}.
\end{proof}

Under the hypotheses of \cref{cor:root-free-critical}, Corollary~4.2 of
Duval, Klivans, and Martin~\cite{dmtcs:2909} identifies the critical-group order
with the next-dimensional tree enumerator:
\begin{equation}
  |K_i(X)|
  =
  \tau_{i+1}(X)
  =
  \sum_{T\in\mathcal T_{i+1}(X)}
  \left|\wtH_i(T;\Z)\right|^2.
  \label{eq:critical-equals-tree}
\end{equation}

\begin{corollary}[Corrected critical-group order]
\label{cor:root-free-critical}
Let $X$ be a finite pure APC complex of dimension $d$, and let
$0\leq i<d$.  Assume $\wtH_{i-1}(X;\Z)=0$ and that $X$ has an
$i$-dimensional spanning tree $R$ with
$\wtH_{i-1}(R;\Z)=0$.  Then $K_i(X)$ is finite and
\begin{equation}
  \log|K_i(X)|
  =
  \sum_{r=0}^{i+1}(-1)^{i+1-r}\ell_r(X)
  +
  2\sum_{q=0}^{i-1}
     (-1)^{i-1-q}\log t_q(X).
  \label{eq:critical-pdet-corrected}
\end{equation}
If, in addition, $\wtH_q(X;\Z)=0$ for $0\leq q\leq i-1$, then
{
\begin{equation}
  \boxed{
  \log|K_i(X)|
  =
  \sum_{r=0}^{i+1}
     (-1)^{i+1-r}
     \log\pdet(B_rB_r^{\mathsf T}).
  }
  \label{eq:critical-pdet-boxed}
\end{equation}
}
\end{corollary}

\begin{proof}
Substitute $k=i+1$ in \cref{prop:alternating-pdet} and use
\eqref{eq:critical-equals-tree}.
\end{proof}

{The alternating identity underlies the main quantum algorithm.
It uses log pseudodeterminants for
$B_0,\ldots,B_{i+1}$ and does not require selecting a higher-dimensional spanning tree.}
Equivalently, the upper-Laplacian degrees are $-1,\ldots,i$.  Even when the homological
conditions needed for \eqref{eq:critical-pdet-boxed} fail, the vector
$(\ell_0,\ldots,\ell_{i+1})$ remains a well-defined spectral output, but its
critical-group interpretation may fail.

\subsection{Homology torsion from certified higher-dimensional trees}
\label{subsec:certified-tree}

For a certified higher-dimensional tree, the matrix-tree determinant recovers
the order of integral homology torsion.  The following criterion identifies
this setting.

\begin{proposition}[Minimal rational filling criterion]
\label{prop:minimal-filling}
Let $k\geq1$, let $X$ be a $k$-dimensional APC complex, and let
$T\subseteq X$ contain $X^{(k-1)}$.  The following linear-algebraic conditions
characterize a $k$-dimensional spanning tree:
\begin{enumerate}[label=(\roman*),leftmargin=2em]
  \item the columns of $B_k(X)$ indexed by $T_k$ are linearly independent over
  $\Q$;
  \item those columns span the $(k-1)$-cycle space
  $\ker(B_{k-1}(X)\otimes\Q)$.
\end{enumerate}
Equivalently, $T$ is an inclusion-minimal choice of $k$-simplices that makes
$\wtH_{k-1}(T;\Q)$ vanish.  Its integral group
$\wtH_{k-1}(T;\Z)$ is finite.  Since $X$ is APC, every rational
$(k-1)$-cycle is a boundary.  The spanning tree $T$ selects a minimal set of
$k$-simplices whose boundaries generate all such cycles.
\end{proposition}

\begin{proof}
Condition (i) is $\wtH_k(T;\Q)=0$.  Because $T$ contains the full
$(k-1)$-skeleton, condition (ii) is
$\wtH_{k-1}(T;\Q)=0$.  Together, they say that the selected boundary columns
form a basis of the rational cycle space and therefore have the face count in
\eqref{eq:tree-face-count}.  They also make $T$ pure.  Indeed, let $F$ be a
$(k-1)$-face.  Since $X$ is pure, some $k$-face $\sigma$ of $X$ contains $F$.
The cycle $\partial\sigma$ is a rational linear combination of the selected
boundary columns.  Its coefficient at $F$ is nonzero, so at least one selected
$k$-face contains $F$.  Every lower-dimensional face is contained in such a ridge
and hence in a selected $k$-face.  Top-dimensional integral homology is a subgroup
of the free group $C_k(T;\Z)$ and is therefore torsion-free; its rational
vanishing implies its integral vanishing.  \Cref{def:simplicial-spanning-tree}
applies.
The equivalence between rational vanishing and finite integral homology follows
from the structure theorem for finitely generated abelian groups.
\end{proof}

This finite group need not be trivial.  We call an inclusion $T\subseteq X$
satisfying the hypotheses and the two conditions of
\cref{prop:minimal-filling} a \emph{certified $k$-dimensional tree}.

For a certified tree $T$, write
\begin{equation}
  h(T):=\left|\wtH_{k-1}(T;\Z)\right|.
  \label{eq:hT}
\end{equation}

\begin{theorem}[Homology-torsion order from a certified higher-dimensional tree]
\label{thm:certified-tree-torsion}
Let $k\geq1$, let $X$ be a finite $k$-dimensional APC complex, and let
$T\subseteq X$ be a certified $k$-dimensional tree.  Then the only
$k$-dimensional spanning tree of $T$ is $T$ itself, and
\begin{equation}
  \tau_k(T)=h(T)^2.
  \label{eq:single-tree-square}
\end{equation}
For any $(k-1)$-dimensional spanning tree $R$ used to reduce
$L_{k-1}^{\up}(T)$,
\begin{equation}
  h(T)
  =
  \frac{t_{k-2}(T)}{t_{k-2}(R)}
  \sqrt{\det\widetilde L_{k-1}^{\up}(T;R)}.
  \label{eq:certified-tree-reduced}
\end{equation}
\end{theorem}

\begin{proof}
By \cref{prop:minimal-filling}, $T$ is a pure $k$-dimensional spanning tree of
$X$.  It is also APC: certification gives rational vanishing in degree $k-1$,
while below that degree $T$ and $X$ have the same homology because they have the
same $(k-1)$-skeleton.
Every $k$-dimensional spanning tree of $T$ must contain $T^{(k-1)}$ and, since
$\wtbeta_k(T)=\wtbeta_{k-1}(T)=0$, must contain $f_k(T)$ top-dimensional
faces.  It is therefore $T$ itself.  The definition
\eqref{eq:tree-enumerator} gives \eqref{eq:single-tree-square}.
Substituting $X=T$ in \eqref{eq:reduced-matrix-tree} and taking the positive
square root yields \eqref{eq:certified-tree-reduced}.
\end{proof}

Here ``certified'' means that the conditions in
\cref{prop:minimal-filling} are supplied as an external promise or witness; the
cost of verifying them is not included in the determinant estimator.  The
reduced-determinant route also requires a $(k-1)$-dimensional root $R$ and the
orders $t_{k-2}(T)$ and $t_{k-2}(R)$, unless both are known to be one.

\paragraph{Trivial lower-homology corrections.}
If both lower correction groups in \eqref{eq:certified-tree-reduced} are
trivial, its prefactor is one and
\begin{equation}
  \boxed{
  h(T)=\sqrt{\det\widetilde L_{k-1}^{\up}(T;R)}.
  }
  \label{eq:certified-tree-simple}
\end{equation}
If $\wtH_q(T;\Z)=0$ for $0\leq q\leq k-2$, then
\eqref{eq:single-tree-square} and \cref{prop:alternating-pdet} yield the
equivalent alternating identity
\begin{equation}
  \boxed{
  \log h(T)
  =
  \frac12\sum_{r=0}^{k}
    (-1)^{k-r}\ell_r(T).
  }
  \label{eq:certified-tree-root-free}
\end{equation}

\begin{remark}
\label{rem:tree-claim-correction}
Certification requires an inclusion-minimal set of $k$-simplices whose
boundaries span the rational $(k-1)$-cycle space.  Requiring integral rather
than rational vanishing would force $h(T)=1$ and eliminate the torsion to be
recovered.
\end{remark}

\begin{example}[Six-vertex projective-plane tree]
\label{ex:rp2-tree}
Writing $123$ for $\{1,2,3\}$, let $T$ be the complex on $[6]$ with facets
\begin{equation}
  123,124,135,146,156,236,245,256,345,346.
  \label{eq:rp2-facets}
\end{equation}
These facets are displayed in \cref{fig:rp2-triangulation}.  Its one-skeleton
is $K_6$, every vertex link is a five-cycle, and
$\chi(T)=6-15+10=1$.  Thus $T$ is the standard six-vertex triangulation of
$\R P^2$; projective-plane two-dimensional trees on six vertices are discussed
in~\cite[Section~3.1]{dmtcs:2909}.  Consequently,
\begin{equation}
  \wtH_2(T;\Z)=0,
  \qquad
\wtH_1(T;\Z)\cong\Z/2\Z.
  \label{eq:rp2-homology}
\end{equation}
Its face numbers are $(f_0,f_1,f_2)=(6,15,10)$; hence $\tau_2(T)=4$.
Choosing as the lower root the graph tree $R$ consisting of the five edges
incident to one vertex gives
\begin{equation}
  \det\widetilde L_1^{\up}(T;R)=4,
  \qquad
  |\wtH_1(T;\Z)|=\sqrt4=2.
  \label{eq:rp2-reduced-det}
\end{equation}
All Betti numbers of $T$ over $\Q$ agree with those of a point, whereas the
alternating nonzero-boundary-spectrum statistic detects its order-two torsion:
\begin{equation}
  \pi_0=6,\qquad
  \pi_1=6^5,\qquad
  \pi_2=4\cdot6^4,
  \label{eq:rp2-pdets}
\end{equation}
so $\tau_2=\pi_2\pi_0/\pi_1=4$ and
$\log|\wtH_1(T;\Z)|=\frac12\log4=\log2$.

The barycentric subdivision $\operatorname{sd}T$ is the clique complex of the
comparability graph of $T$'s face poset. It preserves $T$'s topology, including
the $\Z/2\Z$ homology torsion, but need not preserve $T$'s numerical
pseudodeterminants.
\end{example}

\begin{remark}[Indexing and lower-homology corrections]
\label{rem:indexing}
Equation~\eqref{eq:critical-pdet-boxed} corrects the upper limit and exponent
printed in Corollary~4.4 of~\cite{dmtcs:2909} and retains the lower-homology
factors required under its weaker hypotheses.  The graph case reduces to
$\log|K_0(G)|=\log\pdet(L_G)-\log|V(G)|$, providing an immediate convention
check.  \Cref{app:indexing-correction} gives the detailed comparison.
\end{remark}

\section{Quantum algorithm for log-pseudodeterminant estimation
of upper Hodge Laplacians}
\label{sec:quantum-algorithm}

The identities in \cref{sec:matrix-tree} express the stated finite
group orders in terms of upper Hodge Laplacian log pseudodeterminants
and, when required, homology-group orders in lower degrees.
We give a quantum algorithm that estimates the normalized log
pseudodeterminant from a measurement probability. The algorithm uses
a block encoding of the boundary matrix, quantum singular value
transformation, and amplitude estimation.

\subsection{Overview of the quantum algorithm}
\label{subsec:root-free-quantum-overview}

The matrix-tree identities in \cref{sec:matrix-tree} express the
logarithms of tree enumerators and, under the stated hypotheses,
finite group orders in terms of the quantities $\ell_r(X)$.
We now give a quantum algorithm for estimating these quantities
using coherent access to the corresponding boundary matrices.

For a real boundary matrix $B$ with $D$ columns, write
\begin{equation}
 \ell(B):=\log\pdet(BB^{\mathsf T})
 =\sum_{\sigma_j(B)>0}\log\!\bigl(\sigma_j(B)^2\bigr),
 \label{eq:ellB}
\end{equation}
where the positive singular values are counted with multiplicity.
For $B=B_r$, this agrees with the notation introduced earlier:
$\ell(B_r)=\ell_r(X)=\log\pdet(L_{r-1}^{\up})$.
Adding zero rows or columns leaves this quantity unchanged.
The algorithm estimates $\ell(B)/D$ to prescribed additive error,
where $D$ is the number of input basis states in the chosen
representation of $B$. To estimate $\ell(B)$ to additive error
$\varepsilon$, we estimate $\ell(B)/D$ to additive error
$\varepsilon/D$ and multiply the result by $D$.

The construction uses the bottom-left block of a quantum singular
value transformation (QSVT) circuit~\cite{Gilyen2019QuantumArithmetics}. This block contributes a factor
$\sqrt{1-x^2}$ to the amplitude associated with a normalized singular
value $x$. We therefore approximate
\[
 \sqrt{\frac{\log(1/x^2)}{1-x^2}}
\]
by a polynomial, with a separately chosen value at zero.
The resulting success probability determines $\ell(B)/D$ without
requiring $\rank B$. Amplitude estimation supplies an estimate of
this probability, and classical postprocessing gives the required
log pseudodeterminant.

\subsection{Mapping simplices to qubits}
\label{subsec:simplex-encoding}

Let $X$ be a simplicial complex on the ordered vertex set
$[N]=\{1,\ldots,N\}$. We use the direct mapping of simplices to
qubits~\cite{McArdle2022AQubits}: each vertex is represented by one
qubit. For a subset $\sigma\subseteq[N]$, define
\[
 \ket{\sigma}:=\ket{z_1\cdots z_N},
 \qquad z_j=
 \begin{cases}
  1,&j\in\sigma,\\
  0,&j\notin\sigma.
 \end{cases}
\]
A $k$-simplex is therefore represented by a computational basis state
of Hamming weight $k+1$. The increasing vertex order fixes its
orientation, and the empty face is represented by $\ket{0^N}$.

For the degree-$r$ boundary, with $0\leq r\leq\dim X$, the input
space is the Hamming-weight-$(r+1)$ subspace, of dimension
\begin{equation}
 D_r:=\binom{N}{r+1}.
 \label{eq:candidate-dimension}
\end{equation}
The output space is the Hamming-weight-$r$ subspace.
The subspace spanned by the $r$-simplices of $X$ has dimension $f_r$
and projector
\[
 \Pi_r:=\sum_{\sigma\in X_r}\ket{\sigma}\!\bra{\sigma},
 \qquad
 \Pi_{-1}:=\ket{0^N}\!\bra{0^N}.
\]
Thus the input space contains all $(r+1)$-vertex subsets, whereas
$\Pi_r$ retains only those belonging to $X$.

Write $\mathcal B_r$ for the extension of the simplicial
boundary $B_r$ by zero rows and columns for subsets that are not
simplices of $X$. Its positive singular values agree with those
of $B_r$, and hence
\[
 \ell(\mathcal B_r)=\ell_r(X).
\]
We apply the algorithm to this extended matrix.

\subsection{Input model and oracle access}

Consider a real matrix
\begin{equation}
 B:\mathcal H_R\longrightarrow\mathcal H_L,
 \qquad \dim\mathcal H_R=D\geq1,
 \label{eq:general-B}
\end{equation}
with specified orthonormal input and output bases.
For the direct mapping, $B=\mathcal B_r$ and $D=D_r$. We embed the input and output spaces into the same $s$-qubit register,
where $2^s\geq\max\{\dim\mathcal H_R,\dim\mathcal H_L\}$, using specified
isometries
\[
 V_R:\mathcal H_R\longrightarrow(\mathbb C^2)^{\otimes s},
 \qquad
 V_L:\mathcal H_L\longrightarrow(\mathbb C^2)^{\otimes s}.
\]
The resulting square matrix is $\overline B=V_LBV_R^\dagger$; write
$P_R:=V_RV_R^\dagger$ for the embedded right-space projector.
The $s$ system qubits do not include the block-encoding ancillas.

\begin{assumption}[Oracle access and singular-value gap]
\label{ass:block-trace}
We have controlled access to an exact
$(\alpha,n_B,0)$ block encoding $U_B$ of $\overline B$, with
known $\alpha\geq1$, and to its inverse:
\begin{equation}
 (\bra{0^{n_B}}\otimes I)U_B
 (\ket{0^{n_B}}\otimes I)
 =A:=\frac{\overline B}{\alpha}.
 \label{eq:standard-block-encoding}
\end{equation}
We also have controlled access to a state-preparation unitary
$P_\Phi$ and its inverse, where
\begin{equation}
 P_\Phi\ket0=\ket{\Phi_D}
 :=\frac1{\sqrt D}\sum_{j=1}^{D}
     V_R\ket j\otimes V_R\ket j,
 \label{eq:max-entangled}
\end{equation}
with reusable work registers returned to zero.
The required projector-controlled phase shifts and the selective
phase flips used by amplitude estimation are available.

A supplied parameter $\gamma\in(0,1]$ satisfies
\begin{equation}
 \operatorname{sing}(A)\subseteq\{0\}\cup[\gamma,1],
 \label{eq:singular-gap-promise}
\end{equation}
where $\operatorname{sing}(A)$ denotes the set of singular values
of $A$.
\end{assumption}

One block-encoding query is an application of $U_B$ or
$U_B^\dagger$, including a controlled application when required.
Calls to $P_\Phi$ and to any underlying input oracle are counted
separately. The state $\ket{\Phi_D}$ purifies the maximally mixed
state $P_R/D$ on the embedded input space.

The condition $\alpha\geq1$ ensures that $\sqrt{\log\alpha^2}$
is real. An encoding with normalization $\alpha_0<1$ can be
rescaled to normalization $1$ using one additional ancilla
rotation, without increasing the number of queries. A supplied
normalized gap $\gamma_0$ then becomes $\alpha_0\gamma_0$.

\subsection{Polynomial approximation}

For a positive normalized singular value $x=\sigma/\alpha$,
\[
 \log(1/x^2)=\log\alpha^2-\log\sigma^2.
\]
Assigning the weight $\log\alpha^2$ also to each vector in an
orthonormal basis of $\ker B$ makes the sum of these weights equal
to $D\log\alpha^2-\ell(B)$. The circuit will approximate the square
roots of these weights.

Set
\begin{equation}
 a_\alpha:=\log\alpha^2,\qquad
 s_{\alpha,\gamma}:=a_\alpha+\log(1/\gamma),\qquad
 b_\alpha:=\sqrt{a_\alpha}.
 \label{eq:complementary-normalization}
\end{equation}
If $s_{\alpha,\gamma}=0$, then $\alpha=\gamma=1$ and
$\ell(B)=0$, so no queries are required.
Otherwise, set
\begin{equation}
 \Lambda:=160\sqrt{s_{\alpha,\gamma}}.
 \label{eq:weighted-normalization}
\end{equation}
Here $\alpha$ is the block-encoding normalization, whereas
$\Lambda$ is the polynomial scaling factor.

The bottom-left QSVT block supplies the factor $\sqrt{1-x^2}$.
We therefore approximate
\begin{equation}
 h(x):=\sqrt{\frac{\log(1/x^2)}{1-x^2}}
 \quad(0<|x|<1),\qquad h(\pm1):=1,
 \label{eq:complementary-log-target}
\end{equation}
with endpoint values defined by continuity.
At zero, the polynomial instead approximates $b_\alpha$.
The singular-value gap separates these approximation requirements.

\begin{lemma}[Weighted polynomial approximation]
\label{lem:complementary-log-polynomial}
For $\alpha\geq1$, $0<\gamma\leq1$,
$s_{\alpha,\gamma}>0$, and $0<\xi<1/4$, there is a real even
polynomial $p$ satisfying
\begin{align}
 M(p):=\max_{x\in[-1,1]}\sqrt{1-x^2}|p(x)|
 &\leq40\sqrt{s_{\alpha,\gamma}},
 \label{eq:complementary-poly-bound}\\
 |p(x)-h(x)|&\leq\xi
 &&(\gamma\leq|x|\leq1),
 \label{eq:complementary-poly-outer}\\
 |p(x)-b_\alpha|&\leq\xi
 &&(|x|\leq\gamma/2).
 \label{eq:complementary-poly-kernel}
\end{align}
Its degree is at most an even integer $d_q$ with
\begin{equation}
 d_q=O\!\left[
 (\gamma^{-1}+s_{\alpha,\gamma}^{-1/2})
 \log\frac{C(1+\sqrt{s_{\alpha,\gamma}})}{\gamma\xi}
 \right],
 \label{eq:complementary-degree}
\end{equation}
where $C>1$ and the implied constant are universal.
Thus $q:=p/\Lambda$ satisfies $M(q)\leq1/4$.
\end{lemma}

The proof is given in \cref{app:complementary-polynomial-proof}.
The bound applies to $\sqrt{1-x^2}|q(x)|$, rather than to
$|q(x)|$ alone. This is the condition required for the bottom-left
QSVT construction. The necessity of the additional
$s_{\alpha,\gamma}^{-1/2}$ degree term for the chosen scaling is
established in \cref{prop:weighted-degree-normalization}.

\subsection{Bottom left QSVT construction}

\begin{lemma}[Implementation of the bottom-left QSVT block]
\label{lem:bottom-left-qsvt}
Let $q$ be a real even polynomial of degree at most an even
integer $d_q$, satisfying
\[
 (1-x^2)q(x)^2\leq1
 \qquad(x\in[-1,1]).
\]
There is a unitary circuit $\mathcal V_q$ using one control
ancilla and $n_{\rm Q}=d_q+1$ block-encoding queries such that,
on a right singular vector of $A$ with singular value $x$, the
success probability is
\[
 (1-x^2)q(x)^2.
\]
The control ancilla and block-encoding ancillas are initialized
to zero. Success occurs when the control ancilla is measured
in $\ket0$ and the block-encoding ancillas are not all zero.
The statement includes $x=0$ and $x=1$.
\end{lemma}

\begin{proof}[Proof sketch]
Quantum signal processing (QSP) polynomial completion gives a
complex polynomial $\mathcal Q$ with
$\operatorname{Re}\mathcal Q=q$.
A control ancilla coherently combines the QSP sequences with
opposite phase lists, producing the bottom-left amplitude
$\mathrm{i}\sqrt{1-x^2}q(x)$.
The two sequences have identical block-encoding calls, so the
control acts only on the phase rotations and does not double the
query count. The complete construction, including the zero
singular values, is proved in \cref{app:bottom-left-qsvt-proof}.
\end{proof}

Prepare $\ket{\Phi_D}$ and apply $\mathcal V_q$ to its first
register, with the control ancilla and block-encoding ancillas
initialized to zero. Let $x_1,\ldots,x_D$ be the nonnegative square
roots of the eigenvalues of $B^{\mathsf T}B/\alpha^2$, counted
with multiplicity. The success probability is
\begin{equation}
 P_q
 =\frac1D\sum_{j=1}^{D}(1-x_j^2)q(x_j)^2
 =\frac1{D\Lambda^2}
   \sum_{j=1}^{D}(1-x_j^2)p(x_j)^2.
 \label{eq:complementary-probability}
\end{equation}
This sum includes all $D-\rank B$ zero singular values associated
with $\ker B$.
Write $\mathcal A_q$ for the unitary circuit consisting of this
state preparation and QSVT construction. Each application of
$\mathcal A_q$ or its inverse uses one state-preparation call and
$d_q+1$ block-encoding queries.

Define the ideal weight
\begin{equation}
 w_\alpha(x):=
 \begin{cases}
  \log(1/x^2),&\gamma\leq x\leq1,\\
  a_\alpha,&x=0.
 \end{cases}
 \label{eq:complementary-weight}
\end{equation}
Then
\begin{equation}
 \begin{split}
 S(B)&:=\frac1D\sum_{j=1}^{D}w_\alpha(x_j)
      =a_\alpha-\frac{\ell(B)}D,\\
 P_{\rm ideal}&:=\frac{S(B)}{\Lambda^2},\\
 0\leq S(B)&\leq S_{\max}
 :=\max\{a_\alpha,2\log(1/\gamma)\}
 \leq2s_{\alpha,\gamma}.
 \end{split}
 \label{eq:complementary-trace-identity}
\end{equation}
To verify the identity, let $\rho=\rank B$.
The positive singular values contribute
$\rho a_\alpha-\ell(B)$, while the kernel contributes
$(D-\rho)a_\alpha$. Thus the rank cancels.
The identity also holds for $B=0$, with $\pdet(0)=1$.

For boundary matrices with integer entries, integrality gives
a smaller upper bound on $S(B)$.

\begin{lemma}[Bound on the normalized sum for integer boundary matrices]
\label{lem:integer-trace-bound}
If $B$ has integer entries and $\|B\|\leq\alpha$ with $\alpha\geq1$,
then
\begin{equation}
 0\leq\ell(B)\leq D a_\alpha,\qquad
 0\leq S(B)\leq a_\alpha.
 \label{eq:integer-trace-bound}
\end{equation}
\end{lemma}

\begin{proof}
For $\rho=\rank B>0$, Cauchy--Binet gives
\[
 \pdet(BB^{\mathsf T})
 =\sum_{|I|=|J|=\rho}\det(B_{I,J})^2,
\]
where $I,J$ are row and column index sets.
At least one minor is a nonzero integer, so $\ell(B)\geq0$.
Since every positive singular value is at most $\alpha$,
$\ell(B)\leq\rho\log\alpha^2\leq D a_\alpha$.
The assertion for $S(B)$ follows from
\eqref{eq:complementary-trace-identity}.
The zero matrix satisfies the same bounds by convention.
\end{proof}

\subsection{Amplitude estimation and query complexity}

We estimate $P_q$ by applying amplitude estimation to
$\mathcal A_q$ and its inverse. This use of state preparation
and amplitude estimation to determine a normalized trace follows
Gribling~\cite[Lemmas~9.3 and~11.9]{Gribling2019Thesis}.
The success-probability dependence is retained in the complexity
bound below.

\begin{theorem}[Quantum estimation of upper Hodge Laplacian
log pseudodeterminants]
\label{thm:quantum-pseudologdet}
Assume \cref{ass:block-trace}, $s_{\alpha,\gamma}>0$, and a
supplied bound $S(B)\leq S_\ast\leq2s_{\alpha,\gamma}$.
For $0<\varepsilon_{\mathrm n}\leq1$ and $0<\nu<1$,
there is a quantum algorithm returning $\widehat\ell_{\mathrm n}$
such that
\begin{equation}
 \Pr\!\left[
 \left|\widehat\ell_{\mathrm n}-\frac{\ell(B)}D\right|
 \leq\varepsilon_{\mathrm n}
 \right]\geq1-\nu,
 \label{eq:normalized-success}
\end{equation}
using
\begin{equation}
 \widetilde O\!\left[
 \left(\frac{\sqrt{s_{\alpha,\gamma}}}{\gamma}+1\right)
 \left(
 \frac{\sqrt{S_\ast}}{\varepsilon_{\mathrm n}}
 +\frac1{\sqrt{\varepsilon_{\mathrm n}}}
 \right)\log\frac2\nu
 \right]
 \label{eq:normalized-block-queries}
\end{equation}
block-encoding queries.
\end{theorem}

Here $\widetilde O$ suppresses logarithmic factors in $1/\gamma$
and $(1+s_{\alpha,\gamma})/\varepsilon_{\mathrm n}$.
One may take $S_\ast=S_{\max}$, or $S_\ast=a_\alpha$ for
integer $B$ by \cref{lem:integer-trace-bound}.
The cost of obtaining any separately supplied bound $S_\ast$
is excluded.

\begin{proof}
If $S_\ast\leq\varepsilon_{\mathrm n}/2$, return $a_\alpha$
to arithmetic error at most $\varepsilon_{\mathrm n}/2$.
Since $\ell(B)/D=a_\alpha-S(B)$ and $0\leq S(B)\leq S_\ast$,
this achieves the required accuracy without block-encoding queries.

Otherwise, $\varepsilon_{\mathrm n}<2S_\ast\leq4s_{\alpha,\gamma}$.
The algorithm estimates the success probability of $\mathcal A_q$
and evaluates
\begin{equation}
 \widehat\ell_{\mathrm n}
 :=a_\alpha-\Lambda^2\widehat P.
 \label{eq:normalized-output-definition}
\end{equation}
Let $P_{\rm impl}$ be the success probability of a fixed
implemented circuit satisfying
$|P_{\rm impl}-P_q|\leq\delta_{\rm impl}$.
Amplitude estimation uses this circuit and its inverse.

Before division by $\Lambda$, the ideal success amplitudes have
absolute value at most $\sqrt{2s_{\alpha,\gamma}}$, and the
polynomial approximation changes each amplitude by at most $\xi$.
Consequently, on the event
$|\widehat P-P_{\rm impl}|\leq\delta_{\rm AE}$,
\begin{equation}
 \left|\widehat\ell_{\mathrm n}-\frac{\ell(B)}D\right|
 \leq
 2\sqrt{2s_{\alpha,\gamma}}\xi+\xi^2
 +\Lambda^2(\delta_{\rm impl}+\delta_{\rm AE})+\chi,
 \label{eq:complete-error-bound}
\end{equation}
where $\chi$ bounds the arithmetic error in
\eqref{eq:normalized-output-definition}.
Choose
\begin{equation}
 \xi=\frac{\varepsilon_{\mathrm n}}
          {16\sqrt{2s_{\alpha,\gamma}}},
 \qquad
 \delta_{\rm impl}\leq
       \frac{2\varepsilon_{\mathrm n}}{5\Lambda^2},
 \qquad
 \delta_{\rm AE}\leq
       \frac{\varepsilon_{\mathrm n}}{5\Lambda^2},
 \qquad
 \chi\leq\frac{\varepsilon_{\mathrm n}}5.
 \label{eq:algorithm-error-budget}
\end{equation}
The inequalities $\varepsilon_{\mathrm n}<4s_{\alpha,\gamma}$
and $\varepsilon_{\mathrm n}\leq1$ imply $\xi<1/4$ and
\[
 2\sqrt{2s_{\alpha,\gamma}}\xi+\xi^2
 <
 \frac{\varepsilon_{\mathrm n}}8
 +\frac{\varepsilon_{\mathrm n}}{128}
 <
 \frac{\varepsilon_{\mathrm n}}5.
\]
Thus \eqref{eq:complete-error-bound} is at most
$\varepsilon_{\mathrm n}$.
The implementation tolerances sufficient for
\eqref{eq:algorithm-error-budget} are given in
\cref{subsec:implementation-errors}.

With this choice of $\xi$,
\cref{lem:complementary-log-polynomial} supplies an even polynomial
with an even degree bound
\begin{equation}
 d_q=O\!\left[
 (\gamma^{-1}+s_{\alpha,\gamma}^{-1/2})
 \log\frac{C(1+s_{\alpha,\gamma})}
          {\gamma\varepsilon_{\mathrm n}}
 \right],
 \label{eq:estimator-degree}
\end{equation}
where $C>0$ is an absolute constant.

To bound the amplitude-estimation cost, the same error estimates give
\begin{equation}
 |P_{\rm impl}-P_{\rm ideal}|
 \leq\frac{3\varepsilon_{\mathrm n}}{5\Lambda^2},
 \qquad
 P_{\rm impl}\leq
 \frac{S_\ast+3\varepsilon_{\mathrm n}/5}{\Lambda^2}.
 \label{eq:implemented-probability-bound}
\end{equation}
Amplitude estimation with parameter $m$ returns an estimate satisfying
\begin{equation}
 |\widehat P-P_{\rm impl}|
 \leq
 \frac{2\pi\sqrt{P_{\rm impl}(1-P_{\rm impl})}}m
 +\frac{\pi^2}{m^2}
 \label{eq:probability-dependent-AE}
\end{equation}
with probability at least $8/\pi^2$, using $O(m)$ applications
of the circuit and its inverse
\cite[Theorem~12]{BrassardEtAl2002Amplitude}.
By \eqref{eq:implemented-probability-bound}, choosing
\[
 m=\Theta\!\left(
 \frac{\Lambda\sqrt{S_\ast+\varepsilon_{\mathrm n}}}
      {\varepsilon_{\mathrm n}}
 +\frac{\Lambda}{\sqrt{\varepsilon_{\mathrm n}}}
 \right)
\]
with a sufficiently large constant achieves the required
$\delta_{\rm AE}$.
Taking the median of $O(\log(2/\nu))$ independent estimates
reduces the failure probability to at most $\nu$.
Since $S_\ast>\varepsilon_{\mathrm n}/2$, the total number of
applications of $\mathcal A_q$ or its inverse is
\begin{equation}
 M_{\rm AE}
 =O\!\left(
 \left[
 \frac{\Lambda\sqrt{S_\ast}}{\varepsilon_{\mathrm n}}
 +\frac{\Lambda}{\sqrt{\varepsilon_{\mathrm n}}}
 \right]\log\frac2\nu
 \right).
 \label{eq:normalized-state-queries}
\end{equation}
Together with \eqref{eq:complete-error-bound}, this proves
\eqref{eq:normalized-success}.
Each application uses $d_q+1$ block-encoding queries.
Multiplying \eqref{eq:normalized-state-queries} by $d_q+1$
and using $\Lambda=160\sqrt{s_{\alpha,\gamma}}$ gives
\eqref{eq:normalized-block-queries}.
\end{proof}

The second term in \eqref{eq:normalized-state-queries} accounts
for the $m^{-2}$ term in amplitude estimation.
When $S_\ast>\varepsilon_{\mathrm n}/2$, it is at most a constant
multiple of the first term.
The query count depends on the supplied bound $S_\ast$;
using the unknown value $S(B)$ would require an adaptive algorithm.

\begin{corollary}[Additive error in the unnormalized log pseudodeterminant]
\label{cor:unnormalized-pseudologdet}
Under the assumptions of \cref{thm:quantum-pseudologdet},
additive error $0<\varepsilon\leq D$ in $\ell(B)$ and failure
probability at most $\nu$ can be achieved with
\begin{equation}
 \widetilde O\!\left[
 \left(\frac{\sqrt{s_{\alpha,\gamma}}}{\gamma}+1\right)
 \left(
 \frac{D\sqrt{S_\ast}}{\varepsilon}
 +\sqrt{\frac D\varepsilon}
 \right)\log\frac2\nu
 \right]
 \label{eq:unnormalized-block-queries}
\end{equation}
block-encoding queries.
\end{corollary}

\begin{proof}
Set $\varepsilon_{\mathrm n}=\varepsilon/D$ and return
$D\widehat\ell_{\mathrm n}$, including this multiplication
in the arithmetic error.
\end{proof}

For gate complexity, let $G_B$ bound the cost of a controlled
block-encoding query, $G_\Pi$ the cost of a reflection about
the ancilla-zero subspace, and $G_{\rm phase}$ the cost of a
synthesized QSVT phase rotation. Let $G_{\rm AE}$ be the
amortized additional cost per application of $\mathcal A_q$
or its inverse, including state preparation, the reflections
and controls used by amplitude estimation, and its final processing.
The quantum gate complexity is
\begin{equation}
 O\!\left(
 M_{\rm AE}
 \left[
 G_{\rm AE}
 +(d_q+1)(G_B+G_\Pi+G_{\rm phase})
 \right]
 \right).
 \label{eq:normalized-gates}
\end{equation}
All costs are evaluated at the precisions specified in
\cref{subsec:implementation-errors}.
Classical computation of polynomial coefficients and QSVT phases
is not included in this quantum gate count.

\subsection{Clique complexes case}

Let $G=([N],E)$ be a finite simple undirected graph, where
$[N]=\{1,\ldots,N\}$. Its clique complex $X=\Cl(G)$ consists
of all vertex subsets whose distinct vertices are pairwise
connected by edges:
\[
 \Cl(G)
 =
 \bigl\{
 \sigma\subseteq[N]:
 \{u,v\}\in E
 \text{ for all distinct }u,v\in\sigma
 \bigr\}.
\]
Every subset of a clique is again a clique, so this collection
is a simplicial complex. It includes the empty face and all
single vertices. An $r$-dimensional face of $X$ is therefore
a clique on $r+1$ vertices. Under the direct mapping described in
\cref{subsec:simplex-encoding}, a computational basis state of Hamming weight $r+1$
represents such a face exactly when its selected vertices
form a clique in $G$ with
$B=\mathcal B_r$ and $D=D_r$.
If $\partial_r^{\mathrm{full}}$ is the degree-$r$ boundary
of the full simplex on $[N]$, extended by zero on other
Hamming weights, then
\begin{equation}
 \mathcal B_r
 =\Pi_{r-1}\partial_r^{\mathrm{full}}\Pi_r.
 \label{eq:padded-boundary}
\end{equation}

The required initial state is
\begin{equation}
 \ket{\Phi_{D_r}}
 =\frac1{\sqrt{D_r}}
  \sum_{\substack{\sigma\subseteq[N]\\|\sigma|=r+1}}
  \ket{\sigma}\ket{\sigma}.
 \label{eq:candidate-trace-state}
\end{equation}
It can be prepared from a Dicke state of Hamming weight $r+1$,
followed by $N$ controlled-NOT gates to a zero-initialized
second register~\cite{BartschiEidenbenz2019Dicke}.
Thus
\begin{equation}
 G_{\Phi,r}
 =G_{\mathrm{Dicke}}(N,r+1)+O(N),
 \label{eq:dicke-trace-cost}
\end{equation}
where $G_{\mathrm{Dicke}}$ is the Dicke-state preparation
gate count, with rotation-synthesis costs included at the
required accuracy.
The state contains all $(r+1)$-vertex subsets.
Conditioning on membership in $X$ would instead normalize by
$f_r$, with success probability
\begin{equation}
 \frac{f_r}{D_r}
 \label{eq:simplex-postselection}
\end{equation}
before amplification.

McArdle, Gily\'en, and Berta give an exact
$(\sqrt N,2,0)$ block encoding using one membership-oracle
call in degree $r$ and one in degree $r-1$
\cite[Appendix~C.2.2]{McArdle2022AQubits}.
The membership oracles check the required Hamming weight and
simplex membership, and uncompute their work registers.
Consequently,
\begin{equation}
 \alpha_r=\sqrt N,\qquad
 0<\gamma_r\leq
 \frac{\sigma_{\min}^{+}(B_r)}{\sqrt N},
 \label{eq:mcardle-normalization}
\end{equation}
where $\gamma_r$ is a supplied lower bound.
The two block-encoding ancillas are separate from membership
workspace. The circuit construction is described in
\cref{app:implementation-details}.

For an adjacency oracle
\begin{equation}
 O_E\ket{u,v,b}
 =\ket{u,v,b\oplus[\{u,v\}\in E]},
 \label{eq:edge-oracle}
\end{equation}
where the bracket is the edge indicator and $\oplus$ is
addition modulo two, a degree-$r$ membership test uses
$O(r^2)$ adjacency queries for $r\geq1$.
Writing $G_{\mathrm{mem},r}$ for the gate complexity of the
controlled membership oracle, the block-encoding gate complexity is
\begin{equation}
 G_{B,r}
 =\widetilde O\!\left(
 N+G_{\mathrm{mem},r}+G_{\mathrm{mem},r-1}
 \right).
 \label{eq:boundary-gate-cost}
\end{equation}
This includes the membership circuits and synthesis of the
$O(N)$ elementary rotations in the boundary construction.
For this gate bound, $\widetilde O$ also suppresses logarithmic
synthesis factors.

For $N\geq2$, the algorithm parameters are
\begin{equation}
 a_r=\log N,\qquad
 s_r=\log N+\log(1/\gamma_r),\qquad
 \Lambda_r=160\sqrt{s_r},\qquad
 S_{\ast,r}=\log N.
 \label{eq:clique-estimator-parameters}
\end{equation}
The last equality uses \cref{lem:integer-trace-bound}.
Since $s_r\geq\log2$, the polynomial degree is
$\widetilde O(\gamma_r^{-1})$.
For $\varepsilon_{\mathrm n}\leq1$, the block-encoding
query complexity simplifies to
\[
 \widetilde O\!\left(
 \frac{\Lambda_r\sqrt{\log N}}
      {\gamma_r\varepsilon_{\mathrm n}}
 \log\frac2{\nu_r}
 \right),
\]
where $\nu_r$ is the failure probability assigned to degree $r$.

\begin{corollary}[Quantum algorithm for estimating the logarithm
of a critical-group order]
\label{cor:quantum-critical-order}
Let $X=\Cl(G)$, let $0\leq i<\dim X$, and fix
$\varepsilon>0$ and $0<\nu<1$.
Suppose $X$ satisfies \cref{cor:root-free-critical} and
$\wtH_q(X;\Z)=0$ for $0\leq q\leq i-1$.
Assume \cref{ass:block-trace} for each
$\mathcal B_r$, $1\leq r\leq i+1$, with
$D_r=\binom{N}{r+1}$, $\alpha_r=\sqrt N$, and supplied gaps
$\gamma_r$ satisfying \eqref{eq:mcardle-normalization}.

Choose $0<\varepsilon_{\rm arith}<\varepsilon$,
$0<\varepsilon_r\leq D_r$, and $0<\nu_r<1$ such that
\begin{equation}
 \varepsilon_{\rm arith}
 +\sum_{r=1}^{i+1}\varepsilon_r\leq\varepsilon,
 \qquad
 \sum_{r=1}^{i+1}\nu_r\leq\nu.
 \label{eq:error-allocation}
\end{equation}
Estimate each $\ell_r(X)$ to additive error $\varepsilon_r$,
evaluate $\ell_0(X)=\log N$, and combine the results using
\eqref{eq:critical-pdet-boxed}.
If the evaluation of $\ell_0$ and the final sum has error at most
$\varepsilon_{\rm arith}$, the result estimates $\log|K_i(X)|$
to additive error $\varepsilon$ with probability at least $1-\nu$.
The block-encoding query complexity is
\begin{equation}
 \widetilde O\!\left(
 \sum_{r=1}^{i+1}
 \frac{D_r\Lambda_r\sqrt{\log N}}
      {\gamma_r\varepsilon_r}
 \log\frac2{\nu_r}
 \right).
 \label{eq:critical-query-complexity}
\end{equation}
\end{corollary}

\begin{proof}
Apply \cref{cor:unnormalized-pseudologdet} in each degree.
Because $\varepsilon_r/D_r\leq1$ and $\log N\geq\log2$,
the second amplitude-estimation term is bounded by a constant
multiple of the first. The triangle inequality gives the
total additive error, and the union bound gives the failure
probability.
\end{proof}

The implementation tolerances in degree $r$ use
$\varepsilon_{\mathrm n}=\varepsilon_r/D_r$.
In particular, its normalized arithmetic error satisfies
$\chi_r\leq\varepsilon_r/(5D_r)$.
The additional error  $\varepsilon_{\rm arith}$ applies
to the final classical postprocessing.
\Cref{app:error-allocation} optimizes the error allocation and
includes conversion to adjacency-query complexity.

\subsection{Summary of resource requirements}
\label{subsec:resource-ledger}

The resource bounds distinguish block-encoding queries,
state-preparation calls, adjacency queries, and quantum gates.
They assume the stated coherent oracle access and supplied
singular-value gaps. See \cref{tab:resource-ledger} in the appendix
for a summary of the parameters and implementation tolerances.

Normalization by $f_r$ instead of $D_r$ requires preparing
a purification of the maximally mixed state on the simplex
subspace, including its preparation cost.
Independent repetitions and computations for different boundary
matrices can be executed in parallel when separate processors
and concurrent access to independent oracle instances are
available. This can reduce execution time without reducing
the total number of queries or gates.

\section{Multipartite and torsion-bearing benchmarks}
\label{sec:multipartite}
 
The cost of the estimator in \cref{sec:quantum-algorithm} is controlled by
four inputs for each boundary map: the block-encoding normalization $\alpha_r$,
the normalized gap $\gamma_r$, the candidate dimension $D_r$, and a supplied
upper bound $S_{\ast,r}$ on the complementary normalized sum.  For integral
clique boundaries, \cref{lem:integer-trace-bound} permits
$S_{\ast,r}=\log N$.  A small gap directly inflates the degree of the QSVT
polynomial.  In this section we construct clique complexes for which all four
inputs are under explicit control and for which the torsion, the quantity we
are ultimately after, is known exactly.  They serve as known-answer benchmarks,
and one of them also supports a provable query separation.
 
We proceed in three steps.  First, in \cref{subsec:multipartite-basic}, we
compute every boundary spectrum of the balanced multipartite clique complex
$X_{m,k}$ in closed form.  After the $\sqrt N$ normalization used by the
clique-complex block encoding, every boundary map has norm one and condition
number at most $\sqrt k$.  This family has no torsion.  Second, in
\cref{subsec:torsion-clique-benchmark}, we plant torsion by joining $X_{m,k}$
with a fixed $31$-vertex clique complex that triangulates the real projective
plane. The join keeps a normalized gap of at least $\sqrt{c_Y/N}$ for a fixed
constant $c_Y>0$, and its homology torsion has order $2^{(m-1)^k}$; when
$m=2^k$, the estimator recovers the logarithm of this order to relative error
with polynomially many queries.  Third, in
\cref{subsec:matched-query-separation}, we glue $n$ small blocks at a single
vertex, each of which is a projective plane or a contractible tree according
to one bit of an input string $x$.  The per-simplex log pseudodeterminant of
the result is proportional to the Hamming weight of $x$, which turns our
estimation problem into approximate counting and yields a quadratic query
separation.
\subsection{Balanced multipartite complex}
\label{subsec:multipartite-basic}

Let $m,k\geq2$, and let $G_{m,k}=K_{m,\ldots,m}$ be the complete balanced
$k$-partite graph, with $k$ parts of size $m$.  Its clique complex is
\begin{equation}
  X_{m,k}:=\Cl(G_{m,k})=\chg{\overline K_m}^{\ast k},
  \qquad N_X:=mk,
  \label{eq:multipartite-join}
\end{equation}
where $\chg{\overline K_m}$ is the discrete complex on $m$ vertices and $\ast$ is the
simplicial join.  Thus a face chooses at most one vertex from each part and
$\dim X_{m,k}=k-1$.  The join K\"unneth formula gives
\begin{equation}
  \wtH_q(X_{m,k};\Z)\cong
  \begin{cases}
    \Z^{(m-1)^k},&q=k-1,\\
    0,&q\neq k-1,
  \end{cases}
  \label{eq:multipartite-homology}
\end{equation}
so this first family has no homology torsion.  Its role is to provide exact,
well-conditioned spectral data; \cref{subsec:torsion-clique-benchmark} adds
torsion.

Complete colorful (multipartite clique) complexes were analyzed in all degrees in~\cite[Section~6.1, Equations~(30)--(33)]{DuvalKlivansMartin2011Cellular}; the
balanced top-degree case is also used in~\cite[Appendix~F]{BerryEtAl2024QTDA}.
We record the balanced specialization and the boundary-map consequences needed
by the algorithm.  Here $\lambda_{\min}^{+}$ and $\chg{\zeta_{\min}^{+}}$ denote the
smallest positive eigenvalue and singular value, respectively, and
$\kappa(M)=\|M\|/\chg{\zeta_{\min}^{+}}(M)$ on the nonzero singular subspace.

\begin{theorem}[Spectra and conditioning]
\label{thm:multipartite-hodge-spectrum}
For $0\leq j\leq k-1$,
\begin{equation}
  \Spec\!\left(\Delta_j^{X_{m,k}}\right)
  =\{(k-j-1)m,(k-j)m,\ldots,km\},
  \label{eq:multipartite-hodge-spectrum}
\end{equation}
where $(s+k-j-1)m$ has multiplicity
\begin{equation}
  \binom{k}{j+1}\binom{j+1}{s}(m-1)^{j+1-s},
  \qquad 0\leq s\leq j+1.
  \label{eq:multipartite-hodge-multiplicity}
\end{equation}
Consequently,
\begin{equation}
  \wtbeta_j(X_{m,k})=0\ (0\leq j\leq k-2),\qquad
  \wtbeta_{k-1}(X_{m,k})=(m-1)^k,
  \label{eq:multipartite-betti}
\end{equation}
and
\begin{equation}
  \lambda_{\min}^{+}(\Delta_j)
  =\max\{k-j-1,1\}m\geq m=\frac{N_X}{k}.
  \label{eq:multipartite-hodge-gap}
\end{equation}
For $0\leq r\leq k-1$, the positive squared singular values of $B_r$ are
$(k-a)m$, $0\leq a\leq r$, with multiplicities
\begin{equation}
  \chg{\mu_{r,a}}
  =\binom{k}{a}(m-1)^a\binom{k-a-1}{r-a}.
  \label{eq:boundary-multiplicity}
\end{equation}
Hence
\begin{equation}
  \left\|\frac{B_r}{\sqrt{N_X}}\right\|=1,
  \qquad
  \chg{\zeta_{\min}^{+}}\!\left(\frac{B_r}{\sqrt{N_X}}\right)
  =\sqrt{\frac{k-r}{k}},
  \qquad
  \kappa(B_r)=\sqrt{\frac{k}{k-r}}.
  \label{eq:boundary-conditioning}
\end{equation}
\end{theorem}

\begin{proof}
Equations \eqref{eq:multipartite-hodge-spectrum} and
\eqref{eq:multipartite-hodge-multiplicity} are the balanced specialization of
the complete-colorful spectrum in~\cite[Equation~(32)]{DuvalKlivansMartin2011Cellular}, grouped by the number
$s$ of constant directions.  Equations \eqref{eq:multipartite-betti} and
\eqref{eq:multipartite-hodge-gap} follow immediately.

Under the shifted tensor-product model for a join, write
$\R^m=\R u\oplus W$, where
$u=m^{-1/2}(1,\ldots,1)$, $\dim W=m-1$, $B_0u=\sqrt m$, and $B_0W=0$.
Fixing $a$ factors in $W$ leaves $q=k-a$ copies of the two-term complex
$\R u\xrightarrow{\sqrt m}\R$.  In degree $r$, the remaining boundary is
contraction by $\sqrt m(1,\ldots,1)\in\R^q$ on
$\Lambda^{r+1-a}\R^q$.  Its nonzero singular value is $\sqrt{(k-a)m}$ with
multiplicity $\binom{k-a-1}{r-a}$.  Choosing the $a$ factors and their
$W$-directions gives \eqref{eq:boundary-multiplicity}; the extremal values
$a=0,r$ give \eqref{eq:boundary-conditioning}.
\end{proof}

The face data and normalized gap are summarized by
\begin{equation}
  \sum_{t=0}^{k}f_{t-1}x^t=(1+mx)^k,
  \qquad
  f_j=\binom{k}{j+1}m^{j+1},
  \qquad
  \lambda_{\min}^{+}\!\left(\frac{\Delta_j}{N_X}\right)\geq\frac1k.
  \label{eq:multipartite-face-counts}
\end{equation}
Thus $X_{m,k}$ has $(m+1)^k$ faces including the empty face, while its exact
boundary log pseudodeterminants are
\begin{equation}
  \ell_r(X_{m,k})
  =\sum_{a=0}^{r}\chg{\mu_{r,a}}\log\!\bigl((k-a)m\bigr).
  \label{eq:multipartite-exact-logdet}
\end{equation}
Moreover, $X_{m,k}^{(i)}$ is the independence complex of a rank-$(i+1)$
truncation of a partition matroid and is therefore shellable~\cite{Bjorner1992MatroidShellability}. The discussion following
Corollary~4.4 of~\cite{dmtcs:2909} then supplies the required torsion-free
spanning trees. With~\eqref{eq:multipartite-homology}, this verifies the
hypotheses of~\eqref{eq:critical-equals-tree} for $0\leq i\leq k-2$ and gives
\begin{equation}
  \log|K_i(X_{m,k})|
  =\sum_{r=0}^{i+1}(-1)^{i+1-r}
    \sum_{a=0}^{r}\chg{\mu_{r,a}}\log\!\bigl((k-a)m\bigr).
  \label{eq:multipartite-critical-order}
\end{equation}

\subsection{A torsion-bearing benchmark}
\label{subsec:torsion-clique-benchmark}

Let $T$ be the explicit projective-plane tree in \cref{ex:rp2-tree}, and let
$Y:=\operatorname{sd}T$ be its barycentric subdivision.  Its $31$ vertices are
the nonempty faces of $T$, and its simplices are chains of such faces.  Hence
$Y$ is the clique complex of the comparability graph $\Gamma_T$, in which two
faces are adjacent when one contains the other.  Define
\begin{equation}
  Z_{m,k}:=X_{m,k}\ast Y
  =\Cl\!\left(G_{m,k}\vee\Gamma_T\right),
  \qquad N_Z:=mk+31,
  \qquad d:=\dim Z_{m,k}=k+2,
  \label{eq:torsion-join-benchmark}
\end{equation}
where the graph join $\vee$ adds every edge between the two sets of vertices.  The
construction inherits complementary features from its two factors, as made
precise in the following proposition.

\begin{proposition}[Homology torsion and boundary gap]
\label{prop:torsion-join-benchmark}
For $m,k\geq2$,
\begin{equation}
  \wtH_q(Z_{m,k};\Z)
  \cong
  \begin{cases}
    (\Z/2\Z)^{(m-1)^k},&q=k+1,\\
    0,&q\neq k+1.
  \end{cases}
  \label{eq:torsion-join-homology}
\end{equation}
Consequently, the full complex $Z_{m,k}$ satisfies the certified-tree criterion
with ambient complex $X=Z_{m,k}$, and
\begin{equation}
  \log\left|\wtH_{d-1}(Z_{m,k};\Z)\right|
  =(m-1)^k\log2
  =\frac12\sum_{r=0}^{d}(-1)^{d-r}\ell_r(Z_{m,k}).
  \label{eq:torsion-join-log-order}
\end{equation}
There is also a fixed, precomputable constant $c_Y>0$ such that
\begin{equation}
  \chg{\zeta_{\min}^{+}}\!\left(
    \frac{B_r(Z_{m,k})}{\sqrt{N_Z}}
  \right)
  \geq\sqrt{\frac{c_Y}{N_Z}},
  \qquad 0\leq r\leq d.
  \label{eq:torsion-join-boundary-gap}
\end{equation}
\end{proposition}

\begin{proof}
The only nonzero reduced homology groups of the factors are
$\wtH_{k-1}(X_{m,k};\Z)\cong\Z^{(m-1)^k}$ and
$\wtH_1(Y;\Z)\cong\Z/2\Z$.  The integral join K\"unneth formula gives
\eqref{eq:torsion-join-homology}; its Tor terms vanish because the homology of
$X_{m,k}$ is free.  The join is pure and rationally acyclic, so its top
boundary is injective over $\Q$ and spans the rational $(d-1)$-cycle space.
Thus \cref{prop:minimal-filling} applies with the ambient complex equal to
$Z_{m,k}$, and \eqref{eq:certified-tree-root-free} gives
\eqref{eq:torsion-join-log-order}.

Since $Y$ is fixed and rationally acyclic, its augmented Hodge Laplacians are
positive definite.  Set
\begin{equation}
  c_Y:=\min\!\left\{1,\lambda_{\min}(\Delta_q^Y):-1\leq q\leq2\right\}>0.
  \label{eq:cY-definition}
\end{equation}
By the augmented join-spectrum formula~\cite[Corollary~1]{BerryEtAl2024QTDA}, each Hodge eigenvalue of $Z_{m,k}$ is a
sum containing an eigenvalue of $Y$, and is therefore at least $c_Y$.  The
orthogonal Hodge splitting places every positive squared singular value of
$B_r$ in the spectrum of $\Delta_{r-1}$; division by $N_Z$ proves
\eqref{eq:torsion-join-boundary-gap}.
\end{proof}

The explicit certified bound $c_Y\geq31^{-89}$ is proved in
\cref{app:join-resource-proof}.

Finally, take $m=2^k$ and $k\geq4$, and write
$L_{m,k}:=(m-1)^k\log2$ for the log torsion order in
\eqref{eq:torsion-join-log-order}.  For
$D_r=\binom{N_Z}{r+1}$, the inequality $k+3<N_Z/2$ makes $D_d$ the largest
relevant candidate dimension, and
\begin{equation}
  \frac{\max_{0\leq r\leq d}D_r}{L_{m,k}}
  =\frac{\binom{N_Z}{k+3}}{(m-1)^k\log2}
  \leq\frac{e^3}{\log2}\,m^3(2e)^k
  =N_Z^{O(1)}.
  \label{eq:torsion-join-candidate-ratio}
\end{equation}
Here $N_Z/(k+3)\leq m$, $\binom ns\leq(en/s)^s$, and
$m/(m-1)\leq2$.  To estimate $L_{m,k}$ to relative error
$0<\delta\leq1$, it suffices to use the exact value $\ell_0=\log N_Z$ and,
for $1\leq r\leq d$, estimate $\ell_r/D_r$ to additive error
\begin{equation}
  \varepsilon_{\mathrm n,r}
  :=\min\!\left\{1,\frac{2\delta L_{m,k}}{dD_r}\right\}.
  \label{eq:torsion-join-precision}
\end{equation}
The resulting estimate of $L_{m,k}$ then has error at most $\delta L_{m,k}$.  Because
$\Gamma_T$ is fixed, the graph has succinct adjacency access; assigning failure
probability $\nu/d$ per degree and combining
\eqref{eq:torsion-join-candidate-ratio},
\eqref{eq:torsion-join-boundary-gap}, and
the bound $S_{\ast,r}=\log N_Z$ from \cref{lem:integer-trace-bound},
\cref{thm:quantum-pseudologdet} gives boundary-block-encoding query complexity polynomial in
$N_Z$, $\delta^{-1}$, and $\log(1/\nu)$.  This is a relative-error benchmark
for the log torsion order; \cref{sec:interpretation} states the general
limitations.

\begin{corollary}[Resource comparison with explicit boundary construction]
\label{cor:join-resource-comparison}
Let $m=2^k$, $k\geq4$, $N=N_Z$, and
$L_{m,k}:=(m-1)^k\log2$.  For $0\leq r\leq d$, set
\[
 D_r=\binom{N}{r+1},
 \qquad
 R_{m,k}:=\max_{0\leq r\leq d}\frac{D_r}{L_{m,k}}.
\]
Assume the access conditions of \cref{ass:block-trace} for the
zero-extended boundary matrices, with
$\alpha_r=\sqrt N$ and $\gamma_r=\sqrt{c_Y/N}$.
For $0<\delta\leq1$ and $0<\nu<1$, there is a quantum algorithm
returning $\widehat L$ such that
\[
 \Pr\!\left[
 |\widehat L-L_{m,k}|\leq\delta L_{m,k}
 \right]\geq1-\nu,
\]
using a total of
\begin{equation}
 Q_B=\widetilde O\!\left(
 \sqrt{\frac{N}{c_Y}}
 \left(d+\frac{d^2R_{m,k}}{\delta}\right)
 \log\frac{2d}{\nu}\right)
 \label{eq:join-resource-queries}
\end{equation}
queries to the boundary block encodings in degrees $1,\ldots,d$.
The degree-zero term $\ell_0(Z_{m,k})=\log N$ is evaluated
classically.  The number of adjacency-oracle queries is $O(d^2Q_B)$.

If each multipartite vertex is supplied with its part label and its
index within that part, adjacency can be computed reversibly from
these labels and the fixed graph $\Gamma_T$.  Under the synthesis
assumptions of \cref{subsec:implementation-errors}, the algorithm uses
$\widetilde O(NQ_B)$ quantum gates, excluding the classical computation
of polynomial coefficients and QSVT phases.

The top boundary matrix satisfies
\begin{equation}
 f_d(Z_{m,k})=60m^k,
 \qquad
 \operatorname{nnz}\!\bigl(B_d(Z_{m,k})\bigr)
 =60(k+3)m^k
 =N^{\Theta(\log N)}.
 \label{eq:join-explicit-boundary-size}
\end{equation}
Thus explicitly listing its nonzero entries has superpolynomial cost,
while the stated quantum query and gate bounds are polynomial in $N$
for fixed $\delta$ and $\nu$.
\end{corollary}

The proof is given in \cref{app:join-resource-proof}.  This comparison
concerns explicit boundary construction, not all classical algorithms:
given $m$ and $k$, the closed form $L_{m,k}=(m-1)^k\log2$ is itself
classically efficient to evaluate, and the entire graph can be read
with $O(N^2)$ adjacency queries.

\subsection{A block-local bit-oracle separation}
\label{subsec:matched-query-separation}

The closed-form family above is a scaling benchmark.  For a block-local bit-oracle lower
bound, recall $Y=\Cl(\Gamma_T)$ from \cref{subsec:torsion-clique-benchmark}, fix a
spanning tree $S\subseteq\Gamma_T$, and set $\Gamma_0=S$, $\Gamma_1=\Gamma_T$.  For
$x\in\{0,1\}^n$, put
$W_x:=\Cl(\Gamma_{x_1}\vee_{\mathrm v}\cdots\vee_{\mathrm v}\Gamma_{x_n})$ and
$w(x):=\sum_a x_a$.  Here $\vee_{\mathrm v}$ denotes the one-vertex wedge,
obtained by identifying one chosen vertex from each graph; it is distinct from
the graph join $\vee$ in \eqref{eq:torsion-join-benchmark}.  A clique lies in one
wedge summand apart from the shared vertex, so $W_x$ is a clique complex with
$\wtH_1(W_x;\Z)\cong(\Z/2\Z)^{w(x)}$.

The complex $Y$ has $60$ triangles and is its unique two-tree.  The matrix-tree
recurrence and \cref{ex:rp2-tree} give
$L_Y:=\ell_2(Y)=\log(4\tau_1(Y))>0$.  Use its $60n$ block-local triangles and
$90n$ block-local edges as candidate spaces, retaining a triangle block precisely
when $x_a=1$.  With zero padding,
\begin{equation}
  B_2(W_x)=\operatorname{diag}(x)\otimes B_2(Y),
  \qquad
  \frac{\ell_2(W_x)}{60n}=\frac{L_Y}{60}\frac{w(x)}n,
  \qquad
  \log|\Tor\wtH_1(W_x;\Z)|=w(x)\log2.
  \label{eq:matched-gadget-identities}
\end{equation}
Both algorithms receive the bit oracle
$O_x\ket{a,b}=\ket{a,b\mathbin\oplus x_a}$; here it is equivalent, up to
constant overhead, to membership for a block-local triangle.

\begin{proposition}[Quadratic block-local query separation]
\label{prop:matched-query-separation}
For $0<\varepsilon\leq1/12$, there is a quantum algorithm that estimates
the candidate-normalized log pseudodeterminant $\ell_2(W_x)/(60n)$ to additive
error $(L_Y/60)\varepsilon$, with success probability at least
$2/3$, using $O(\varepsilon^{-1})$ coherent queries to $O_x$.
Every randomized classical algorithm using classical queries to $O_x$ and
the same guarantee requires
$\Omega(\min\{n,\varepsilon^{-2}\})$ queries.  Hence
$\varepsilon=\Theta(n^{-1/2})$ gives
$O(\sqrt n)$ quantum queries versus $\Omega(n)$ classical queries.
\end{proposition}

\begin{proof}
The wedge decomposition and block diagonal boundary prove
\eqref{eq:matched-gadget-identities}.  A one-query encoding of
$\operatorname{diag}(x)$, composed with a fixed block encoding of $B_2(Y)$,
implements the required boundary block encoding.  Since $B_2(Y)$ is fixed, the possible
nonzero singular values of $B_2(W_x)/\sqrt6$ form a fixed finite set.  Hence one
may choose a fixed-degree real even polynomial that takes the complementary
target value $b_\alpha$ at zero and the value $h(x)$ from
\eqref{eq:complementary-log-target} at every such nonzero singular value.  A
constant rescaling enforces the weighted QSVT bound.  The bottom-left
construction, with this fixed rescaling undone in classical postprocessing,
and amplitude estimation then use $O(\varepsilon^{-1})$ coherent queries.  Here
$\alpha=\sqrt6$, $\gamma=\sqrt{c_Y/6}$, and the integer-matrix
bound permits $S_\ast=\log6$, all independently of $n$ and $x$.  By the second
identity in \eqref{eq:matched-gadget-identities}, Yao's principle shows that
distinguishing nearby central Hamming weights requires
$\Omega(\min\{n,\varepsilon^{-2}\})$ classical queries, which proves the claimed
lower bound.
\end{proof}

\begin{remark}[The constants]
\label{rem:gadget-constants}
For the gadget, $\|B_2(Y)\|=\sqrt6$ and
$\chg{\zeta_{\min}^{+}}(B_2(Y))=\sqrt{c_Y}=0.312869\ldots$, so
\cref{ass:block-trace} holds with $\alpha=\sqrt6$, $D=60n$, normalized gap
$\gamma=\sqrt{c_Y/6}=0.127728\ldots$, and $S_\ast=\log6$.  The quantities
$\alpha$, $\gamma$, and $S_\ast$ are independent of $n$ and $x$.
\end{remark}
Under the promise $n/3\leq w(x)\leq2n/3$, the stated accuracy is relative error at
most $3\varepsilon$ in the log torsion order.  For fixed $\varepsilon$ both query
bounds are constant: the separation uses shrinking precision and a promise-adapted
block-local candidate basis, not all global subsets.  No gate-complexity separation is
claimed.

\section{Interpretation and scope}\label{sec:interpretation}

As summarized in \cref{fig:quantum-pipeline}, the quantum routine returns an
estimate $\widehat{\ell_r(X)/D_r}$ of the boundary log pseudodeterminant
normalized by the candidate-space dimension $D_r$.  Internally, amplitude
estimation estimates a success probability, which the known rescaling and
offset subtraction convert into this normalized quantity.  Defining
\[
  \widehat\ell_r
  :=
  D_r\,\widehat{\ell_r(X)/D_r},
\]
the resulting estimates are combined using the identities in
\cref{tab:invariant-map} and \cref{sec:matrix-tree}.

Their arithmetic interpretation depends on the stated topological hypotheses.
The lower-homology orders $t_q$ and, for the certified-tree branch, the
certificate itself are supplied outside the estimator's guarantee.  Even in
the critical-group branch, the output is not generally a homology-torsion
order: by \cref{prop:critical-exact-sequence}, $|K_i(X)|$ also contains the
index of the image of the upper Laplacian in the integral boundaries.  The
projective-plane example, where $|K_1(T)|=4$ but
$|\wtH_1(T;\Z)|=2$, makes this distinction explicit.  Individual boundary
log pseudodeterminants and critical-group orders can also depend on the
simplicial representation even when integral homology is unchanged.

\paragraph{Error propagation and precision.}
If
\[
  \left|
    \widehat{\ell_r(X)/D_r}-\ell_r(X)/D_r
  \right|
  \leq\varepsilon_{\mathrm n,r},
\]
then $\widehat\ell_r$ has additive error
$\varepsilon_r=D_r\varepsilon_{\mathrm n,r}$.  If $\log t_q$ is supplied
with additive error $\delta_q$, then, with $\ell_0$ inserted exactly,
\begin{equation}
  |\widehat{\log\tau_k}-\log\tau_k|
  \leq
  \sum_{r=1}^{k}\varepsilon_r
  +2\sum_{q=0}^{k-2}\delta_q.
  \label{eq:interpretation-error-propagation}
\end{equation}
The same bound applies to the critical-group logarithm after setting
$k=i+1$, and half of the right-hand side applies to the certified-tree
homology logarithm.  Cancellation among alternating terms can make the
target small without reducing this worst-case error, so relative accuracy
in each $\ell_r$ does not imply relative accuracy in the final result.

For a positive order $M$, additive error $\varepsilon$ in $\log M$ confines
$M$ to a factor $e^{\pm\varepsilon}$.  For $M>1$, relative error $\delta$ in
$\log M$ confines the corresponding order only to
$[M^{1-\delta},M^{1+\delta}]$.  Exact recovery by exponentiation and rounding
generally requires additive log error on the scale of $1/M$.  Even an exact
order does not determine the invariant factors of the group, which require
additional integer-lattice information.

These guarantees remain conditional on the block-encoding,
state-preparation, singular-gap, normalized-sum-bound, and implementation
assumptions of \cref{ass:block-trace}.  Boundary-query complexity alone does
not imply the same quantum gate complexity; the additional implementation
costs are recorded in \eqref{eq:normalized-gates}.

\paragraph{What the query separation shows.}
For the candidate-normalized target $\ell_2(W_x)/(60n)$,
\cref{prop:matched-query-separation} gives
$O(\varepsilon^{-1})$ coherent bit-oracle queries, whereas every randomized
classical algorithm with classical access to the same bit oracle and the same
guarantee requires $\Omega(\min\{n,\varepsilon^{-2}\})$ queries.  These bounds
give a quadratic separation in the shrinking-precision regime, for example
$O(\sqrt n)$ quantum queries versus $\Omega(n)$ classical queries when
$\varepsilon=\Theta(n^{-1/2})$.  At fixed $\varepsilon$, both query bounds are
constant.  This comparison counts queries to an oracle that reveals which
fixed local blocks are present.  It does not establish an advantage in total
gate complexity or for the all-subsets clique encoding.

The family $Z_{m,k}$ serves a different purpose.  Its polynomial query and,
under the stated implementation assumptions, gate bounds show the resource
scaling obtained by avoiding explicit construction of a boundary matrix with
$N^{\Theta(\log N)}$ entries.  Because its answer is also available in closed
form, this is a scaling benchmark rather than evidence of general classical
hardness.

\paragraph{Filtrations.}
After the degree-wise rescaling above, the routine yields estimates of the
levelwise summary
\begin{equation}
  t\longmapsto
  \bigl(\ell_0(X_t),\ldots,\ell_{i+1}(X_t)\bigr).
  \label{eq:levelwise-summary}
\end{equation}
The algebraic hypotheses and correction data must be supplied at every level
where an arithmetic interpretation is claimed.  These scalar summaries do not
encode the maps induced by $X_s\hookrightarrow X_t$, so they do not produce an
integral persistence barcode.  No stability or noise-robustness bound for the
spectral curves is proved here.
\section{Conclusion}\label{sec:conclusion}

Quantum algorithms for topological data analysis have largely focused on the
kernel of the Hodge Laplacian and therefore on Betti numbers. This work shows
that the nonzero boundary spectrum is also an algorithmically accessible
source of arithmetic information. The quantum step estimates normalized
boundary log pseudodeterminants without a separate rank input, while
matrix-tree identities provide the conditions under which these spectral
quantities acquire a torsion-sensitive interpretation. Keeping these two
steps separate makes the scope of the method transparent and allows the same
spectral primitive to support several topological applications.

More broadly, our results clarify what is required for meaningful quantum
savings in this setting. Large chain spaces alone are in general insufficient:
meaningful efficiency and interpretation guarantees also depend on coherent
access, spectral gaps, precision requirements, and, where appropriate, the
relevant topological certificates. Our examples show that these requirements
are compatible with nontrivial torsion and can yield savings from avoiding
explicit boundary-matrix construction and, in a controlled oracle setting, a
quantum query advantage. Extending this framework to less structured and
data-derived complexes, together with matched classical comparisons and robust
behavior across filtrations, is a natural next step toward torsion-sensitive
quantum topological data analysis.

\section*{Acknowledgements}
\addcontentsline{toc}{section}{Acknowledgements}

\noindent \textbf{DT} acknowledges the following support. This work was supported by the Dutch National Growth Fund (NGF), as part of the Quantum Delta NL program.

\noindent \textbf{AIL} acknowledges support from the DFG under Germany's Excellence Strategy - EXC-2123 QuantumFrontiers-2 - 390837967, the Quantum Valley Lower Saxony, the BMBF project Quics, the QuantERA project ResourceQ, and the Research Council of Lithuania under the Program ‘University Excellence Initiatives’ of the Ministry of Education, Science and Sports of the Republic of Lithuania (Measure No. 12-001-01-01-01 ‘Improving the Research and Study Environment’), Project No. S-A-UEI-23-11.

\noindent \textbf{MYR} is supported by the project Divide and Quantum (Project No.~1389.20.241) of the research programme NWA-ORC, which is (partly) financed by the Dutch Research Council (NWO).

\noindent \textbf{All authors} are grateful to the organizers of the CQT Quantum Hackamonth at the National University of Singapore for providing the opportunity to initiate this project. Special thanks to Alessandro Luongo for useful discussions.

\noindent \textbf{AI use statement.} The authors used free and paid versions
of GPT, Claude, and Grammarly, and most notably Fable 5.0 and GPT-5.6 Sol, for
brainstorming, finding references, checking notation and formula indices,
proofreading, producing figures, and exploring alternative proofs. GPT-5.6 Sol
also assisted with adapting the appendix-only squared-amplitude alternative to
the paper's existing framework.

We emphasize that these tools were used interactively and under targeted human
direction, not as autonomous contributors or substitutes for scientific
judgment. The work resulted from months of human collaboration, discussion,
and exchange of ideas, conducted both in person and online.

All authors independently reviewed and approved the manuscript and take full
responsibility for every result and claim, including those developed with AI
assistance.

\clearpage

\bibliographystyle{plain}
\bibliography{bibliography}

@article{apers2023,
  author        = {Apers, Simon and Gribling, Sander and Sen, Sayantan and Szab{\'o}, D{\'a}niel},
  title         = {A (Simple) Classical Algorithm for Estimating {B}etti Numbers},
  journal       = {Quantum},
  volume        = {7},
  pages         = {1202},
  year          = {2023}
}

@phdthesis{Gribling2019Thesis,
  author = {Gribling, Sander Jan},
  title  = {Applications of Optimization to Factorization Ranks
            and Quantum Information Theory},
  school = {Tilburg University},
  year   = {2019},
  note   = {CentER Dissertation Series, volume 603},
  doi    = {10.26116/center-lis-1925},
  url    = {https://ir.cwi.nl/pub/28887/Thesis_S_Gribling.pdf}
}

@inproceedings{DuvalKlivansMartin2011FPSAC,
  author    = {Duval, Art M. and Klivans, Caroline J.
               and Martin, Jeremy L.},
  title     = {Critical Groups of Simplicial Complexes},
  booktitle = {23rd International Conference on Formal Power Series
               and Algebraic Combinatorics (FPSAC 2011)},
  series    = {Discrete Mathematics and Theoretical Computer
               Science Proceedings},
  volume    = {AO},
  pages     = {269--280},
  year      = {2011},
  doi       = {10.46298/dmtcs.2909},
  url       = {https://dmtcs.episciences.org/2909}
}

@inproceedings{Rudolph2026,
  author    = {Rudolph, Dorian},
  title     = {Towards a Universal Gateset for {QMA$_1$}},
  booktitle = {51st International Symposium on Mathematical
               Foundations of Computer Science (MFCS 2026)},
  series    = {Leibniz International Proceedings in Informatics},
  volume    = {386},
  pages     = {98:1--98:19},
  publisher = {Schloss Dagstuhl -- Leibniz-Zentrum f{\"u}r Informatik},
  year      = {2026},
  doi       = {10.4230/LIPIcs.MFCS.2026.98},
  url       = {https://drops.dagstuhl.de/entities/document/10.4230/LIPIcs.MFCS.2026.98}
}

@article{QueffelecZarouf2019,
  author        = {Queff{\'e}lec, Herv{\'e} and Zarouf, Rachid},
  title         = {On {Bernstein}'s Inequality for Polynomials},
  journal       = {Analysis and Mathematical Physics},
  volume        = {9},
  pages         = {1181--1207},
  year          = {2019},
  doi           = {10.1007/s13324-019-00294-x},
  eprint        = {1903.10801},
  archivePrefix = {arXiv},
  url           = {https://arxiv.org/abs/1903.10801}
}

@article{BerryEtAl2024QTDA,
  author        = {Berry, Dominic W. and Su, Yuan and Gyurik, Casper and King, Robbie and Basso, Joao and Del Toro Barba, Alexander and Rajput, Abhishek and Wiebe, Nathan and Dunjko, Vedran and Babbush, Ryan},
  title         = {Analyzing Prospects for Quantum Advantage in Topological Data Analysis},
  journal       = {PRX Quantum},
  volume        = {5},
  number        = {1},
  pages         = {010319},
  year          = {2024},
  doi           = {10.1103/PRXQuantum.5.010319},
  eprint        = {2209.13581},
  archivePrefix = {arXiv},
  primaryClass  = {quant-ph}
}

@article{https://doi.org/10.1112/S0024609397003305,
  author        = {Biggs, Norman},
  title         = {Algebraic Potential Theory on Graphs},
  journal       = {Bulletin of the London Mathematical Society},
  volume        = {29},
  number        = {6},
  pages         = {641--682},
  year          = {1997},
  doi           = {10.1112/S0024609397003305}
}

@article{BoissonnatMaria2019,
  author        = {Boissonnat, Jean-Daniel and Maria, Cl{\'e}ment},
  title         = {Computing Persistent Homology with Various Coefficient Fields in a Single Pass},
  journal       = {Journal of Applied and Computational Topology},
  volume        = {3},
  number        = {1--2},
  pages         = {59--84},
  year          = {2019},
  doi           = {10.1007/s41468-019-00025-y}
}

@article{CarlssonEtAl2008,
  author        = {Carlsson, Gunnar and Ishkhanov, Tigran and de Silva, Vin and Zomorodian, Afra},
  title         = {On the Local Behavior of Spaces of Natural Images},
  journal       = {International Journal of Computer Vision},
  volume        = {76},
  number        = {1},
  pages         = {1--12},
  year          = {2008},
  doi           = {10.1007/s11263-007-0056-x}
}

@article{crichigno2024,
  author        = {Crichigno, Marcos and Kohler, Tamara},
  title         = {Clique Homology is {QMA$_1$}-Hard},
  journal       = {Nature Communications},
  volume        = {15},
  pages         = {9846},
  year          = {2024}
}

@article{DuvalKlivansMartin2011Cellular,
  author        = {Duval, Art M. and Klivans, Caroline J. and Martin, Jeremy L.},
  title         = {Cellular Spanning Trees and {L}aplacians of Cubical Complexes},
  journal       = {Advances in Applied Mathematics},
  volume        = {46},
  number        = {1--4},
  pages         = {247--274},
  year          = {2011},
  doi           = {10.1016/j.aam.2010.05.005},
  eprint        = {0908.1956},
  archivePrefix = {arXiv},
  primaryClass  = {math.CO}
}

@article{dmtcs:2909,
  author        = {Duval, Art M. and Klivans, Caroline J. and Martin, Jeremy L.},
  title         = {Critical Groups of Simplicial Complexes},
  journal       = {Annals of Combinatorics},
  volume        = {17},
  number        = {1},
  pages         = {53--70},
  year          = {2013},
  doi           = {10.1007/s00026-012-0168-z},
  eprint        = {1101.3981},
  archivePrefix = {arXiv},
  primaryClass  = {math.CO}
}

@article{duval2008simplicialmatrixtreetheorems,
  author        = {Duval, Art M. and Klivans, Caroline J. and Martin, Jeremy L.},
  title         = {Simplicial Matrix-Tree Theorems},
  journal       = {Transactions of the American Mathematical Society},
  volume        = {361},
  number        = {11},
  pages         = {6073--6114},
  year          = {2009},
  doi           = {10.1090/S0002-9947-09-04898-3},
  eprint        = {0802.2576},
  archivePrefix = {arXiv},
  primaryClass  = {math.CO}
}

@article{Frosini2013,
  author        = {Frosini, Patrizio},
  title         = {Stable Comparison of Multidimensional Persistent Homology Groups with Torsion},
  journal       = {Acta Applicandae Mathematicae},
  volume        = {124},
  number        = {1},
  pages         = {43--54},
  year          = {2013},
  doi           = {10.1007/s10440-012-9769-0}
}

@article{gyurik2022,
  author        = {Gyurik, Casper and Cade, Chris and Dunjko, Vedran},
  title         = {Towards Quantum Advantage via Topological Data Analysis},
  journal       = {Quantum},
  volume        = {6},
  pages         = {855},
  year          = {2022}
}

@article{Hayakawa_2022,
  author        = {Hayakawa, Ryu},
  title         = {Quantum Algorithm for Persistent {B}etti Numbers and Topological Data Analysis},
  journal       = {Quantum},
  volume        = {6},
  pages         = {873},
  year          = {2022},
  doi           = {10.22331/q-2022-12-07-873}
}

@article{torsion_in_med,
  author        = {Li, Zhen and Qi, Mingming and Huang, Juyuan and Zhang, Wei and Tan, Xu and Chen, Yifan},
  title         = {Geometry-Enhanced Graph Neural Networks Accelerate {circRNA} Therapeutic Target Discovery},
  journal       = {Frontiers in Genetics},
  volume        = {16},
  pages         = {1633391},
  year          = {2025},
  doi           = {10.3389/fgene.2025.1633391}
}

@article{lloyd2015quantumalgorithmstopologicalgeometric,
  author        = {Lloyd, Seth and Garnerone, Silvano and Zanardi, Paolo},
  title         = {Quantum Algorithms for Topological and Geometric Analysis of Data},
  journal       = {Nature Communications},
  volume        = {7},
  pages         = {10138},
  year          = {2016},
  doi           = {10.1038/ncomms10138},
  eprint        = {1408.3106},
  archivePrefix = {arXiv},
  primaryClass  = {quant-ph}
}

@article{Lorenzini2008SmithNF,
  author        = {Lorenzini, Dino J.},
  title         = {{S}mith Normal Form and {L}aplacians},
  journal       = {Journal of Combinatorial Theory, Series B},
  volume        = {98},
  number        = {6},
  pages         = {1271--1300},
  year          = {2008},
  doi           = {10.1016/j.jctb.2008.02.002}
}

@article{Martyn2021GrandAlgorithms,
  author        = {Martyn, John M. and Rossi, Zane M. and Tan, Andrew K. and Chuang, Isaac L.},
  title         = {Grand Unification of Quantum Algorithms},
  journal       = {PRX Quantum},
  volume        = {2},
  number        = {4},
  pages         = {040203},
  year          = {2021},
  doi           = {10.1103/PRXQuantum.2.040203}
}

@article{McArdle2022AQubits,
  author        = {McArdle, Sam and Gily{\'e}n, Andr{\'a}s and Berta, Mario},
  title         = {A Streamlined Quantum Algorithm for Topological Data Analysis with Exponentially Fewer Qubits},
  journal       = {Quantum},
  volume        = {10},
  pages         = {2058},
  year          = {2026},
  doi           = {10.22331/q-2026-04-10-2058},
  eprint        = {2209.12887},
  archivePrefix = {arXiv},
  primaryClass  = {quant-ph}
}

@article{ObayashiYoshiwaki2023,
  author        = {Obayashi, Ippei and Yoshiwaki, Michio},
  title         = {Field Choice Problem in Persistent Homology},
  journal       = {Discrete \& Computational Geometry},
  volume        = {70},
  number        = {3},
  pages         = {645--670},
  year          = {2023},
  doi           = {10.1007/s00454-023-00544-7}
}

@article{Scali_2024,
  author        = {Scali, Stefano and Umeano, Chukwudubem and Kyriienko, Oleksandr},
  title         = {Quantum Topological Data Analysis via the Estimation of the Density of States},
  journal       = {Physical Review A},
  volume        = {110},
  number        = {4},
  pages         = {042616},
  year          = {2024},
  doi           = {10.1103/PhysRevA.110.042616}
}

@article{SchaubEtAl2020Hodge,
  author        = {Schaub, Michael T. and Benson, Austin R. and Horn, Paul and Lippner, Gabor and Jadbabaie, Ali},
  title         = {Random Walks on Simplicial Complexes and the Normalized {H}odge 1-{L}aplacian},
  journal       = {SIAM Review},
  volume        = {62},
  number        = {2},
  pages         = {353--391},
  year          = {2020},
  doi           = {10.1137/18M1201019}
}

@article{ShenEtAl2025TorGNN,
  author        = {Shen, Cong and Liu, Xiang and Luo, Jiawei and Xia, Kelin},
  title         = {Torsion Graph Neural Networks},
  journal       = {IEEE Transactions on Pattern Analysis and Machine Intelligence},
  volume        = {47},
  number        = {4},
  pages         = {2946--2956},
  year          = {2025},
  doi           = {10.1109/TPAMI.2025.3528449}
}

@article{WeiWei2025PersistentTopologicalLaplacians,
  author        = {Wei, Xiaoqi and Wei, Guo-Wei},
  title         = {Persistent Topological {L}aplacians---A Survey},
  journal       = {Mathematics},
  volume        = {13},
  number        = {2},
  pages         = {208},
  year          = {2025},
  doi           = {10.3390/math13020208}
}

@article{ZomorodianCarlsson2005,
  author        = {Zomorodian, Afra and Carlsson, Gunnar},
  title         = {Computing Persistent Homology},
  journal       = {Discrete \& Computational Geometry},
  volume        = {33},
  number        = {2},
  pages         = {249--274},
  year          = {2005},
  doi           = {10.1007/s00454-004-1146-y}
}

@inproceedings{BartschiEidenbenz2019Dicke,
  author        = {B{\"a}rtschi, Andreas and Eidenbenz, Stephan},
  title         = {Deterministic Preparation of {D}icke States},
  booktitle     = {Fundamentals of Computation Theory},
  series        = {Lecture Notes in Computer Science},
  volume        = {11651},
  pages         = {126--139},
  publisher     = {Springer},
  year          = {2019},
  doi           = {10.1007/978-3-030-25027-0_9},
  eprint        = {1904.07358},
  archivePrefix = {arXiv},
  primaryClass  = {quant-ph}
}

@inproceedings{Gilyen2019QuantumArithmetics,
  author        = {Gily{\'e}n, Andr{\'a}s and Su, Yuan and Low, Guang Hao and Wiebe, Nathan},
  title         = {Quantum Singular Value Transformation and Beyond: Exponential Improvements for Quantum Matrix Arithmetics},
  booktitle     = {Proceedings of the 51st Annual ACM SIGACT Symposium on Theory of Computing},
  pages         = {193--204},
  publisher     = {ACM},
  year          = {2019},
  doi           = {10.1145/3313276.3316366},
  eprint        = {1806.01838},
  archivePrefix = {arXiv},
  primaryClass  = {quant-ph}
}

@inproceedings{king2024,
  author        = {King, Robbie and Kohler, Tamara},
  title         = {Promise Clique Homology on Weighted Graphs is {QMA$_1$}-Hard and Contained in {QMA}},
  booktitle     = {Proceedings of the 65th IEEE Symposium on Foundations of Computer Science (FOCS)},
  publisher     = {IEEE},
  year          = {2024}
}

@inproceedings{luongo2024quantumalgorithmsspectralsums,
  author        = {Luongo, Alessandro and Shao, Changpeng},
  title         = {Quantum Algorithms for Spectral Sums},
  booktitle     = {Proceedings of the AAAI Conference on Artificial Intelligence},
  volume        = {40},
  pages         = {24178--24188},
  publisher     = {AAAI Press},
  year          = {2026},
  doi           = {10.1609/aaai.v40i29.39597},
  eprint        = {2011.06475},
  archivePrefix = {arXiv},
  primaryClass  = {quant-ph}
}

@incollection{Bjorner1992MatroidShellability,
  author        = {Bj{\"o}rner, Anders},
  title         = {The Homology and Shellability of Matroids and Geometric Lattices},
  booktitle     = {Matroid Applications},
  editor        = {White, Neil},
  series        = {Encyclopedia of Mathematics and its Applications},
  volume        = {40},
  pages         = {226--283},
  publisher     = {Cambridge University Press},
  address       = {Cambridge},
  year          = {1992},
  doi           = {10.1017/CBO9780511662041.008}
}

@incollection{BrassardEtAl2002Amplitude,
  author        = {Brassard, Gilles and H{\o}yer, Peter and Mosca, Michele and Tapp, Alain},
  title         = {Quantum Amplitude Amplification and Estimation},
  booktitle     = {Quantum Computation and Quantum Information},
  editor        = {Lomonaco, Samuel J.},
  series        = {Contemporary Mathematics},
  volume        = {305},
  pages         = {53--74},
  publisher     = {American Mathematical Society},
  address       = {Providence, RI},
  year          = {2002},
  doi           = {10.1090/conm/305/05215},
  eprint        = {quant-ph/0005055},
  archivePrefix = {arXiv},
  primaryClass  = {quant-ph}
}

@book{edelsbrunner2010,
  author        = {Edelsbrunner, Herbert and Harer, John L.},
  title         = {Computational Topology: An Introduction},
  publisher     = {American Mathematical Society},
  address       = {Providence, RI},
  year          = {2010}
}

@book{hatcher2002,
  author        = {Hatcher, Allen},
  title         = {Algebraic Topology},
  publisher     = {Cambridge University Press},
  address       = {Cambridge},
  year          = {2002}
}

@misc{cade2021,
  author        = {Cade, Chris and Crichigno, P. Marcos},
  title         = {Complexity of Supersymmetric Systems and the Cohomology Problem},
  howpublished  = {arXiv:2107.00011},
  year          = {2021},
  eprint        = {2107.00011},
  archivePrefix = {arXiv},
  primaryClass  = {quant-ph}
}

@misc{kerenidis2022quantummachinelearningsubspace,
  author        = {Kerenidis, Iordanis and Prakash, Anupam},
  title         = {Quantum Machine Learning with Subspace States},
  howpublished  = {arXiv:2202.00054},
  year          = {2022},
  eprint        = {2202.00054},
  archivePrefix = {arXiv},
  primaryClass  = {quant-ph}
}

@misc{nghiem2025torsion,
  author        = {Nghiem, Nhat A.},
  title         = {Towards Quantum Topological Data Analysis: Torsion Detection},
  howpublished  = {arXiv:2508.19943},
  year          = {2025},
  eprint        = {2508.19943},
  archivePrefix = {arXiv},
  primaryClass  = {quant-ph}
}

@misc{strelchuk2026average,
  author        = {Strelchuk, Sergii and Subramanian, Sathyawageeswar and Weso{\l}owski, Adam},
  title         = {Average-Case Hardness of {B}etti Number Estimation},
  howpublished  = {arXiv:2609.12777},
  year          = {2026},
  eprint        = {2609.12777},
  archivePrefix = {arXiv},
  primaryClass  = {quant-ph}
}

@misc{wesolowski2026torsion,
  author        = {Weso{\l}owski, Adam},
  title         = {Torsion Detection in Clique Complexes is Conditionally {QMA$_1$}-Hard},
  howpublished  = {arXiv:2609.14110},
  year          = {2026},
  eprint        = {2609.14110},
  archivePrefix = {arXiv},
  primaryClass  = {quant-ph}
}

@misc{nghiemBerryPhan2026torsion,
  author        = {Nghiem, Nhat A. and Berry, Dominic W. and Phan, Trung V.},
  title         = {Quantum Topological Data Analysis Beyond {B}etti Numbers: Complexity Hardness \& An Algorithm for Torsion Witness},
  howpublished  = {arXiv:2609.28112},
  year          = {2026},
  eprint        = {2609.28112},
  archivePrefix = {arXiv},
  primaryClass  = {quant-ph}
}

\clearpage
\appendix

\section{Kernel-suppressing squared-amplitude alternative}
\label{app:squared-amplitude-estimator}
For comparison with the bottom-left construction, we outline a
kernel-suppressing squared-amplitude variant and give resource bounds
conditional on the polynomial parameters specified below.
{
Under the hypotheses of \cref{cor:root-free-critical}, when the lower integral
homology corrections vanish, the identity is
\begin{align}
  \log|K_i(X)|
  &=
  \sum_{r=0}^{i+1}(-1)^{i+1-r}
  \log\pdet\!\left(L_{r-1}^{\up}\right) \notag\\
  &=
  \sum_{r=0}^{i+1}
  \sum_{\lambda_j(L_{r-1}^{\up})>0}
  (-1)^{i+1-r}\log\lambda_j(L_{r-1}^{\up}),
  \qquad
  L_{r-1}^{\up}:=B_rB_r^{\mathsf T} .
  \label{eq:squared-route-root-free}
\end{align}
}
{
We combine the block encoding of $\mathcal{B}_r$ from~\cite{McArdle2022AQubits}, QSVT from~\cite{Gilyen2019QuantumArithmetics,Martyn2021GrandAlgorithms}, and quantum
spectral-sum estimation from~\cite{luongo2024quantumalgorithmsspectralsums}.  The first step is to implement
a block encoding of the boundary operator.}
{
After embedding the rectangular boundary in compatible zero-padded domain and
codomain spaces, the block encoding satisfies
\begin{equation}
  \bigl(\bra{0^{n_{B,r}}}\otimes I\bigr)U_{\mathcal{B}_r}
  \bigl(\ket{0^{n_{B,r}}}\otimes I\bigr)
  =
  \frac{\mathcal{B}_r}{\sqrt N}.
  \label{eq:squared-route-boundary-block}
\end{equation}
}
Here $n_{B,r}$ denotes the number of block-encoding ancillas, and
$\mathcal{B}_r$ is the zero-padded boundary from
\eqref{eq:padded-boundary}.  For a clique complex,
coherent simplex-membership access is needed; this means a reversible
membership oracle that can act on superpositions.

{
Next, we apply a polynomial transformation of $\mathcal{B}_r/\sqrt N$ via QSVT with a
polynomial $p_r(x)$ approximating the square root of a logarithm.}
{
Let
\begin{equation}
  0<\gamma_r\leq
  \frac{\zeta_{\min}^{+}(\mathcal{B}_r)}{\sqrt N},
  \qquad
  \mathcal D_{\gamma_r}
  :=
  [-1,-\gamma_r]\cup[\gamma_r,1],
  \label{eq:squared-route-gap-domain}
\end{equation}
where $\zeta_{\min}^{+}(\mathcal{B}_r)$ is the smallest positive singular value.  On the
promised spectrum, fix a constant $\chg{c_\star}>1$, independent of the problem
parameters, and take the target function
\begin{equation}
  g_r(x)
  =
  \begin{cases}
    \sqrt{\log(\chg{c_\star}/x^2)}, & x\in\mathcal D_{\gamma_r},\\
    0, & |x|\leq\gamma_r/2,
  \end{cases}
  \qquad
  p_r(x)\approx g_r(x).
  \label{eq:squared-amplitude-overview}
\end{equation}
}
The interval $\gamma_r/2<|x|<\gamma_r$ is a transition region containing no
promised singular value.
{
For this alternative route we keep the polynomial dependence explicit.  Fix
$\chg{\eta_r}\in(0,1)$ and assume that an even real polynomial $p_r$ of degree $\chg{d_{\log,r}}$
is supplied such that
\begin{align}
  |p_r(x)-g_r(x)|&\leq\chg{\eta_r}
  &&(|x|\in[\gamma_r,1]),\label{eq:squared-route-polynomial-nonzero}\\
  |p_r(x)|&\leq\chg{\eta_r}
  &&(|x|\leq\gamma_r/2).\label{eq:squared-route-polynomial-zero}
\end{align}
Define
\begin{equation}
  \Lambda_r^{\mathrm{sq}}
  :=
  \max\!\left\{
    1,\sqrt{\log(\chg{c_\star}/\gamma_r^2)},
    \lVert p_r\rVert_{[-1,1]}
  \right\},
  \qquad
  \widetilde p_r:=p_r/\Lambda_r^{\mathrm{sq}}.
  \label{eq:squared-route-polynomial-rescaling}
\end{equation}
Then $\widetilde p_r$ is an admissible even QSVT polynomial.  Such a $p_r$
exists after choosing a continuous even interpolation across the transition
window and applying polynomial approximation; \add{an explicit bound
on $\chg{d_{\log,r}}$ is given in \cref{thm:squared-route-degree} below.}  The bounds below
are otherwise conditional on the supplied parameters
$(\chg{d_{\log,r}},\Lambda_r^{\mathrm{sq}},\chg{\eta_r})$.}

{\color{addition}
\paragraph{Degree bound.}
The fixed shift in \eqref{eq:squared-amplitude-overview} keeps the argument of
the square root strictly positive at $|x|=1$, removing the endpoint branch
point.

\begin{theorem}[Degree of the squared-amplitude polynomial]
\label{thm:squared-route-degree}
For every $\gamma_r\in(0,1]$ and $\chg{\eta_r}\in(0,1)$ there is an even real
polynomial $p_r$ satisfying
\eqref{eq:squared-route-polynomial-nonzero}--\eqref{eq:squared-route-polynomial-zero}
with
\begin{equation}
  \chg{d_{\log,r}}
  =
  O_{c_\star}\!\left(
    \frac{1}{\gamma_r}
    \operatorname{polylog}\frac{1}{\gamma_r\chg{\eta_r}}
  \right),
  \qquad
  \Lambda_r^{\mathrm{sq}}
  =
  O_{c_\star}\!\left(\sqrt{\log(c_\star/\gamma_r^2)}\right).
  \label{eq:squared-route-degree}
\end{equation}
\end{theorem}

\begin{proof}
Put $a=3/4$, let
$\kappa_{c_\star}:=\min\{1/32,(\sqrt{c_\star}-1)/8\}$, and set
$\delta:=\kappa_{c_\star}\gamma_r$.  Apply the bounded multi-interval
approximation theorem~\cite[Theorem~68]{Gilyen2019QuantumArithmetics}
to $I_-=[-1,-a\gamma_r]$, $I_0=[-\gamma_r/2,\gamma_r/2]$, and
$I_+=[a\gamma_r,1]$.  Use the auxiliary target
$\sqrt{\log(c_\star/x^2)}$ on $I_-\cup I_+$ and zero on $I_0$.

These intervals are separated by $\Theta(\gamma_r)$, and the chosen width
$\delta$ satisfies the hypotheses of the cited theorem.  The outer branches
extend holomorphically to disks whose radii exceed the corresponding interval
radii by $2\delta$.  These disks remain a distance
$\Omega_{c_\star}(\gamma_r)$ from the origin and avoid
$\pm\sqrt{c_\star}$.  Cauchy's estimates therefore verify the theorem's
Taylor-coefficient condition with a common $B\geq1$ satisfying
$B=O_{c_\star}\!\left(\gamma_r^{-1}
\sqrt{\log(c_\star/\gamma_r^2)}\right)$.
The target on the enlarged real intervals is bounded by
$O_{c_\star}\!\left(\sqrt{\log(c_\star/\gamma_r^2)}\right)$.

Applying the cited theorem with
$\eta'_r:=\min\{\eta_r,(6B)^{-1}\}$ gives a polynomial $P$ of degree
$O_{c_\star}\!\left(\gamma_r^{-1}
\log\!\frac{e}{\gamma_r\eta_r}\right)$ with the required approximation and
global norm bounds.  Finally,
$p_r(x):=(\operatorname{Re}P(x)+\operatorname{Re}P(-x))/2$ is real and even
and has no larger degree, error, or norm.  Since
$[\gamma_r,1]\subseteq[a\gamma_r,1]$, this proves
\eqref{eq:squared-route-polynomial-nonzero}--
\eqref{eq:squared-route-polynomial-zero} and the stated bounds.
\end{proof}
}

{
Next, we apply amplitude estimation using the flagged valid-simplex trace state
$\chg{\ket{\Phi_{f_r}}}$,}
{
\begin{equation}
  \chg{\ket{\Phi_{f_r}}}
  :=
  \frac{1}{\sqrt{\chg{f_r}}}
  \sum_{j=1}^{\chg{f_r}}
  \ket{\sigma_r^j}\ket{\sigma_r^j},
  \label{eq:squared-route-valid-trace-state}
\end{equation}
}
{
The state preparation of $\chg{\ket{\Phi_{f_r}}}$ starts from the Dicke state,
a uniform superposition of $(r+1)$-Hamming-weight basis states, followed by
flagging basis states representing valid $r$-simplices and copying them into
another qubit register.  Writing $\chg{D_r}=\binom{N}{r+1}$, this gives a
probabilistic implementation from the state}
\begin{equation}
  {
  \sqrt{\frac{\chg{f_r}}{\chg{D_r}}}\,
  \chg{\ket{\Phi_{f_r}}}\ket0
  +
  \ket{\mathrm{garb}}\ket1,}
  \label{eq:squared-route-flagged-trace-state}
\end{equation}
where $\ket{\mathrm{garb}}$ is subnormalized and contains the basis states
representing invalid $r$-simplex candidates.

{
After the polynomial transformation, the component with the flag and QSVT
signal ancillas equal to zero has the form}
{
\begin{equation}
  \ket{0^{n_{B,r}}}\,
  \sqrt{\frac{\chg{f_r}}{\chg{D_r}(\Lambda_r^{\mathrm{sq}})^2}}\,
  \left(
    p_r^{\mathrm{SV}}\!\left(\frac{\mathcal{B}_r}{\sqrt N}\right)
    \otimes I
  \right)
  \chg{\ket{\Phi_{f_r}}}
  +
  \ket{\perp}.
  \label{eq:squared-route-final-state}
\end{equation}
}
Here $p_r^{\mathrm{SV}}$ denotes the right-space singular-value
transformation, $\Lambda_r^{\mathrm{sq}}$ is the rescaling factor used to keep the QSVT
polynomial bounded, and $\ket{\perp}$ is orthogonal to the all-zero signal
subspace.

{
By amplitude estimation, we obtain the probability of that all-zero outcome.}
{
Because a probability is a squared norm, the trace identity is
\begin{align}
  P_r^{\mathrm{amp}}
  &=
  \frac{1}{\chg{D_r}(\Lambda_r^{\mathrm{sq}})^2}
  \Tr_{\mathrm{valid},r}\!\left[
    p_r^{\mathrm{SV}}\!\left(\frac{\mathcal{B}_r}{\sqrt N}\right)^\ast
    p_r^{\mathrm{SV}}\!\left(\frac{\mathcal{B}_r}{\sqrt N}\right)
  \right] \notag\\
  &\approx
  \frac{1}{\chg{D_r}(\Lambda_r^{\mathrm{sq}})^2}
  \sum_{\zeta_j(\mathcal{B}_r)>0}
  \log\!\left(\frac{\chg{c_\star}N}{\zeta_j(\mathcal{B}_r)^2}\right) \notag\\
  &=
  \frac{\rank(\mathcal{B}_r)\log(\chg{c_\star}N)-\ell(\mathcal{B}_r)}
       {\chg{D_r}(\Lambda_r^{\mathrm{sq}})^2},
  \qquad
  \ell(\mathcal{B}_r)
  =
  \sum_{\zeta_j(\mathcal{B}_r)>0}\log\zeta_j(\mathcal{B}_r)^2
  =\ell(B_r).
  \label{eq:squared-route-amplitude}
\end{align}
}
Here $\Tr_{\mathrm{valid},r}$ is the trace over the valid $r$-simplex
subspace; invalid candidates remain in the flagged complement.  In the second
line, $\zeta_j(\mathcal{B}_r)^2=\lambda_j(L_{r-1}^{\up})$.
Consequently, the squared-amplitude estimate of $\ell(B_r)=\ell(\mathcal{B}_r)$ is
{
\begin{equation}
  \ell(B_r)=\ell(\mathcal{B}_r)
  \approx
  \rank(\mathcal{B}_r)\log(\chg{c_\star}N)
  -
  \chg{D_r}(\Lambda_r^{\mathrm{sq}})^2P_r^{\mathrm{amp}}.
  \label{eq:squared-route-rank-correction}
\end{equation}
}
More precisely, if the implemented normalized QSVT block differs from
$\widetilde p_r^{\mathrm{SV}}(\mathcal{B}_r/\sqrt N)$ in operator norm by at
most $e_{\mathrm{svt},r}$, and amplitude estimation returns $\widehat P_r$
with probability error at most $\delta_{\mathrm{AE},r}$, define
\begin{equation}
  \widehat\ell_r
  :=
  \rank(\mathcal{B}_r)\log(\chg{c_\star}N)-\chg{D_r}(\Lambda_r^{\mathrm{sq}})^2\widehat P_r.
  \label{eq:squared-route-estimate}
\end{equation}
Then
\begin{equation}
  |\widehat\ell_r-\ell(B_r)|
  \leq
  \chg{D_r}(2\Lambda_r^{\mathrm{sq}}\chg{\eta_r}+\chg{\eta_r^2})
  +\chg{D_r}(\Lambda_r^{\mathrm{sq}})^2
  \left(
    \delta_{\mathrm{AE},r}
    +2e_{\mathrm{svt},r}
    +e_{\mathrm{svt},r}^2
  \right).
  \label{eq:squared-route-error-budget}
\end{equation}
For $0<\varepsilon_r\leq \chg{D_r}$, it is sufficient to choose
\begin{equation}
  \chg{\eta_r}\leq\frac{\varepsilon_r}{9\chg{D_r}\Lambda_r^{\mathrm{sq}}},
  \qquad
  \delta_{\mathrm{AE},r}
  \leq\frac{\varepsilon_r}{3\chg{D_r}(\Lambda_r^{\mathrm{sq}})^2},
  \qquad
  e_{\mathrm{svt},r}
  \leq\frac{\varepsilon_r}{9\chg{D_r}(\Lambda_r^{\mathrm{sq}})^2}.
  \label{eq:squared-route-error-allocation}
\end{equation}
Thus $\chg{d_{\log,r}}$ in the resource bounds is understood to be the degree needed for
\eqref{eq:squared-route-polynomial-nonzero}--
\eqref{eq:squared-route-polynomial-zero} at the value of $\chg{\eta_r}$ required by
\eqref{eq:squared-route-error-allocation}.
{
The last step is to apply this construction to the required boundary
operators, followed by summing the rank-corrected estimates $\widehat\ell_r$
with the alternating signs in \eqref{eq:squared-route-root-free}.}
{
More precisely, the augmented term $\ell(B_0)=\log N$ is inserted exactly,
and the quantum procedure is applied for $r=1,\ldots,i+1$.}

\paragraph{Resource accounting.}
{
The complexity analysis has the following components.}
\begin{enumerate}[label=(\roman*),leftmargin=2em]
  \item {
  The block encoding $U_{\mathcal{B}_r}$ is supplied by~\cite{McArdle2022AQubits} using Clifford loaders from~\cite{kerenidis2022quantummachinelearningsubspace} and a clique-membership
  oracle.}
  Its gate cost depends on the coherent input model and is denoted by
  $G_{B,r}$.

  \item {
  The number of queries to $U_{\mathcal{B}_r}$ needed to construct the polynomial
  transformation is the degree $\chg{d_{\log,r}}$ of the polynomial $p_r$.}

  \item {
  Amplitude estimation calls the preceding polynomial transformation
  inversely proportionally to the requested probability error.}
  {
  To obtain additive error $\varepsilon_r$ in $\ell(B_r)$, assuming that the
  rank correction is supplied exactly, this requires
\begin{equation}
    O\!\left(
      \frac{\chg{D_r}(\Lambda_r^{\mathrm{sq}})^2}{\varepsilon_r}
      \log\frac{2}{\nu_r}
    \right)
  \label{eq:squared-route-ae-uses}
\end{equation}
  uses of the amplitude-estimation unitary and its inverse, where $\nu_r$ is
  the permitted failure probability.}
  If $P_{\mathrm{flag},r}$ prepares the flagged Dicke state in
  \eqref{eq:squared-route-flagged-trace-state}, including membership testing
  and copying, the same bound applies to uses of $P_{\mathrm{flag},r}$,
  $P_{\mathrm{flag},r}^{\dagger}$, and their associated amplitude-estimation
  reflections.  Their gate cost depends on the coherent membership-access
  model and is additional to the boundary-block query count below.

  \item {
  Sequential execution sums the costs over $r$.  With sufficient hardware,
  parallel execution can reduce the circuit depth to the largest single-degree
  depth, while the total gate and query counts remain additive.}

  \item {
  In terms of the approximation degrees, the resulting
  boundary-block-encoding query count is
  \begin{equation}
    O\!\left(
      \sum_{r=1}^{i+1}
      \frac{\chg{d_{\log,r}}\chg{D_r}(\Lambda_r^{\mathrm{sq}})^2}{\varepsilon_r}
      \log\frac{2}{\nu_r}
    \right).
    \label{eq:squared-route-query-count}
  \end{equation}
  }
\end{enumerate}

{\color{addition}
Substituting the value of $\chg{\eta_r}$ required by
\eqref{eq:squared-route-error-allocation} into \eqref{eq:squared-route-degree}
gives $\chg{d_{\log,r}}=\widetilde O_{c_\star}(\gamma_r^{-1})$, with the
suppressed logarithms depending on $1/\gamma_r$ and
$\chg{D_r}\Lambda_r^{\mathrm{sq}}/\varepsilon_r$, so
\eqref{eq:squared-route-query-count} becomes
\begin{equation}
  \widetilde O\!\left(
    \sum_{r=1}^{i+1}
    \frac{\chg{D_r}(\Lambda_r^{\mathrm{sq}})^2}{\gamma_r\varepsilon_r}
    \log\frac{2}{\nu_r}
  \right).
  \label{eq:squared-route-query-count-explicit}
\end{equation}
Up to logarithmic factors, this route has the same leading
$\chg{D_r}/(\gamma_r\varepsilon_r)$ dependence as the primary route, while
additionally requiring the rank correction.}

\paragraph{Why the rank correction appears.}
For this construction, write $A_r=\mathcal{B}_r/\sqrt N$ and let
$p_r^{\mathrm{SV}}(A_r)$ denote the operator obtained by applying $p_r$ to the
singular values of $A_r$.  Equation~\eqref{eq:squared-route-amplitude} then
shows explicitly that amplitude estimation returns the squared-norm trace
\begin{equation}
  \Tr\!\left[
    p_r^{\mathrm{SV}}(A_r)^\ast p_r^{\mathrm{SV}}(A_r)
  \right],
  \label{eq:squared-route-squared-trace}
\end{equation}
not $\Tr p_r^{\mathrm{SV}}(A_r)$.  It therefore requires
$\rank(\mathcal{B}_r)$ as an additional input or a separate rank-estimation
subroutine.  By contrast, the primary bottom-left construction prescribes a
nonzero kernel value, so the kernel and nonzero-spectrum contributions combine
as in \eqref{eq:complementary-trace-identity} and the rank cancels.  Both routes
use the same block-encoding, QSVT, and amplitude-estimation primitives, but the
primary route uses $\ket{\Phi_{D_r}}$, whereas this variant uses the flagged
valid-simplex component $\ket{\Phi_{f_r}}$.

\begin{remark}[Rank correction]
\label{rem:squared-amplitude}
The preceding clique-normalized correction has the following generic form.
Suppose QSVT implements the normalized polynomial
$q=p/\Lambda^{\mathrm{sq}}$, with $p$
approximating the squared-amplitude target in
\eqref{eq:squared-amplitude-overview}, and estimates the corresponding success
probability.  That probability is the squared norm
\begin{equation}
  p_{\mathrm{amp}}
  =
  \frac1D\Tr\!\left[q^{\mathrm{SV}}(A)^\ast
                         q^{\mathrm{SV}}(A)\right].
  \label{eq:squared-route-generic-amplitude}
\end{equation}
Here $q^{\mathrm{SV}}(A)$ denotes the operator obtained by applying $q$ to the
singular values of $A$ through QSVT.
It is not $\Tr q^{\mathrm{SV}}(A)$.  With $\alpha=\sqrt N$,
\begin{equation}
  D(\Lambda^{\mathrm{sq}})^2p_{\mathrm{amp}}
  \approx
  \rank(B)\log(\chg{c_\star}N)-\ell(B),
  \label{eq:squared-amplitude-rank}
\end{equation}
so the squared-amplitude estimator requires a separate rank correction.
The primary bottom-left construction instead assigns the complementary kernel
weight $a_\alpha$, producing the rank cancellation in
\eqref{eq:complementary-trace-identity} without such a subroutine.
If the supplied rank has absolute error
$\varepsilon_{\mathrm{rank},r}$, it contributes at most
$\varepsilon_{\mathrm{rank},r}\log(c_\star N)$ to the log-error budget,
in addition to the cost of obtaining that estimate.  The choice between the
primary and squared-amplitude estimators is independent of the algebraic choice
between alternating boundary pseudodeterminants and a rooted reduced
determinant.
\end{remark}

\section{Analytic properties and weighted polynomial approximation}
\label{app:complementary-polynomial-proof}

This appendix proves \cref{lem:complementary-log-polynomial} and the
bounds relating polynomial degree and weighted norm. The approximation
polynomial is $p$, and the QSVT polynomial is $q=p/\Lambda$, with
$\Lambda$ defined in \eqref{eq:weighted-normalization}. An even degree
bound $d_q$ for $p$ gives the QSVT sequence length $n_{\rm Q}=d_q+1$.
Throughout,
\[
 a_\alpha=\log\alpha^2\geq0,\qquad
 t_\gamma=\log(1/\gamma),\qquad
 s_{\alpha,\gamma}=a_\alpha+t_\gamma>0,\qquad
 b_\alpha=\sqrt{a_\alpha}.
\]
For a real polynomial $p$, write
\begin{equation}
 M(p):=\max_{x\in[-1,1]}\sqrt{1-x^2}\,|p(x)|.
 \label{eq:weighted-polynomial-norm}
\end{equation}
We write $\|p\|_{[-1,1]}=\max_{x\in[-1,1]}|p(x)|$ for the uniform
norm. The implemented polynomial satisfies $M(q)=M(p)/\Lambda$;
this is the maximum magnitude of the bottom-left QSVT amplitude,
including the factor $\sqrt{1-x^2}$.

\subsection{Analytic continuation and uniform polynomial approximation}

The function
\[
 h(x)=\sqrt{\frac{\log(1/x^2)}{1-x^2}}
 \quad(0<|x|<1),\qquad h(\pm1)=1,
\]
is even, with the endpoint values defined by continuity. On the right half-plane,
use the principal logarithm to define
\[
 F(z)=\frac{2\operatorname{Log}z}{z^2-1},\qquad F(1)=1.
\]
The singularity at $1$ is removable. Away from $1$, neither the numerator
nor the denominator vanishes in this half-plane. Hence $F$ is
holomorphic and nonzero there. Since the half-plane is simply connected,
$F$ has a holomorphic square root $h_+$, uniquely specified by
$h_+(1)=1$. On the positive real axis $F$ is positive, so $h_+$ agrees
with $h$ on $(0,1]$. These observations also justify the endpoint
definition in \eqref{eq:complementary-log-target}.

\begin{lemma}[Uniform approximation]
\label{lem:complementary-uniform-approximation}
For $0<\xi<1/4$, there is a real even polynomial $p$ satisfying
\eqref{eq:complementary-poly-outer} and
\eqref{eq:complementary-poly-kernel}, with
\begin{equation}
 \|p\|_{[-1,1]}\leq2\sqrt{1+s_{\alpha,\gamma}},\qquad
 \deg p=O\!\left(
  \gamma^{-1}\log\frac{C\sqrt{1+s_{\alpha,\gamma}}}{\gamma\xi}
 \right).
 \label{eq:complementary-uniform-approximation}
\end{equation}
Here and below, $C>1$ and the constants implicit in $O(\cdot)$ are
universal.
\end{lemma}

\begin{proof}
Apply the bounded multi-interval approximation theorem
\cite[Theorem~68]{Gilyen2019QuantumArithmetics} to
\[
 I_-=[-1,-3\gamma/4],\qquad
 I_0=[-\gamma/2,\gamma/2],\qquad
 I_+=[3\gamma/4,1],\qquad \delta=\gamma/32.
\]
All interval radii are at least $\delta$, and the gaps between consecutive
intervals are $8\delta$. Thus the minimum of the neighborhood parameters
and these gaps, as defined in the theorem, is $\delta$.
The center and radius of $I_+$ are
\[
 m_+=\frac{1+3\gamma/4}{2},\qquad
 r_+=\frac{1-3\gamma/4}{2}.
\]
On the closed disk $|z-m_+|\leq R_+:=r_++2\delta$, one has
$\operatorname{Re}z\geq u_\gamma:=11\gamma/16>0$.
The identity
\[
 \frac{\operatorname{Log}z}{z-1}
 =\int_0^1\frac{dv}{1-v+vz}
\]
holds throughout the right half-plane, including the continuous value
at $z=1$. Thus
\begin{align*}
 |F(z)|
 &\leq\frac{2}{1+u_\gamma}
       \int_0^1\frac{dv}{1-v+vu_\gamma}
   =\frac{2\log(1/u_\gamma)}{1-u_\gamma^2}\\
 &\leq\frac{512}{135}\log\frac{16}{11\gamma}
   <4\log\frac{16}{11\gamma}
   \leq4(1+s_{\alpha,\gamma}).
\end{align*}
Set $L_0=8\sqrt{1+s_{\alpha,\gamma}}$ and use the local functions
$h_+(z)/L_0$ on the positive interval, $h_+(-z)/L_0$ on the negative
interval, and $b_\alpha/L_0$ on the central interval. Each function
has modulus at most $1/4$ on its enlarged interval and on the
corresponding disk used below.

If $a_k$ are the Taylor coefficients of $h_+/L_0$ about $m_+$,
Cauchy's estimate gives
\[
 \sum_{k\geq0}(r_++\delta)^k|a_k|
 \leq\frac14\sum_{k\geq0}
       \left(\frac{r_++\delta}{R_+}\right)^k
 =\frac{R_+}{4\delta}
 =\frac4\gamma-\frac52\leq\frac4\gamma.
\]
The reflected function satisfies the same estimate, and the corresponding
weighted sum of absolute Taylor coefficients for the constant function
is at most $1/4$. Hence $B_*=4/\gamma$ bounds these sums for all three
functions. With three intervals, the approximation error
\[
 \eta_0=\min\{\xi/L_0,\gamma/24\}
\]
satisfies the theorem's admissibility condition
$\eta_0\leq(2B_*\cdot3)^{-1}$. The theorem gives a polynomial $P$
with uniform norm at most $1/4$ on $[-1,1]$, error at most $\eta_0$ on each
interval, and degree
\[
 O\!\left(\gamma^{-1}\log\frac{12}{\gamma\eta_0}\right)
 =O\!\left(\gamma^{-1}
        \log\frac{C\sqrt{1+s_{\alpha,\gamma}}}{\gamma\xi}\right).
\]
Taking coefficientwise real parts and averaging with reflection gives
\[
 p(x)=\frac{L_0}{2}
       \bigl(\operatorname{Re}P(x)+\operatorname{Re}P(-x)\bigr).
\]
Taking real parts and symmetrizing preserve the approximation bounds
and do not increase the degree or uniform norm. Multiplication by $L_0$
then gives the stated bounds for $p$.
\end{proof}

\subsection{Fixed-interval auxiliary approximations}

\begin{lemma}[Auxiliary polynomials]
\label{lem:complementary-fixed-interval}
Define
\[
 R(x)=\frac{\sqrt{-2\log x}}{\arccos x}
 \quad(3/4\leq x<1),\qquad R(1)=1.
\]
For $0<\eta<1/4$, there are real even polynomials $r_\eta,c_\eta$
of degree $O(\log(C/\eta))$ such that
\begin{align*}
 |r_\eta(x)-R(|x|)|&\leq\eta,&
 |c_\eta(x)|&\leq\eta
 &&(3/4\leq|x|\leq1),\\
 |r_\eta(x)|&\leq\eta,&
 |c_\eta(x)-1|&\leq\eta
 &&(|x|\leq1/2),
\end{align*}
and
\[
 \|r_\eta\|_{[-1,1]}\leq2,\qquad
 \|c_\eta\|_{[-1,1]}\leq1.
\]
\end{lemma}

\begin{proof}
We first show that $R$ has an analytic continuation to a neighborhood of $1$. Put
\[
 L(z)=\frac{-\operatorname{Log}z}{1-z},\qquad L(1)=1,
 \qquad
 K(w)=\sum_{k\geq0}\kappa_k w^k,
 \qquad
 \kappa_k=\frac{\binom{2k}{k}}{4^k(2k+1)}.
\]
On $|z-1|<1/2$, the integral representation from the preceding proof
shows that $L$ is holomorphic, has positive real part, and obeys
$|L(z)|\leq2$. The coefficients $\kappa_k$ are positive and sum
to $\pi/2$, by the Taylor series for $\arcsin$. Therefore
\[
 \left|K\!\left(\frac{1-z}{2}\right)-1\right|
 \leq\frac14\left(\frac\pi2-1\right)<\frac17.
\]
It follows that
\[
 \frac{\sqrt{L(z)}}{K((1-z)/2)}
\]
is holomorphic and bounded by $2$ on that disk. On $[3/4,1]$ it
equals $R(x)$: the identity
\[
 (\arccos x)^2=2(1-x)K((1-x)/2)^2
\]
follows from $\arccos x=2\arcsin\sqrt{(1-x)/2}$, with the value
at $x=1$ obtained by continuity.

Apply \cite[Theorem~68]{Gilyen2019QuantumArithmetics} to the three
intervals
\[
 [-1,-3/4],\qquad[-1/2,1/2],\qquad[3/4,1],
 \qquad\delta=1/64.
\]
For the positive interval, the disk centered at $7/8$ with radius
$5/32$ lies in $|z-1|<1/2$ and has radius equal to the interval
radius plus $2\delta$. The same holds after reflection for the negative
interval. Use the analytic continuation of $R/4$, its reflection, and
the zero function to construct $r_\eta/4$. Their uniform norms on the
corresponding disks are at most $1/2$, and the weighted sums of absolute
Taylor coefficients are at most $(1/2)(5/32)/\delta=5$. Use the constant
functions zero, $1/2$, zero to construct $c_\eta/2$; the corresponding
coefficient sums are at most $1/2$.
The intervals and neighborhood widths are fixed, and the bounds on the
weighted coefficient sums are universal constants. Choose
errors $\eta_r=\min\{\eta/4,1/30\}$ and $\eta_c=\eta/2$.
These errors satisfy the theorem's condition for three intervals with
bounds $5$ and $1/2$ on the weighted sums of absolute Taylor
coefficients, respectively, and give degree $O(\log(C/\eta))$. Both
rescaled auxiliary polynomials have uniform norm at most $1/2$ on $[-1,1]$.
Rescale and symmetrize
as in the preceding proof to obtain the asserted bounds.
\end{proof}

\subsection{Polynomial construction with a bound on the weighted norm}

\begin{proof}[Proof of \cref{lem:complementary-log-polynomial}]
If $s_{\alpha,\gamma}\geq1/64$,
\cref{lem:complementary-uniform-approximation} gives
\[
 M(p)\leq2\sqrt{1+s_{\alpha,\gamma}}
 \leq2\sqrt{65}\sqrt{s_{\alpha,\gamma}}
 <17\sqrt{s_{\alpha,\gamma}},
\]
with the required degree. It remains to consider
$0<s_{\alpha,\gamma}<1/64$.

Choose the auxiliary order $n$ to be the largest positive odd integer
satisfying $n\leq1/(4\sqrt{s_{\alpha,\gamma}})$. Since the upper bound exceeds
two, this choice satisfies
\begin{equation}
 \frac1{12\sqrt{s_{\alpha,\gamma}}}
 \leq n\leq\frac1{4\sqrt{s_{\alpha,\gamma}}}.
 \label{eq:weighted-chebyshev-order}
\end{equation}
Let $\theta_\gamma=\arccos\gamma$. Integrating
$\tan\theta\geq\theta$ gives $\log(\sec\theta)\geq\theta^2/2$,
and therefore
\[
 \theta_\gamma\leq\sqrt{2t_\gamma}
 \leq\sqrt{2s_{\alpha,\gamma}},\qquad
 n\theta_\gamma\leq\frac{\sqrt2}{4}<\frac\pi4.
\]
Also $\gamma\geq e^{-s_{\alpha,\gamma}}>3/4$.

Let $T_n,U_{n-1}$ be the Chebyshev polynomials of the first and
second kind, respectively, so
\[
 T_n(\cos\theta)=\cos(n\theta),\qquad
 \sin\theta\,U_{n-1}(\cos\theta)=\sin(n\theta).
\]
Using the coefficients $\kappa_k$ from the preceding proof, define
\begin{align}
 F_m(y)&=\sum_{k=0}^m\kappa_k y^{2k+1},
 \label{eq:weighted-arcsine-truncation}\\
 b_{n,m}(x)&=\frac{U_{n-1}(x)}n
       \sum_{k=0}^m\kappa_k\bigl(1-T_n(x)^2\bigr)^k.
 \label{eq:weighted-chebyshev-polynomial}
\end{align}
Because $n$ is odd, $b_{n,m}$ is even. Its degree is at most
$(2m+1)n-1$, and
\begin{equation}
 \sin\theta\,b_{n,m}(\cos\theta)
 =\frac{F_m(\sin(n\theta))}{n}.
 \label{eq:weighted-chebyshev-identity}
\end{equation}
Since $|T_n(x)|\leq1$ and $|U_{n-1}(x)|\leq n$ on $[-1,1]$,
and $\sum_{k\geq0}\kappa_k=\pi/2$, we have
\begin{equation}
 \|b_{n,m}\|_{[-1,1]}\leq\pi/2,
 \qquad M(b_{n,m})\leq\frac{\pi}{2n}.
 \label{eq:weighted-chebyshev-bounds}
\end{equation}

For $0\leq\theta\leq\theta_\gamma$, let $y=\sin(n\theta)$.
Then $0\leq y\leq1/\sqrt2$ and $\arcsin y=n\theta$. Positivity
of the Taylor coefficients gives
\[
 0\leq\arcsin y-F_m(y)
 \leq\frac\pi2 y^{2m+3}
 \leq\frac\pi2 y\,2^{-(m+1)}.
\]
Divide by $n\sin\theta$ and use
$|\sin(n\theta)|\leq n\sin\theta$ on this interval. We obtain
\begin{equation}
 \left|b_{n,m}(x)-\frac{\arccos x}{\sqrt{1-x^2}}\right|
 \leq\frac\pi2\,2^{-(m+1)}
 \qquad(\gamma\leq x\leq1).
 \label{eq:weighted-chebyshev-error}
\end{equation}
At $x=1$, both functions equal one by continuity.

Take the auxiliary polynomials from
\cref{lem:complementary-fixed-interval} and set
\begin{equation}
 p(x)=b_{n,m}(x)r_\eta(x)+b_\alpha c_\eta(x).
 \label{eq:weighted-polynomial-construction}
\end{equation}
For $\gamma\leq x\leq1$,
\[
 h(x)=R(x)\frac{\arccos x}{\sqrt{1-x^2}},\qquad |R(x)|\leq2.
\]
Consequently, $|p(x)-h(x)|$ on $[\gamma,1]$ is at most
\[
 (\pi/2+b_\alpha)\eta+\pi\,2^{-(m+1)}.
\]
On $|x|\leq\gamma/2\leq1/2$, the error from the desired value
$b_\alpha$ is at most $(\pi/2+b_\alpha)\eta$. Choose
\[
 \eta=\frac{\xi}{4(\pi/2+b_\alpha)}
\]
and an integer $m=O(\log(C/\xi))$ such that
$\pi2^{-(m+1)}\leq\xi/2$. This proves both approximation bounds;
evenness gives the bound on $[-1,-\gamma]$. The weighted norm
satisfies
\[
 M(p)\leq\frac\pi n+b_\alpha
 \leq(12\pi+1)\sqrt{s_{\alpha,\gamma}}
 <40\sqrt{s_{\alpha,\gamma}}.
\]
Finally,
\begin{align*}
 \deg p
 &\leq(2m+1)n-1+O(\log(C/\eta))\\
 &=O\!\left[
  \left(\gamma^{-1}+s_{\alpha,\gamma}^{-1/2}\right)
  \log\frac{C(1+\sqrt{s_{\alpha,\gamma}})}{\gamma\xi}
 \right].
\end{align*}
The case $\gamma=1$ with $a_\alpha>0$ is included by the endpoint
limits above. Taking an even integer upper bound $d_q$ on the degree
gives the convention used in \cref{lem:complementary-log-polynomial}.
\end{proof}

If $s_{\alpha,\gamma}=0$, then $\alpha=\gamma=1$. Every positive
singular value of $B$ equals one, so $\ell(B)=0$. This trivial
estimation case is excluded from the approximation lemma: a zero
weighted norm would force the polynomial to vanish identically and is
incompatible with approximating $h(1)=1$ to error less than one.

\subsection{Lower bounds on degree and weighted norm}

\begin{proposition}[Relation between degree and weighted norm]
\label{prop:weighted-degree-normalization}
Every real polynomial $p$ of degree at most $d$ satisfies
\begin{equation}
 |p(1)|\leq(d+1)M(p).
 \label{eq:weighted-degree-normalization}
\end{equation}
In particular, the approximation condition on $\gamma\leq|x|\leq1$
in \eqref{eq:complementary-poly-outer} implies
$(d+1)M(p)\geq1-\xi$.
\end{proposition}

\begin{proof}
The function $\mathcal T(\theta)=\sin\theta\,p(\cos\theta)$ is a
trigonometric polynomial of degree at most $d+1$, with
$\|\mathcal T\|_\infty=M(p)$ and $\mathcal T'(0)=p(1)$.
Bernstein's trigonometric inequality
\cite{QueffelecZarouf2019} gives
\[
 |p(1)|\leq\|\mathcal T'\|_\infty
 \leq(d+1)\|\mathcal T\|_\infty=(d+1)M(p).
\]
The last assertion follows from $|p(1)-1|\leq\xi$.
\end{proof}

For fixed $\xi<1$, a weighted bound $M(p)=O(\sqrt{s_{\alpha,\gamma}})$
therefore requires $d+1=\Omega(s_{\alpha,\gamma}^{-1/2})$ as
$s_{\alpha,\gamma}\downarrow0$. This explains the additional degree
term in \cref{lem:complementary-log-polynomial}.

For a polynomial $p$ of degree at most $d$, a related lower bound
follows by approximating $\log(1/\gamma^2)$ with
$(1-\gamma^2)p(\gamma)^2$. If $0<\gamma<1$ and
\[
 \left|(1-\gamma^2)p(\gamma)^2-2t_\gamma\right|
 \leq\varepsilon_w,\qquad 0\leq\varepsilon_w<2t_\gamma,
\]
then, with $\mathcal T(\theta)=\sin\theta\,p(\cos\theta)$,
we have $\mathcal T(0)=0$. This identity and Bernstein's inequality imply
\begin{equation}
 (d+1)M(p)
 \geq\frac{\sqrt{2t_\gamma-\varepsilon_w}}{\arccos\gamma}
 \geq\sqrt{1-\frac{\varepsilon_w}{2t_\gamma}}.
 \label{eq:weighted-signal-degree-bound}
\end{equation}
Indeed, integrate $\mathcal T'$ between $0$ and $\arccos\gamma$
and then use $(\arccos\gamma)^2\leq2t_\gamma$. The assumption
$\varepsilon_w<2t_\gamma$ is essential: an allowed error of at least
$2t_\gamma$ would permit a zero approximation at $x=\gamma$. These
bounds apply to polynomials satisfying the stated approximation
conditions; they are not lower bounds for all quantum algorithms.

The approximation requirements also impose a lower bound on the
weighted norm itself:
\begin{equation}
 M(p)\geq\max\left\{
  (b_\alpha-\xi)_+,\,
  \bigl(\sqrt{2t_\gamma}-\sqrt{1-\gamma^2}\,\xi\bigr)_+
 \right\},
 \label{eq:weighted-normalization-lower-bound}
\end{equation}
where $u_+=\max\{u,0\}$. These are the constraints at $x=0$ and
$x=\gamma$, respectively. In particular, if
$\xi\leq\sqrt{s_{\alpha,\gamma}}/4$, then
$M(p)=\Omega(\sqrt{s_{\alpha,\gamma}})$: at least one of
$a_\alpha,t_\gamma$ is at least $s_{\alpha,\gamma}/2$, and the
corresponding term in \eqref{eq:weighted-normalization-lower-bound}
is a positive universal multiple of $\sqrt{s_{\alpha,\gamma}}$.
Thus the bound $M(p)=O(\sqrt{s_{\alpha,\gamma}})$ is optimal up to
constant factors under this accuracy condition. With this bound on
$M(p)$, the $s_{\alpha,\gamma}^{-1/2}$ term in the degree bound is
also necessary, up to logarithmic factors, for this approximation
problem.

\section{QSVT construction and implementation errors}
\label{app:quantum-circuit-analysis}

\subsection{Proof of the bottom-left QSVT construction}
\label{app:bottom-left-qsvt-proof}

We prove \cref{lem:bottom-left-qsvt}.
Let
\[
 J_B\ket{\psi}:=\ket{0^{n_B}}\ket{\psi},
 \qquad
 \Pi_B:=J_BJ_B^\dagger,
 \qquad
 F_q:=q\!\left(\sqrt{A^\dagger A}\right).
\]
Since $q$ is even, $F_q$ is a polynomial in $A^\dagger A$.
For a singular value decomposition
\[
 B=\sum_{j=1}^{\rho}\sigma_j\ket{u_j}\!\bra{v_j},
\]
where $\rho=\rank B$ and $x_j=\sigma_j/\alpha$,
\begin{equation}
 V_R^\dagger F_qV_R
 =\sum_{j=1}^{\rho}q(x_j)\ket{v_j}\!\bra{v_j}
  +q(0)\Pi_{\ker B},
 \label{eq:complementary-right-transform}
\end{equation}
where $\Pi_{\ker B}$ is the orthogonal projector onto $\ker B$.
Also define
\begin{equation}
 C_B:=(I-\Pi_B)U_BJ_B,\qquad
 C_B^\dagger C_B=I-A^\dagger A.
 \label{eq:complementary-boundary-map}
\end{equation}
We establish the stronger operator identity
\begin{equation}
 (\bra0_c\otimes(I-\Pi_B))\mathcal V_q
 (\ket0_c\otimes J_B)
 =\mathrm{i}C_BF_q.
 \label{eq:direct-bottom-left-block}
\end{equation}

Set
\[
 \omega(x):=\sqrt{1-x^2},\qquad
 Z:=\operatorname{diag}(1,-1),\qquad
 W(x):=
 \begin{pmatrix}
  x&\mathrm{i}\omega(x)\\
  \mathrm{i}\omega(x)&x
 \end{pmatrix}.
\]
The polynomial completion theorem
\cite[Theorem~5]{Gilyen2019QuantumArithmetics}, applied to
$\widetilde P=0$ and $\widetilde Q=q$, gives complex polynomials
$\mathcal P,\mathcal Q$ satisfying
\[
 \operatorname{Re}\mathcal P=0,\qquad
 \operatorname{Re}\mathcal Q=q,\qquad
 |\mathcal P(x)|^2+(1-x^2)|\mathcal Q(x)|^2=1
 \quad(x\in[-1,1]).
\]
Here $\mathcal P$ is odd of degree at most
$n_{\rm Q}=d_q+1$, and $\mathcal Q$ is even of degree at most
$n_{\rm Q}-1$.

The QSP characterization theorem
\cite[Theorem~3]{Gilyen2019QuantumArithmetics} supplies a phase
list $\boldsymbol\phi=(\phi_0,\ldots,\phi_{n_{\rm Q}})$ such that
\begin{equation}
 S_{\boldsymbol\phi}(x)
 :=e^{\mathrm{i}\phi_0Z}
   \prod_{j=1}^{n_{\rm Q}}
   \bigl(W(x)e^{\mathrm{i}\phi_jZ}\bigr)
 =
 \begin{pmatrix}
  \mathcal P(x)&\mathrm{i}\omega(x)\mathcal Q(x)\\
  \mathrm{i}\omega(x)\overline{\mathcal Q(x)}
    &\overline{\mathcal P(x)}
 \end{pmatrix}.
 \label{eq:bottom-left-qsp-completion}
\end{equation}
For real $x$,
$\overline{W(x)}=ZW(x)Z$, where the bar denotes entrywise
complex conjugation. Therefore
\[
 S_{-\boldsymbol\phi}(x)
 =Z\overline{S_{\boldsymbol\phi}(x)}Z,
\]
and the bottom-left entry of
$(S_{\boldsymbol\phi}(x)+S_{-\boldsymbol\phi}(x))/2$
is $\mathrm{i}\omega(x)q(x)$.

To implement these sequences, put $Z_B=2\Pi_B-I$ and define
\[
 \mathcal W_+
 =\mathrm{i}e^{-\mathrm{i}\pi Z_B/4}
 U_Be^{-\mathrm{i}\pi Z_B/4},
 \qquad
 \mathcal W_-
 =\mathrm{i}e^{-\mathrm{i}\pi Z_B/4}
 U_B^\dagger e^{-\mathrm{i}\pi Z_B/4}.
\]
For a singular pair
$A\ket v=x\ket u$, $A^\dagger\ket u=x\ket v$,
with $0<x<1$, write
$\ket{\widehat v}=J_B\ket v$,
$\ket{\widehat u}=J_B\ket u$, and
\[
 \ket{v_\perp}
 =\frac{U_B^\dagger\ket{\widehat u}
            -x\ket{\widehat v}}{\omega(x)},
 \qquad
 \ket{u_\perp}
 =\frac{U_B\ket{\widehat v}
            -x\ket{\widehat u}}{\omega(x)}.
\]
These are unit vectors in $\operatorname{ran}(I-\Pi_B)$.
The matrix of $U_B$ from
$(\widehat v,v_\perp)$ to $(\widehat u,u_\perp)$, and that of
$U_B^\dagger$ in the reverse direction, is
\[
 R(x)=
 \begin{pmatrix}
  x&\omega(x)\\
  \omega(x)&-x
 \end{pmatrix}.
\]
Since
$\mathrm{i}e^{-\mathrm{i}\pi Z/4}R(x)e^{-\mathrm{i}\pi Z/4}
=W(x)$,
the alternating sequence
\[
 \mathcal S_{\boldsymbol\phi}
 =e^{\mathrm{i}\phi_0Z_B}\mathcal W_+
  e^{\mathrm{i}\phi_1Z_B}\mathcal W_-
  e^{\mathrm{i}\phi_2Z_B}\cdots
  \mathcal W_+e^{\mathrm{i}\phi_{n_{\rm Q}}Z_B}
\]
implements \eqref{eq:bottom-left-qsp-completion} on the paired
subspaces. The rightmost operation acts first.
Since $n_{\rm Q}$ is odd, $\mathcal W_+$ occurs at both ends,
and the sequence uses $n_{\rm Q}$ block-encoding queries.

Let $H_c$ be the Hadamard gate on the control ancilla, and define
\begin{equation}
 \mathcal V_q
 =(H_c\otimes I)
 \left(
 \ket0\!\bra0_c\otimes\mathcal S_{\boldsymbol\phi}
 +\ket1\!\bra1_c\otimes\mathcal S_{-\boldsymbol\phi}
 \right)
 (H_c\otimes I).
 \label{eq:bottom-left-selector}
\end{equation}
The two branches differ only in the signs of the QSP phases.
The fixed $-\pi/4$ rotations and all block-encoding calls are
identical, so the total query count remains $n_{\rm Q}$.
No entrywise-conjugate oracle is required.

For $0<x<1$, the complementary output direction is
$C_B\ket v/\omega(x)$, giving
\eqref{eq:direct-bottom-left-block}.
For $A\ket v=0$, the vector $U_BJ_B\ket v$ lies entirely in
$\operatorname{ran}(I-\Pi_B)$.
The alternating queries map $J_B\ket v$ to $U_BJ_B\ket v$
and back. If
\[
 \theta_{\boldsymbol\phi}
 =\sum_{j=0}^{n_{\rm Q}}(-1)^{j+1}\phi_j,
\]
then
\[
 \mathcal S_{\boldsymbol\phi}J_B\ket v
 =\mathrm{i}^{n_{\rm Q}}
 e^{\mathrm{i}\theta_{\boldsymbol\phi}}U_BJ_B\ket v.
\]
Averaging the two phase lists gives
\[
 \mathrm{i}^{n_{\rm Q}}\cos\theta_{\boldsymbol\phi}
 U_BJ_B\ket v
 =\mathrm{i}q(0)U_BJ_B\ket v,
\]
by the scalar QSP identity at zero.
For $x=1$, the sequence remains in
$\operatorname{ran}\Pi_B$, and its complementary output is zero.
This proves \eqref{eq:direct-bottom-left-block} in every case.

A single odd-length QSP sequence has $|\mathcal Q(0)|=1$.
The control ancilla permits the required value of $q(0)$ by
extracting $\operatorname{Re}\mathcal Q$.
The construction implements $\mathrm{i}C_BF_q$ directly;
it does not require a separate implementation of $F_q$.
The circuit is unitary and can be inverted by reversing its gates.

\subsection{Implementation errors}
\label{subsec:implementation-errors}

We derive sufficient tolerances for
$|P_{\rm impl}-P_q|\leq\delta_{\rm impl}$.
For the perturbation analysis, define
\begin{equation}
 \begin{split}
 K_q&:=\mathrm{i}C_BF_qV_R,\\
 E_q&:=K_q^\dagger K_q
 =V_R^\dagger F_q^\dagger(I-A^\dagger A)F_qV_R.
 \end{split}
 \label{eq:complementary-success-effect}
\end{equation}
The probability for a normalized input vector $\ket v$ is
$\|K_q\ket v\|^2$.
Thus an operator-norm bound on the difference between two
operators of the form $K^\dagger K$ bounds the corresponding
probability difference for every input state.

Let $U_{\rm ref}$ be an exact unitary with normalized encoded
matrix $A_{\rm ref}$, where
\[
 \|A_{\rm ref}-A\|\leq\delta_B
\]
on the full system register.
Both $A_{\rm ref}$ and $A$ are contractions.
No singular-value gap is required for $A_{\rm ref}$ in this
perturbation argument.

The operator $E_q$ depends only on the normalized encoded
matrix, although the complementary output vectors can depend
on its unitary dilation. We may therefore compare the canonical
unitary dilations
\[
 \mathcal J(A)=
 \begin{pmatrix}
  A&\sqrt{I-AA^\dagger}\\
  \sqrt{I-A^\dagger A}&-A^\dagger
 \end{pmatrix}
\]
and $\mathcal J(A_{\rm ref})$.
These matrices are used only in the proof.

For positive semidefinite operators $X,Y$,
\[
 \|\sqrt X-\sqrt Y\|\leq\sqrt{\|X-Y\|}.
\]
Indeed, with $t=\|X-Y\|$, operator monotonicity gives
$\sqrt X\leq\sqrt{Y+tI}\leq\sqrt Y+\sqrt t\,I$,
and interchanging $X,Y$ gives the reverse bound.
Since
\[
 \|A_{\rm ref}A_{\rm ref}^\dagger-AA^\dagger\|\leq2\delta_B,
 \qquad
 \|A_{\rm ref}^\dagger A_{\rm ref}-A^\dagger A\|\leq2\delta_B,
\]
we have
\[
 \|\mathcal J(A_{\rm ref})-\mathcal J(A)\|
 \leq\delta_B+\sqrt{2\delta_B}.
\]

Apply the same phase lists to both dilations.
Telescoping the $n_{\rm Q}=d_q+1$ block-encoding queries bounds
the difference between the circuit unitaries by
$n_{\rm Q}(\delta_B+\sqrt{2\delta_B})$.
Compression does not increase the norm, and contractions
$K_1,K_2$ satisfy
\[
 \|K_1^\dagger K_1-K_2^\dagger K_2\|
 \leq2\|K_1-K_2\|.
\]
Writing $E_{q,\rm ref}$ for the operator obtained with
$A_{\rm ref}$ gives
\begin{equation}
 \|E_{q,\rm ref}-E_q\|
 \leq2n_{\rm Q}(\delta_B+\sqrt{2\delta_B})
 \leq5n_{\rm Q}\sqrt{\delta_B},
 \qquad 0\leq\delta_B\leq1.
 \label{eq:reference-effect-error}
\end{equation}

Suppose the implemented QSVT circuit differs from the reference
circuit by at most $e_{\rm circ}$ in operator norm.
Its corresponding operator $E_{\rm impl}$ then satisfies
\begin{equation}
 \|E_{\rm impl}-E_q\|
 \leq5n_{\rm Q}\sqrt{\delta_B}+2e_{\rm circ}.
 \label{eq:implementation-effect-bound}
\end{equation}
If the prepared state differs from the ideal input state by
trace distance at most $\delta_\Phi$, its contribution to the
probability error is at most $\delta_\Phi$.
Consequently, a sufficient implementation bound is
\[
 |P_{\rm impl}-P_q|
 \leq
 \delta_\Phi+5n_{\rm Q}\sqrt{\delta_B}+2e_{\rm circ}.
\]

For a chosen tolerance $\bar\delta_E>0$, it suffices to impose
\begin{equation}
 \delta_B\leq
 \min\left\{
  1,\left(\frac{\bar\delta_E}{10n_{\rm Q}}\right)^2
 \right\},
 \qquad
 e_{\rm circ}\leq\frac{\bar\delta_E}{4}.
 \label{eq:implementation-tolerances}
\end{equation}
Taking
\[
 \bar\delta_E=\frac{\varepsilon_{\mathrm n}}{5\Lambda^2},
 \qquad
 \delta_\Phi\leq
 \frac{\varepsilon_{\mathrm n}}{5\Lambda^2}
\]
gives the implementation tolerance in
\eqref{eq:algorithm-error-budget}.

If each synthesized block-encoding call has operator-norm
error at most $\delta_U$ relative to the corresponding reference
call, then
\[
 e_{\rm circ}\leq n_{\rm Q}\delta_U+\sum_j\delta_j,
\]
where $\delta_j$ bounds the error in each remaining synthesized
unitary factor, including phase rotations and control-ancilla
gates. These bounds apply to controlled calls whenever used.
If $\delta_U$ is measured directly relative to an ideal
encoding of $A$, set $\delta_B=0$ to avoid counting the same
perturbation twice.

Errors in evaluating
$a_\alpha-\Lambda^2\widehat P$ contribute to $\chi$.
Errors in polynomial coefficients or QSP phases must instead
satisfy the polynomial-approximation or circuit-synthesis
tolerances. Numerical bounds used to choose accuracies and
query counts are evaluated with conservative rounding.

Amplitude estimation uses one fixed implemented circuit and
its exact inverse. Its systematic probability error is therefore
included in $\delta_{\rm impl}$ and is not multiplied by
$M_{\rm AE}$.
For synthesis errors in the additional reflections and control
operations of amplitude estimation, first reduce the ideal
failure probability to $\nu/2$. Require the sum of their
operator-norm errors over all repetitions to be at most $\nu/2$.
Telescoping bounds the change in the output distribution by
$\nu/2$, so the total failure probability is at most $\nu$.
Errors already included in the fixed circuit and its inverse
are not counted again.

The gate complexity in Eq.~\eqref{eq:normalized-gates} is evaluated
at these accuracies. With reusable work qubits, a reflection
about the block-encoding ancilla-zero subspace uses
$O(n_B+1)$ elementary gates.
The analysis counts oracle queries and quantum gates; it does
not include stochastic execution noise or quantum
error-correction overhead.


\begin{table}[htbp]
\centering\small
\begin{tabular}{@{}L{0.28\linewidth}L{0.66\linewidth}@{}}
\toprule
Parameter & Meaning\\
\midrule
$D_r=\binom{N}{r+1}$
& dimension of the Hamming-weight-$(r+1)$ input space\\
$\alpha_r,\gamma_r$
& block-encoding normalization and supplied normalized
  singular-value gap\\
$a_r,s_r,\Lambda_r$
& $a_r=\log N$, $s_r=\log N+\log(1/\gamma_r)$, and
  polynomial scaling factor $\Lambda_r=160\sqrt{s_r}$\\
$S_{\ast,r}$
& supplied upper bound on
  $S(\mathcal B_r)=a_r-\ell_r(X)/D_r$;
  integrality permits $S_{\ast,r}=\log N$\\
$d_{q,r}$
& even upper bound on the polynomial degree\\
$n_{{\rm Q},r}=d_{q,r}+1$
& block-encoding queries per application of $\mathcal A_q$\\
$M_{{\rm AE},r}$
& total applications of $\mathcal A_q$ or its inverse\\
$G_{B,r},G_{\Phi,r}$
& gate complexities of controlled block-encoding queries
  and initial-state preparation\\
$G_{\mathrm{mem},r}$
& gate complexity of the controlled membership oracle\\
$\xi_r,\delta_{{\rm impl},r}$
& polynomial approximation and implementation probability
  error bounds\\
$\delta_{{\rm AE},r},\chi_r$
& amplitude-estimation error and classical arithmetic error
  in the normalized output\\
$\delta_{B,r},\delta_{U,r},\delta_{\Phi,r}$
& errors in the normalized encoded matrix, synthesized
  block-encoding unitary, and prepared state\\
$\varepsilon_r,\nu_r$
& unnormalized additive error and failure probability
  in boundary degree $r$\\
$\varepsilon_{\rm arith}$
& arithmetic error for evaluating $\ell_0$ and
  combining estimates from different degrees\\
\bottomrule
\end{tabular}
\caption{Parameters used in the query and gate complexity bounds.
Controlled operations, inverses, state preparation, and the
additional operations required by amplitude estimation are
included in Eq.~\eqref{eq:normalized-gates}.}
\label{tab:resource-ledger}
\end{table}


\section{Implementation details}
\label{app:implementation-details}

This appendix provides the clique-boundary compression, the error analysis used
in \cref{cor:quantum-critical-order}, and the resource accounting for
\cref{cor:join-resource-comparison}.

\subsection{Unitary compression of the clique boundary}
\label{app:clique-compression}

The normalization in \eqref{eq:mcardle-normalization} follows from the
following unitary compression. On the full occupation register define
\begin{equation}
 U_{\mathrm{fer}}:=\frac{\partial+\partial^\dagger}{\sqrt N}
 =\frac1{\sqrt N}\sum_{v=0}^{N-1}Z_0\cdots Z_{v-1}X_v.
 \label{eq:fermionic-boundary-unitary}
\end{equation}
Here $X_v,Z_v$ are Pauli operators and $\partial$ is the augmented simplicial boundary on the direct sum of all subset spaces, with signs determined by the fixed vertex order.
The summands are Hermitian, square to the identity, and anticommute
pairwise; hence $U_{\mathrm{fer}}^2=I$.
Let the clean oracle $O_{m_t}$ flip a flag if and only if its input is
a valid $t$-simplex. For a flag $a$, put
\[
 F_t^{(a)}:=X_aO_{m_t}
 =\Pi_t\otimes I_a+(I-\Pi_t)\otimes X_a.
\]
Using two distinct flags gives the boundary block-encoding unitary
\begin{equation}
 U_{B_r}:=F_{r-1}^{(a_2)}
 (U_{\mathrm{fer}}\otimes I_{a_1a_2})F_r^{(a_1)},\qquad
 \bra{00}U_{B_r}\ket{00}
 =\Pi_{r-1}U_{\mathrm{fer}}\Pi_r
 =\frac{\mathcal B_r}{\sqrt N}.
 \label{eq:explicit-clique-compression}
\end{equation}
Here $\mathcal B_r$ is the clique boundary padded to the full occupation
register, as in \eqref{eq:padded-boundary}. The left and right
Hamming weights select the lowering part of
$\partial+\partial^\dagger$. This is an exact $(\sqrt N,2,0)$ block encoding when clean membership oracles are treated as primitives
\cite[Appendix~C.2.2, Equations~(50)--(52) and Table~12]{McArdle2022AQubits}.
Each boundary call uses one membership call in degree $r$ and one in
degree $r-1$, with controlled versions when required. The two signal
flags are distinct from reusable membership workspace, which is returned
to zero and whose gate and qubit costs must be included. If that workspace
is retained as part of the physical unitary, include the projector onto its all-zero state in the block-encoding signal projector; the same
compression identity holds.

\subsection{Allocation of error across boundary degrees}
\label{app:error-allocation}

For \eqref{eq:critical-query-complexity}, ignoring logarithmic factors, put
$w_r=D_r\Lambda_r\sqrt{\log N}/\gamma_r$ and
$\varepsilon_{\rm spec}=\varepsilon-\varepsilon_{\rm arith}>0$.
The allocation minimizing $\sum_rw_r/\varepsilon_r$ subject to
$0<\varepsilon_r\leq D_r$ and
$\sum_r\varepsilon_r\leq\varepsilon_{\rm spec}$ is
\begin{equation}
 \varepsilon_r=\min\{D_r,t\sqrt{w_r}\},\qquad
 \sum_r\varepsilon_r=\min\!\left\{\varepsilon_{\rm spec},\sum_rD_r\right\},
 \label{eq:optimal-error-allocation}
\end{equation}
where $t>0$ satisfies the sum constraint. The Karush--Kuhn--Tucker conditions give this formula. If no cap is active, it reduces to
$\varepsilon_r=\varepsilon_{\rm spec}\sqrt{w_r}/\sum_s\sqrt{w_s}$; with
$\nu_r=\nu/(i+1)$ the resulting bound, with logarithmic factors suppressed as above, is
\begin{equation}
 \widetilde O\!\left(
 \frac{(\sum_r\sqrt{w_r})^2}{\varepsilon_{\rm spec}}
 \log\frac{2(i+1)}\nu\right).
 \label{eq:optimal-leading-cost}
\end{equation}
When a cap is active, substitute \eqref{eq:optimal-error-allocation}
into \eqref{eq:critical-query-complexity}. Adjacency-query optimization
multiplies $w_r$ by the membership-query cost in degree $r$.

If sharper bounds $S_{\ast,r}$ are supplied, the general two-term
bound may instead be optimized. With these bounds fixed independently
of the error allocation, define
\[
 u_r:=\frac{D_r\Lambda_r\sqrt{S_{\ast,r}}}{\gamma_r},\qquad
 v_r:=\frac{\Lambda_r\sqrt{D_r}}{\gamma_r}.
\]
After suppressing logarithmic and confidence factors, the objective to minimize is
$\sum_r(u_r/\varepsilon_r+v_r/\sqrt{\varepsilon_r})$.
If $\varepsilon_{\rm spec}\geq\sum_rD_r$, all errors are set to
$D_r$. Otherwise, for a multiplier $\lambda>0$, let
$z_r(\lambda)>0$ be the unique solution of
\[
 \frac{u_r}{z_r^2}+\frac{v_r}{2z_r^{3/2}}=\lambda.
\]
In the latter case, the minimizer subject to the same individual and
total-error bounds is
$\varepsilon_r=\min\{D_r,z_r(\lambda)\}$, with the same sum
constraint. Multiplying both $u_r$ and $v_r$ by the membership cost
gives the adjacency-query objective. These rules optimize the displayed
upper bounds after suppressing logarithmic and confidence factors;
they do not assert algorithmic optimality.

\subsection{Resource comparison for the torsion family}
\label{app:join-resource-proof}

We prove \cref{cor:join-resource-comparison}.  We first make the fixed
gap constant in \eqref{eq:cY-definition} fully explicit.  The augmented
chain dimensions of $Y$ are $1,31,90,60$.  Since $Y$ is rationally
acyclic, each augmented full Hodge Laplacian $\Delta_q^Y$ is a positive
definite integer matrix and hence has determinant at least one.  Moreover,
$\|\Delta_q^Y\|\leq31$: the clique-boundary normalization bounds its upper
and lower summands by $31$, and their images are orthogonal.  If $n_q$ is
the dimension of $\Delta_q^Y$, the product of its eigenvalues is at least
one while every eigenvalue is at most $31$, so
\[
 \lambda_{\min}(\Delta_q^Y)\geq31^{-(n_q-1)}\geq31^{-89}.
\]
Consequently,
\begin{equation}
 c_Y\geq31^{-89}.
 \label{eq:cY-explicit-bound}
\end{equation}
This supplies the gap promise without relying on a floating-point
spectral calculation.

Now take $m=2^k$, $k\geq4$, put $N=N_Z$, and recall
$L_{m,k}=(m-1)^k\log2$.  In boundary degree $r$, the normalized
parameters are
\[
 \alpha_r=\sqrt N,
 \qquad
 \gamma_r=\sqrt{c_Y/N},
 \qquad
 S_{\ast,r}=\log N,
 \qquad
 s_r=O(\log N),
\]
where constants may depend on the fixed complex $Y$.  With the precision
choice in \eqref{eq:torsion-join-precision} and failure probability
$\nu/d$ in each degree, \cref{thm:quantum-pseudologdet} gives
\[
 Q_{B,r}=\widetilde O\!\left(
 \sqrt{\frac N{c_Y}}\,
 \varepsilon_{\mathrm n,r}^{-1}
 \log\frac{2d}{\nu}\right).
\]
Here the second amplitude-estimation term is absorbed into the first up
to an absolute factor because $\varepsilon_{\mathrm n,r}\leq1$ and
$N\geq2$.  The chosen precision satisfies
\[
 \varepsilon_{\mathrm n,r}^{-1}
 =\max\!\left\{1,\frac{dD_r}{2\delta L_{m,k}}\right\}.
\]
Using \eqref{eq:torsion-join-candidate-ratio},
\[
 \sum_{r=1}^{d}\varepsilon_{\mathrm n,r}^{-1}
 \leq d+\frac{d}{2\delta L_{m,k}}\sum_{r=1}^{d}D_r
 \leq d+\frac{d^2R_{m,k}}{2\delta}.
\]
Summing the degreewise costs proves \eqref{eq:join-resource-queries};
the union bound gives total failure probability at most $\nu$.  A
boundary call uses $O(d^2)$ adjacency queries, proving the stated
adjacency-query conversion.

For the gate count, label each multipartite vertex by its part and its
within-part index and use a fixed table for $\Gamma_T$.  Adjacency is then
reversible with $\operatorname{polylog}N$ gates.  Since $d=O(\log N)$,
a reversible scan of the $N$-qubit occupation register can extract the at
most $d+1$ selected labels, test their pairwise adjacencies, and uncompute
its workspace using $\widetilde O(N)$ gates.  Dicke-state preparation uses
$O(N(d+1))$ elementary rotation gates
\cite{BartschiEidenbenz2019Dicke}; the fermionic boundary circuit,
preparation inverses, controls, and zero-state reflections also cost
$\widetilde O(N)$ per primitive.  Substitution into
\eqref{eq:normalized-gates} therefore gives $\widetilde O(NQ_B)$ quantum
gates.  The required rotation-synthesis precision contributes only
logarithmic factors because the polynomial degree and inverse error
tolerances are polynomial in $N$, $\delta^{-1}$, and $\nu^{-1}$.
Classical computation of the polynomial coefficients and QSVT phase
sequences is not included.

Finally, let
$F_Z(z)=\sum_{r=-1}^{d}f_r(Z)z^{r+1}$, with $f_{-1}=1$.  The face
numbers of $Y$ are $(31,90,60)$, and joins multiply face polynomials, so
\[
 F_Z(z)=(1+mz)^k(1+31z+90z^2+60z^3).
\]
The leading coefficient gives $f_d(Z_{m,k})=60m^k$.  Every column of
$B_d$ has exactly $d+1=k+3$ nonzero entries, proving
\eqref{eq:join-explicit-boundary-size}.  Since $m^k=2^{k^2}$ and
$\log N=\Theta(k)$, this size is
$\exp(\Theta((\log N)^2))=N^{\Theta(\log N)}$.  This is an output-size
lower bound only for methods that explicitly list the matrix; it does
not constrain classical algorithms using the succinct description or
the known closed form for $L_{m,k}$.

\section{Indexing and convention check for the critical-group formula}
\label{app:indexing-correction}

Corollary~4.4 of~\cite{dmtcs:2909} prints an alternating product with upper
limit $i$ and exponent $i-j$.  With the definition of $\pi_j$ used there, the
critical group $K_i$ is instead enumerated by $(i+1)$-dimensional trees.
Combining that critical-group/tree identity with the recurrence
\eqref{eq:matrix-tree-recurrence} gives an upper boundary index $i+1$ and sign
$(-1)^{i+1-r}$, as in \eqref{eq:critical-pdet-corrected}.

The same recurrence retains the lower-homology contribution
\begin{equation}
  2\sum_{q=0}^{i-1}(-1)^{i-1-q}\log t_q(X).
  \label{eq:appendix-lower-correction}
\end{equation}
It disappears only when the relevant integral reduced homology groups are
trivial, which is the additional hypothesis used in~\eqref{eq:critical-pdet-boxed}; compare~\cite[Corollary~2.10]{DuvalKlivansMartin2011Cellular}.  At $i=0$, the corrected
formula specializes to
\begin{equation}
  \log|K_0(G)|
  =\ell_1(G)-\ell_0(G)
  =\log\pdet(L_G)-\log|V(G)|,
  \label{eq:appendix-graph-specialization}
\end{equation}
which is Kirchhoff's graph formula.  This specialization fixes both the boundary
index and the sign convention.

\end{document}